\documentclass[11pt, reqno]{amsart}
\usepackage{amsmath}
\usepackage{amsfonts}
\usepackage{comment}
\usepackage{amssymb}
\usepackage{amsthm}
\usepackage[english]{babel}
\usepackage{latexsym}
\usepackage{blindtext}
\usepackage{bm}

\newcommand{\bbN}{{\mathbb{N}}}

\newcommand{\bbR}{{\mathbb{R}}}

\newcommand{\bbD}{{\mathbb{D}}}

\newcommand{\bbH}{{\mathbb{H}}}
\newcommand{\bbE}{{\mathbb{E}}}

\newcommand{\bbC}{{\mathbb{C}}}

\newcommand{\no}{\nonumber}

\newcommand{\eps}{\varepsilon}

\newcommand{\supp}{\text{\rm{supp}}}

\newcommand{\beq}{\begin{equation}}

\newcommand{\eeq}{\end{equation}}

\newcommand{\ba}{\begin{align}}

\newcommand{\ea}{\end{align}}

\renewcommand{\Re}{\mathop{\mathrm{Re}}}
\renewcommand{\Im}{\mathop{\mathrm{Im}}}

\DeclareMathOperator{\Tr}{Tr}
\DeclareMathOperator{\Var}{Var}

\allowdisplaybreaks

\numberwithin{equation}{section}

\usepackage{mathrsfs}
\usepackage{enumitem}
\usepackage{mathtools}
\usepackage[utf8]{inputenc}
\newtheorem{theorem}{Theorem}[section]

\newtheorem{prop}[theorem]{Proposition}

\newtheorem{lemma}[theorem]{Lemma}

\newtheorem{corollary}[theorem]{Corollary}

\theoremstyle{definition}

\newtheorem{remark}[theorem]{Remark}

\title[Mesoscopic Universality on the Unit Circle]{Mesoscopic Universality for Circular Orthogonal Polynomial Ensembles}
\author[J.~Breuer and D.~Ofner]{Jonathan Breuer$^{1,3}$ and Daniel Ofner$^{2,3}$}
\thanks{$^1$ Institute of Mathematics, The Hebrew University of Jerusalem, Jerusalem, 91904, Israel. E-mail: jbreuer@math.huji.ac.il.}
\thanks{$^2$ Institute of Mathematics, The Hebrew University of Jerusalem, Jerusalem, 91904, Israel. E-mail: daniel.ofner@mail.huji.ac.il.}
\thanks{$^3$ Research supported in part by Israel Science Foundation Grant No.\ 1378/20}
\begin{document}
\maketitle
\sloppy

\begin{abstract}
We study mesoscopic fluctuations of orthogonal polynomial ensembles on the unit circle. We show that asymptotics of such fluctuations are
stable under decaying perturbations of the recurrence coefficients, where the appropriate decay rate depends on the scale considered. By
directly proving Gaussian limits for certain constant coefficient ensembles, we obtain mesoscopic scale Gaussian limits for a large class of
orthogonal polynomial ensembles on the unit circle.

As a corollary we prove mesocopic central limit theorems (for all mesoscopic scales) for the $\beta=2$ circular Jacobi ensembles with real parameter $\delta>-1/2$. 
\end{abstract}

\section{Introduction}

Let $\mu$ be a probability measure on the unit circle, $\partial \mathbb{D}$, whose support is an infinite set. The \emph{orthogonal polynomial ensemble} (OPE)
of size $n$ associated with the underlying measure $\mu$ is a probability measure on the $n$-dimensional torus proportional to

\begin{equation}\label{eq:OPE}
   \Pi_{1\leq j<k\leq n}
\left|e^{i\theta_k}-e^{i\theta_j}\right|^2\Pi_{1\leq j \leq n}d\mu(\theta_j).
\end{equation}

Such point processes arise naturally in random matrix theory and in statistical mechanics. For example, one of the best known OPE's on the
unit circle is the Circular Unitary Ensemble (CUE) where $d\mu(\theta)=d\theta/2\pi$. The CUE describes the joint eigenvalue distribution of a
randomly chosen $n\times n$ unitary matrix with respect to the Haar measure on the unitary group $U(n)$.

Another well studied family of examples is that of the `$\beta=2$' circular Jacobi ensembles, also known as the Hua-Pickrell ensembles, where $d\mu_\delta(\theta)=\left(1-e^{-i\theta} \right)^\delta\left( 1-e^{i\theta}\right)^{\overline{\delta}}d\theta$ where $\textrm{Re}\delta>-1/2$ (see \cite[Section 3]{Forrester}). Such ensembles arise naturally from the study of the characteristic polynomial of a Haar distributed element in the $n$-dimensional unitary group and have been of considerable interest both in the context of random matrix theory and that of group representations \cite{borodin-olsh,Bourgade,Bour-Nik-Rou1,Bour-Nik-Rou2,Hua,liu,Neretin,Pickrell1,Pickrell2,witt-forr}.

The underlying measure $\mu$ admits a sequence of monic orthogonal polynomials $\left(\Phi_j\right)_{j=0}^\infty$, uniquely defined by
\begin{equation} \label{eq:monicOP}
\deg(\Phi_j)=j, \quad \int \Phi_j(z) \overline{P(z)}d\mu(z)=0
\end{equation}
for any polynomial $P$ with $\deg(P)<j$. The polynomials $\left(\Phi_j\right)_{j=0}^\infty$ are known to satisfy a two term recurrence
relation
\beq
z\Phi_n(z)=\Phi_{n+1}(z)+\overline{\alpha_n}z^n\overline{\Phi_n\left(1/\overline{z}\right)}
\eeq
where $\alpha_n\in \mathbb{D}$ (see \cite{Szego}). In fact, this relation defines a bijection between probability measures with infinite
support on $\partial \mathbb{D}$ and sequences $\left(\alpha_n\right)_{n=0}^\infty$  with $\alpha_n \in \mathbb{D}$ for all $n$ (see, e.g., \cite{simonopuc}) (the
sequence $\left(\alpha_n\right)_{n=0}^\infty$ associated with a measure is often referred to as the measure's \emph{Verblunsky} coefficients).
As in several previous works (\cite{bi-or-ens,realope,duits-kozhan,hardy,KVA,Lambert,Ledoux1,Ledoux2} and references therein), the fundamental idea behind this work is that the
asymptotic properties of a measure's Verblunsky coefficients can be used to determine the asymptotic properties of the associated OPE.

In this work we focus on \emph{mesoscopic} fluctuations of the empirical (aka normalized counting) measure of the point process, namely the
measure
$$\nu_n=\frac{1}{n}\sum_{k=1}^n\delta_{e^{i\theta_k}}$$
where the points $e^{i\theta_k}$ are picked randomly according to the distribution \eqref{eq:OPE}. The mesoscopic scales are intermediate
scales between the microscopic and macroscopic scales in the sense that they involve focusing on the point process on an arc of length $\sim
\frac{1}{n^\gamma}$ where $0<\gamma<1$. For nice enough underlying measures the number of points in such an arc is $\sim n^{1-\gamma}$.

Our main purpose is to provide a tool to control the mesoscopic fluctuations of sufficiently nice sampling functions of an OPE, using the Verblunsky coefficients associated with the corresponding measure. This will allow us to deduce a mesoscopic Central Limit Theorem (CLT) for a large class of such ensembles. Our main results say that a mesoscopic CLT holds for mesoscopic scales, as long as the Verblunsky coefficients converge fast enough (depending on the scale) to a limit. This is the unit circle analog of the work carried out in \cite{realope}. However, the translation to the unit circle, and in particular the necessary analysis of the constant coefficient case, is not as straightforward as one might hope. 

A particular consequence of our results is a mesoscopic CLT for the case that the $(\alpha_n)_{n=1}^\infty$ satisfy $\alpha_n=\mathcal{O}(n^{-1})$. The case $\alpha_n=0$ corresponds to the CUE introduced above. In that case a mesoscopic CLT (though with a slightly different formulation from ours) was established by Soshnikov in \cite{soshnikov}. Mesoscopic fluctuations have been studied in other contexts in various works, including \cite{duits-berggren, Bot-Kho1, Bot-Kho2, realope,Fyo-kho-sim, he-knowles, Lambert-Johansson, Lambert-meso, landon-sosoe, lod-simm}, but we are not aware of any other results up to now, on circular models. As a corollary of our main results we shall obtain a mesocopic CLT for the Hua-Pickrell ensembles with $\delta \in \mathbb{R}$ (see Theorem \ref{thm:Hua-Pickrell} below).

A useful tool in the study of asymptotics of point processes is that of linear statistics: given a point process with $n$ points on a space
$\Xi$ and a (nice) function $f:\Xi \rightarrow \mathbb{C}$, the associated linear statistic is the random variable
\begin{equation} \nonumber
X_{f,n}=\sum f(x_i)
\end{equation}
where the sum is over the $n$ random points of the point process. In our case, in order to study mesoscopic fluctuations of $\nu_n$ we want to
also let $f$ depend on $n$ in such a way that it is effectively supported in an arc of size $n^{-\gamma}$ around a certain fixed $\theta_0$.
The function
$$
f_z(\theta)=\textrm{Re} \frac{e^{i\theta}+z}{e^{i\theta}-z}
$$
(i.e.\ the Poisson kernel) is a convenient candidate for such a manipulation and will be the starting point for our analysis. Thus, for a
fixed $\gamma \in (0,1)$ and $\eta$ with $\textrm{Re}\eta>0$ we define
\begin{equation} \label{eq:PoissonScaled}
\Psi_{n,\gamma,\eta} \left(\theta \right)=\frac{1}{n^\gamma}\text{Re}\left(\frac{e^{i\theta}+\left(1-\frac{\eta}{n^\gamma}\right)}
{e^{i\theta}-\left(1-\frac{\eta}{n^\gamma}\right)}\right),
\end{equation}
and for a fixed $\theta_0$ we shall consider the linear statistic
\begin{equation} \label{eq:PoissonLinStat}
X_{\Psi_{n,\gamma,\eta}}^{\theta_0}=\sum_{k=1}^n\Psi_{n,\gamma,\eta} \left(\theta_k-\theta_0 \right).
\end{equation}

We shall prove
\begin{theorem}\label{thm:CLTPoissonIntro}
Let $\mu$ be a probability measure on the unit circle with Verblunsky coefficients $\left(\alpha_n\right)_{n=0}^\infty$ such that
$$\alpha_n=\alpha+\mathcal O(n^{-\beta})$$
as $n\to \infty$, for some $-1<\alpha<1$, $0<\beta<1$. Then for any $\theta_0\in \supp (\mu)^\circ $  and $0<\gamma<\beta$ we have that
\begin{equation}\label{eq:CLTPoisson}
X_{\Psi_{n,\gamma,\eta}}^{\theta_0}-\bbE X_{\Psi_{n,\gamma,\eta}}^{\theta_0}\xrightarrow[n\to
\infty]{\mathcal{D}}\mathcal{N}\left(0,\sigma^2\right),
\end{equation}
as $n\rightarrow \infty$. Where $\xrightarrow[n\to \infty]{\mathcal{D}}$ denotes convergence in distribution and $\mathcal{N}(0,\sigma^2)$ is the normal
distribution with variance
\begin{equation}\label{eq:VariancePoisson}
\sigma^2=\sigma(\eta)^2=\frac{2}{\eta+\overline{\eta}}.
\end{equation}
\end{theorem}

The convergence described in Theorem \ref{thm:CLTPoissonIntro} can be extended to a more general class of functions. However, generally, rescaling
functions on the circle is not straightforward.  Indeed, a standard way of rescaling a linear statistic $X_f$ in the case of a process on $\mathbb{R}$, for $f$ with appropriate decay at $\pm \infty$, is to replace $f$ with $f(n^\gamma\cdot)$. However, this approach does not naturally translate to the setting of the unit circle, since functions on $\partial \mathbb{D}$ naturally correspond to \emph{periodic} functions on $\mathbb{R}$. Thus, the rescaling we use for $\Psi$ above is not of the form $\Psi(n^\gamma\cdot)$. 

One possible solution to this issue is to simply identify $\partial \mathbb{D}$ with the interval $[0,2\pi) \subseteq \mathbb{R}$ and to use the `standard' scaling. However, while this approach may be simpler in terms of notation, we choose the more `geometrical' approach involving mapping the unit circle to the line using a M\"obius transformation and scaling the functions on $\mathbb{R}$ before pulling them back to $\partial \mathbb{D}$. As we shall show in Section 5 below, under this approach, the scaling used above for $\Psi$ does indeed correspond to the standard scaling under this procedure. Moreover, this approach will also allow us to exploit the density arguments of \cite[Section 5]{realope} to extend the result of Theorem \ref{thm:CLTPoissonIntro} to functions with support contained in a subarc of $\partial \mathbb{D}$. 

Thus, let $f:\partial\mathbb{D}\to \mathbb {R}$ with $\text{supp}f\subset \left \{e^{i\theta} \mid \theta \in (-\pi,\pi)\right \}$ (we abuse
notation and write below  $\supp f \subseteq (-\pi,\pi)$ for such $f$). Let $\mathcal{M}$ be the \textit{M\" obius} transformation
\begin{equation} \label{eq:Mobius}
\begin{split} & \mathcal{M}:\mathbb{H}_+\to \overline{\mathbb{D}} \\
& \mathcal{M}(z)=\frac{i-z}{i+z}
\end{split}
\end{equation}
where $\mathbb{H}_+=\left\{z\in \mathbb{C}\:|\:Im(Z)\geq0\right\}$. Note that $\mathcal{M}(\overline{\mathbb{R}})=\partial\mathbb{D}$ with $\mathcal{M}(\pm \infty)=-1$ and $\mathcal{M}^{-1}(e^{i\theta})=\tan(\theta/2)$.  Fixing $\theta_0$ the associated shifted and rescaled linear statistic is
\beq \label{eq:RealRescaling}
\begin{split}
X_{\widetilde{f},n,\gamma}^{\theta_0}&=\sum_{k=1}^n f\left(\mathcal{M}\left(n^\gamma\left(\mathcal{M}^{-1}\left(e^{i\left(\theta_k-\theta_0\right)}\right)\right)\right)\right) \\
&=\sum_{k=1}^n f\left(\mathcal{M}\left(n^\gamma \tan \left(\frac{\theta_k-\theta_0}{2}\right) \right) \right)
\end{split}
\eeq
Thus, $X_{\widetilde{f},n,\gamma}^{\theta_0}$ takes into consideration only those points which are within a distance $\sim\frac{1}{n^\gamma}$ around $e^{i\theta_0}$, meaning that we consider the mesoscopic scale.

Theorem \ref{thm:CLTPoissonIntro} extends to
\begin{theorem}\label{thm:CLTgeneral}
Let $\mu$ be a probability measure on the unit circle with Verblunsky coefficients $\left(\alpha_n\right)_{n=0}^\infty$ such that
$$\alpha_n=\alpha+\mathcal O(n^{-\beta})$$
as $n\to \infty$, for some $-1<\alpha<1$, $0<\beta<1$. Then for any continuously differentiable $f$ whose support is a compact subset of $\partial \mathbb{D}\setminus \{-1\}$, and any $\theta_0\in \supp (\mu)^\circ $ and $0<\gamma<\beta$ we
have that
\begin{equation}\label{eq:CLTgeneral}
X_{\widetilde{f},n,\gamma}^{\theta_0}-\bbE X_{\widetilde{f},n,\gamma}^{\theta_0} \to \mathcal{N}(0, \sigma_f^2),
\end{equation}
as $n\rightarrow \infty$ in distribution, where $\mathcal{N}(0,\sigma_f^2)$ is the normal distribution with variance
\begin{equation}\label{eq:CLTgeneralvariance}
\sigma_f^2= \frac{1}{ 4\pi^2} \iint \left( \frac{f\circ \mathcal{M}(x)-f\circ \mathcal{M}(y)}{x-y} \right)^2 dxdy.
\end{equation}
\end{theorem}

\begin{remark} 
With a slight abuse of notation we shall denote by $C_c^1(-\pi,\pi)$ the class of continuously differentiable functions supported on a compact subset of $\partial \mathbb{D}\setminus \{-1\}$.
\end{remark}

Theorems \ref{thm:CLTPoissonIntro} and \ref{thm:CLTgeneral} are our two main results, and can be seen as CLT's, though
without the classical normalization of $\sqrt{n}$. The absence of this normalization is a well known phenomenon, which arises from the strong repulsion between points of
an OPE leading to boundedness of the variance of linear statistics for differentiable $f$. The reason for starting with the fluctuations
associated with $\Psi_{n,\gamma,\eta}$ is the fact that $\Psi_{n,\gamma,\eta}$ is an approximation of identity, hence has all its mass
concentrated in an arc around $0$, of size proportional to $\frac{1}{n^\gamma}$. We will use Theorem \ref{thm:CLTPoissonIntro} to obtain
Theorem \ref{thm:CLTgeneral} via an approximation argument.

Some explicit examples for measures, $\mu$, for which our results immediately imply a mesoscopic CLT for any $\gamma\in(0,1)$ are (the names we use are from \cite[Section 1.6]{simonopuc}):
\begin{itemize}
    \item The degree one Bernstein-Szeg\H o measure, where $d\mu(\theta)= \frac{1-r}{\left|1-re^{i\theta}\right|}\frac{d\theta}{2\pi}$ for
        $r\in(0,1)$ and hence $\alpha_n=0$ for $n>1$.
    \item The `single inserted mass point measure' where\\ $d\mu(\theta)=(1-h)\frac{d\theta}{2\pi}+h\delta(\theta-\theta_0)$ for $h\in(0,1)$
        and therefore $\alpha_n=\mathcal{O}\left(n^{-1}\right)$
    \item The `single nontrivial moment measure' where $d\mu(\theta)=\left(1-\cos(\theta)\right)\frac{d\theta}{2\pi}$ and therefore
        $\alpha_n=-\frac{1}{n+2}.$
\end{itemize}

We note the following corollary of Theorem \ref{thm:CLTgeneral}.
\begin{theorem} \label{thm:Hua-Pickrell}
Let $\delta \in \mathbb{R}$ with $\delta>-1/2$ and consider the corresponding Hua-Pickrell ensemble, namely the orthogonal polynomial ensemble
\beq \no 
c_{\delta,n} \prod_{1\leq j<k\leq n}\left|e^{i\theta_k}-e^{i\theta_j} \right|^2\prod_{1\leq j \leq n}\left(1-e^{-i\theta_j} \right)^\delta\left( 1-e^{i\theta_j}\right)^{\delta}d\theta_j.
\eeq
Then for any continuously differentiable $f$ whose support is a compact subset of $\partial \mathbb{D}\setminus \{-1\}$, and any $\theta_0\in [0,2\pi) $ and $0<\gamma<1$ we
have that
\begin{equation}\label{eq:CLTHuaPickrell}
X_{\widetilde{f},n,\gamma}^{\theta_0}-\bbE X_{\widetilde{f},n,\gamma}^{\theta_0} \to \mathcal{N}(0, \sigma_f^2),
\end{equation}
as $n\rightarrow \infty$ in distribution, where $\sigma_f^2$ is as in \eqref{eq:CLTgeneralvariance} above.
\end{theorem}

The proof of this theorem is given towards the end of the Introduction.

Note that $\sigma_f^2$ is independent of $\theta_0$, $\alpha$ and $\gamma$, which means that for a large class of OPE's, a mesoscopic CLT
holds (for multiple scales) with the same parameters. Therefore this result can be thought of also as a universality result.

In fact, Theorems \ref{thm:CLTPoissonIntro} and \ref{thm:CLTgeneral} follow from the following universality results, which say that if two measures
have coefficient sequences whose difference approaches zeros (at an appropriate rate) then the mesoscopic fluctuations (on an appropriate scale) for the corresponding
OPE's are asymptotically equal. This is the unit circle analog of Theorem 1.2 in \cite{realope}, which deals with mesoscopic universality for
real line OPE's. In the proof we use the approach and tools developed in \cite{realope}.

\begin{theorem} \label{thm:UniversalityPoissonIntro}
Let $\mu$ and $\mu_0$ be two probability measures on $\partial \mathbb{D}$ and denote by $\{\alpha_n\}_{n=0}^\infty$ and
$\{\alpha_n^0\}_{n=0}^\infty$ the respective associated recurrence coefficients. Let $\theta_0 \in [0,2\pi)$ be such that there exists a
neighborhood $\theta_0 \in I$ on which the following two conditions are satisfied: \\
(i) $\mu_0$ restricted to $I$ is absolutely continuous with respect to Lebesgue measure and its Radon-Nikodym derivative is bounded there.  \\
(ii) The orthonormal polynomials for $\mu_0$ are uniformly bounded on $I$.\\
Assume further that
\beq 
\alpha_n-\alpha_n^0=\mathcal O(n^{-\beta})
\eeq
as $n \to \infty$ for some $0<\beta<1$.

Then for any $0<\gamma<\beta$ we have
\beq \label{eq:conclusion}
\left|\bbE \left(X_{\Psi_{n,\gamma,\eta}}^{\theta_0}-\bbE X_{\Psi_{n,\gamma,\eta}}^{\theta_0} \right)^m -\bbE_0
\left(X_{\Psi_{n,\gamma,\eta}}^{\theta_0}-\bbE_0 X_{\Psi_{n,\gamma,\eta}}^{\theta_0}\right)^m\right| \rightarrow 0
\eeq
as $n \rightarrow \infty$, where $\bbE$ and $\bbE_0$ denote the expectation with respect the orthogonal polynomial ensemble corresponding to
$\mu$ and $\mu_0$ respectively.
\end{theorem}

Similarly,

\begin{theorem} \label{thm:main-result}
For $\mu$, $\mu_0$, $\theta_0$, $\beta$ and $\gamma$ satisfying the conditions of Theorem \ref{thm:UniversalityPoissonIntro} and any $f\in
C_c^1((-\pi,\pi))$ and any $0<\gamma<\beta$ we have
\beq \label{eq:conclusion}
\left|\bbE \left(X_{\widetilde{f},n,\gamma}^{\theta_0}-\bbE X_{\widetilde{f},n,\gamma}^{\theta_0} \right)^m -\bbE_0
\left(X_{\widetilde{f},n,\gamma}^{\theta_0}-\bbE_0 X_{\widetilde{f},n,\gamma}^{\theta_0}\right)^m\right| \rightarrow 0
\eeq
as $n \rightarrow \infty$.
\end{theorem}

In light of Theorems \ref{thm:UniversalityPoissonIntro} and \ref{thm:main-result}, in order to prove Theorems \ref{thm:CLTPoissonIntro} and
\ref{thm:CLTgeneral} we only need to consider the case where $\alpha_n=\alpha\in(-1,1)$. We do this via the recurrence coefficients as well.
While the analysis of the $\alpha=0$ case (which is in fact already known for a slightly different family of functions by \cite{soshnikov}) is a rather straightforward computation using block Toeplitz matrices, the $\alpha \neq 0$ case turns out to be very different and technically rather challenging. The reason is that while
the CMV representation (see below) gives rise to block Toeplitz matrices with diagonal blocks in the $\alpha=0$ case, this is no longer true
for $\alpha \neq 0$. A similar problem was encountered in \cite{duits-kozhan} where the macroscopic fluctuations were treated. There, a
certain ``unfolding'' of the bases was used to solve the problem. Here we bypass the problem by using the GGT representation (described below
as well). This, however, is only the first step and some nontrivial computations are needed to verify the Gaussian limit. The details are
described in Section 3.

\smallskip

As noted above, the use of recurrence coefficients for the study of OPE's has recently shown itself to be quite fruitful \cite{hardy}. One
example which is closely related to our subject, is a corollary of  \cite{Golinskii} which states that if $\alpha_n\xrightarrow[n\to
\infty]{}\alpha \in \partial\mathbb{D}$ then $\nu_n$ converges weakly with probability 1 to $\nu_\alpha$ where
$$d\nu_\alpha(\theta)/d\theta
=\frac{1}{2\pi}\frac{\sin\left(\frac{\theta}{2}\right)}{\sqrt{\sin\left(\frac{\theta}{2}\right)^2-|\alpha|^2}}.$$
On the real line, a theorem by Kuijlaars and Van Assche \cite{KVA} provides a detailed formula relating the density of states (i.e.\ the weak
limit of the empirical measure) with the asymptotics of the (real) recurrence coefficients. The fact that the fluctuations of the empirical
measure can be analyzed using asymptotics of the recurrence coefficients was realized in \cite{bi-or-ens}, where a macroscopic CLT was proved for a
large class of real line OPE's. The unit circle analog of this work was carried out in \cite{duits-kozhan} mentioned above, where a
macroscopic CLT was shown to hold as long as $\alpha_n$ converges. As remarked above, our work has been inspired by \cite{realope} where the
mesoscopic scale was studied for real line OPE's. For other results relating asymptotics of OPE's with the associated recurrence coefficients
see \cite{hardy} and references therein.

Before describing the structure of the paper we provide the proof of Theorem \ref{thm:Hua-Pickrell}
\begin{proof}[Proof of Theorem \ref{thm:Hua-Pickrell}]
Fix $\delta \in \mathbb{R}$ with $\delta>-1/2$ and let 
\beq \no
d\mu_\delta(\theta)=\left(1-e^{-i\theta} \right)^\delta\left( 1-e^{i\theta}\right)^{\overline{\delta}}d\theta=(2-2\cos(\theta))^\delta d\theta
\eeq
Let $\left(\alpha_n \right)_{n=1}^\infty$ be the Verblunsky coefficients associated with $\mu_\delta$. Note that since $\mu$ is invariant under complex conjugation, the $\alpha_n$'s are all real. In light of Theorem \ref{thm:CLTgeneral}, it suffices to show that $\alpha_n=\mathcal{O}(n^{-1})$.

We shall use the Szeg\H o mapping (\cite[Chapter 13]{simonopuc}, induced by the map $\theta \mapsto 2\cos(\theta)$ and denoted by $Sz$, to send $\mu_\delta$ to a measure on $[-2,2]$ and relate $(\alpha_n)$ with the recurrence coefficients of $Sz(\mu_\delta)$. Explicitly, $Sz(\mu_\delta)(dx)=(2-x)^{\delta-1/2}(2+x)^{-1/2}dx$ so that $Sz (\mu_\delta)$ is a Jacobi measure on $[-2,2]$. For such measures it is known \cite{KMVV} that the recurrence coefficients satisfy
\beq \no
a_n=1+\mathcal{O}(n^{-2}), \qquad b_n=\mathcal{O}(n^{-2})
\eeq  
and by \cite[Theorem 13.1.10]{simonopuc}, the relation to $(\alpha_n)$ is through 
\beq \no
\alpha_{2k}=\frac{u_k^+-u_k^-}{u_k^++u_k^-}, \qquad \alpha_{2k-1}=1-\frac{1}{2}\left(u_k^++u_k^- \right)
\eeq
where $u_k^{\pm}$ is defined by 
\beq \label{eq:uk}
u_k^{\pm}=2\pm b_{k+1}-a_k^2\left(u_{k-1}^{\pm} \right)^{-1}, \qquad \left( u_{-1}^{\pm}\right)^{-1}=0.
\eeq

Now, by Rakhmanov's Theorem (see, e.g., \cite[Corollary 9.1.11]{simonopuc}), since $\mu_\delta$ is absolutely continuous with support all of $\partial\mathbb{D}$, $\alpha_n \rightarrow 0$ as $n \rightarrow \infty$. This immediately implies that $u_n^{\pm}\rightarrow 1$.  It thus suffices to show that $u_n^{\pm}=1+\mathcal{O}(n^{-1})$. To this end, write $u_n^{\pm}=1+\xi_n^{\pm}$. Then \eqref{eq:uk} together with the asymptotic estimates for $a_n, b_n$ imply that
\beq \no
\xi_{n+1}^{\pm}-\xi_n^{\pm}=-\frac{\left(\xi_n^{\pm}\right)^2}{1+\xi_n^{\pm}}+f_n^{\pm}
\eeq
where $f_n^{\pm}$ is some sequence satisfying $|f_n^{\pm}| \leq \frac{C_0}{n^2}$ for some constant $C_0>0$. Since it is irrelevant for everything that comes next, we omit the superscript from now on. 

We consider $n$ sufficiently large so that $1/2 \leq 1+\xi_n \leq 2$ and also $C_0/n<1$. It follows from the fact that $\xi_n \rightarrow 0$, that if $\xi_n$ is negative for some $n$, it must satisfy $\xi_n \geq -\frac{\sqrt{2C_0}+1}{n}$. Indeed,  otherwise
\beq \no 
\xi_{n+1}-\xi_n \leq -\frac{1}{n^2}<0
\eeq
so $\xi_{n+1}<\xi_n$ and, by iterating, $\xi_{n+k}<-\frac{\sqrt{2C_0}+1}{n}$ which is a contradiction. 

If $\xi_{n_0}>\frac{2\sqrt{C_0}+1}{n}$ for some $n_0$ then 
\beq \no 
\xi_{n_0+1}-\xi_{n_0} \leq -\frac{\xi_{n_0}^2}{4}
\eeq
so $\xi_{n_0+1}<\xi_{n_0}$. Thus, if $\xi_{n+1}>0$ we get $\xi_n^2>\xi_n\xi_{n+1}$ and therefore
\beq \no
\xi_{n+1}^{-1}-\xi_n^{-1} \geq \frac{1}{4}.
\eeq
As long as $\xi_{n_0+k}>\frac{2\sqrt{C_0}+1}{n}$ this can be iterated to obtain $\xi_{n_0+k}\leq \frac{4\xi_{n_0}}{k\xi_{n_0}+4}$.

Finally, if $\xi_n \leq \frac{2\sqrt{C_0}+1}{n}$ then
\beq \no
\xi_{n+1}-\xi_n \leq \frac{C_0}{n^2}<\frac{1}{n}
\eeq
so $\xi_{n+1} \leq \frac{2\sqrt{C_0}+2}{n}\leq \frac{2\sqrt{C_0}+3}{n+1}$
for $n$ sufficiently large. But if $\frac{2\sqrt{C_0}+1}{n+1}<\xi_{n+1} \leq \frac{2\sqrt{C_0}+3}{n+1}$ then by the previous paragraph, $\xi_{n+2}<\xi_{n+1}$ and this will repeat until $\xi_{n+k}\leq \frac{2\sqrt{C_0}+1}{n+k}$. To conclude, either the sequence $n\xi_n$ achieves a maximum at some point $n_0$ or it is bounded. Either way, we obtain $\xi_n=\mathcal{O}(n^{-1})$ and we are done.
\end{proof}

\smallskip

The rest of this paper is structured as follows. The connection between recurrence coefficients and mesoscopic fluctuations is explained in
detail in Section 2 using the spectral theorem. In particular, we introduce there the GGT and CMV matrices which are the matrix
representations of multiplication by $z$ in $L_2(\mu)$ with respect to two different bases connected to the orthogonal polynomials. We then
show that the distribution of the mesoscopic fluctuations can be analyzed through traces of functions of the GGT/CMV operator.

In Section 3, we prove Theorem \ref{thm:CLTPoissonIntro} in the special case $\alpha_n\equiv \alpha$ with $\alpha \in (-1,1)$, and
Section 4 contains the proof of Theorem \ref{thm:UniversalityPoissonIntro}, completing the proof of Theorem \ref{thm:CLTPoissonIntro}.

Note that Theorems \ref{thm:CLTPoissonIntro} and \ref{thm:UniversalityPoissonIntro} are not, strictly speaking, special cases of Theorems
\ref{thm:CLTgeneral} and (respectively) \ref{thm:main-result}. The advantage of starting with proving these theorems comes from the fact that
the Poisson kernel has several properties that are particularly useful for our analysis.

Section 5 completes the proofs of Theorems \ref{thm:CLTgeneral} and \ref{thm:main-result} by combining the results of Sections 3 and 4 with an
approximation argument adapted from \cite{realope} to our needs.

{\bf Acknowledgments} We thank Maurice Duits and Ofer Zeitouni for useful discussions.


\section{Preliminaries}
Throughout this section, $\mu$ denotes a fixed probability measure on $\partial \mathbb{D}$ with infinite support.

\subsection{The GGT and CMV matrices}
In this subsection we introduce two matrix representations of the operator of multiplication by $z$ in
$L_2\left(\partial\mathbb{D},\mu\right)$. We shall see in the following subsection how the study of fluctuations translates to the analysis of
these particular matrices.

We denote by $\varphi_j$ the $j$'th ortho\emph{normal} polynomial associated with $\mu$. Namely, $\varphi_j=\Phi_j/\|\Phi_j\|$ where $\Phi_j$
is the monic orthogonal polynomial defined in \eqref{eq:monicOP}.
of monic polynomials $\left\{\Phi_j\right\}_{j=0}^\infty$ for which

As $\left\|\Phi_j\right\|_2=\Pi_{k=0}^{j-1}\left(1-|\alpha_k|^2\right)^{1/2}$ (see \cite[Theorem 2.1]{opucononefoot}) the two term recurrence
relation becomes $$\rho_n\varphi_{n+1}(z)=z\varphi_n(z)+\overline{\alpha_n}\varphi^*_n(z)$$ where for a polynomial $p(z)$ of degree $n$ we
denote $$p^*(z)=z^n\overline{p\left(1/\overline{z}\right)}$$
(known as the Szeg\H o dual), and
$$\rho_n=\sqrt{1-\left|\alpha_n\right|^2}.$$
Let $$V=\overline{\textrm{span}\{\varphi_j \mid j \geq 0 \}}.$$ It follows from Szeg\H o's Limit Theorem (see for example
\cite[Theorem 1.5.7]{simonopuc}) that $V=L_2(\mu)$ iff $\sum_{n=0}^\infty\left|\alpha_n\right|^2=\infty$.

Consider the operator $$f\mapsto zf$$ on $f \in L_2 (\mu)$. In the case that $V=L_2(\mu)$, it has a simple matrix representation with respect
to the orthonormal set $\left \{ \varphi_j \right\}_{j=0}^\infty$, known as the GGT matrix (see \cite{simonopuc}):
\begin{equation} \label{eq:GGT}
G=\begin{pmatrix} \overline{\alpha_0} & \overline{\alpha_1}\rho_0 & \overline{\alpha_2}\rho_0\rho_1 & \overline{\alpha_3}\rho_0\rho_1\rho_2 &
\ldots \\
\rho_0 & -\overline{\alpha_1}\alpha_0 &  -\overline{\alpha_2}\alpha_0\rho_1 & -\overline{\alpha_3}\alpha_0\rho_1\rho_2 & \ldots \\
0 & \rho_1 & -\overline{\alpha_2}\alpha_1 & -\overline{\alpha_3}\alpha_1\rho_2  & \ldots \\
0 & 0 & \rho_2 & -\overline{\alpha_3}\alpha_2 & \ldots\\
\ldots & \ldots & \ldots & \ldots & \ddots
\end{pmatrix}.
\end{equation}

In the case that $V \neq L_2(\mu)$, it is useful to consider a different representation for the operator $f \rightarrow zf$. By applying the
Gram-Schmidt procedure to the set $\left\{1,z,z^{-1},z^2,z^{-2},\ldots\right\}$ (which is always dense in $L_2$), we obtain the complete
orthonormal set $\left\{\chi_j\right\}_{j=0}^\infty$ (known as the \textit{CMV} basis). The matrix representation of multiplication by $z$
with respect to $\left(\chi_j\right)_{j=0}^\infty$ is known as the \textit{CMV} matrix (see \cite{simonopuc}):
\begin{equation} \label{eq:CMV}
\mathcal{C}=
\begin{pmatrix}
\overline{\alpha_0} & \overline{\alpha_1} \rho_0& \rho_0 \rho_1&0&0&\ldots\\
\rho_0 & -\overline{\alpha_1} \alpha_0& -\rho_1 \alpha_0&0&0&\ldots\\
0& \overline{\alpha_2}\rho_1&-\overline{\alpha_2}\alpha_1& \overline{\alpha_3}\rho_2 & \rho_3\rho_2 & \ldots\\
0&\rho_2\rho_1 & -\rho_2\alpha_1 & -\overline{\alpha_3}\alpha_2 & -\rho_3\alpha_2 & \ldots\\
0&0&0&\overline{\alpha_4}\rho_3&-\overline{\alpha_4}\alpha_3&\ldots\\
\vdots&\vdots&\vdots&\vdots&\vdots&\ddots
\end{pmatrix}.
\end{equation}

\subsection{OPE as a determinantal point process}

We start with the standard observation (see e.g.\ \cite{koenig}) that
\begin{equation} \nonumber
\Pi_{1\leq j<k\leq n}
\left|e^{i\theta_k}-e^{i\theta_j}\right|^2 \propto \det\left(K_n(\theta_j,\theta_k)\right)_{1\leq j,k \leq n}
\end{equation}
where $K_n$ is the \textit{Christoffel-Darboux} kernel associated with $\mu$
$$K_n(\theta,\phi)=\sum_{j=0}^{n-1}\varphi_j\left(e^{i\theta}\right)\overline{\varphi_j\left(e^{i\phi}\right)}.$$
Note that $K_n$ is also the integral kernel of the orthogonal projection
from $L_2\left(\mu\right)$ onto $\text{span}\left\{1,z,\ldots,z^{n-1}\right\}$.
It follows that the OPE is a determinantal point process with kernel $K_n$ \cite{borodin}.

The kernel $K_n$ is useful when utilizing the GGT representation of the previous section. However, it is not ideal when we need the CMV
representation. For this note that
\begin{equation}\label{eq:knpn}
\begin{split}
    &\Pi_{1\leq j<k\leq n} \left|e^{i\theta_k}-e^{i\theta_j}\right|^2 \\
    & \quad = \det\left(\exp\left(i\left(j-1-\left\lfloor\frac{n-1}{2}\right\rfloor\right)\theta_k\right)\right)_{1\leq j,k \leq n} \\
    & \qquad \qquad \times \det\left(\exp\left(-i\left(j-1-\left\lfloor\frac{n-1}{2}\right\rfloor\right)\theta_k\right)\right)_{1\leq j,k \leq
    n}\\
    & \quad \propto \det(\widetilde{K}_n(\theta_j,\theta_k))_{1\leq k,j \leq n}
\end{split}
\end{equation}
where
$$\widetilde{K}_n(\theta,\varphi)=\sum_{j=0}^{n-1}\chi_j\left(e^{i\theta}\right)\overline{\chi_j\left(e^{i\phi}\right)}.$$
so that $\widetilde{K}_n$ is also a kernel for the OPE.

Note that both kernels are kernels of orthogonal projections in $L_2(\mu)$.

\subsection{Analysis of cumulants}\label{cumulantsanalysis}

The cumulant generating function of the random variable $X^{\theta_0}_{\widetilde{f},{n,\gamma}}$ is 
\begin{equation} \nonumber
\log\mathbb{E}\left[e^{tX^{\theta_0}_{\widetilde{f},{n,\gamma}}}\right]=\sum_{k=1}^\infty C^{(n)}_k \left(X^{\theta_0}_{\widetilde{f},{n,\gamma}}
\right)t^k.
\end{equation}
Using the Taylor series of $\log(1 + x)$ and $e^x$, we can express $C^{(n)}_k\left(X^{\theta_0}_{\widetilde{f},{n,\gamma}} \right)$ as
a polynomial in the first $k$ moments, and vice versa we can express the $k^{th}$ moment as a polynomial in the first $k$ cumulants.\\ The
latter fact means that proving Theorem \ref{thm:main-result} is the same as proving
\begin{equation} \nonumber
\left|C^{(n)}_k \left(X^{\theta_0}_{\widetilde{f},{n,\gamma}} \right)-C^{(n)}_{k,0}\left(X^{\theta_0}_{\widetilde{f},{n,\gamma}}
\right)\right|\xrightarrow[n\to \infty ]{}0.
\end{equation}
where $C^{(n)}_k,C^{(n)}_{k,0}$ are taken with respect to $\mu$ and $\mu_0$ respectively. Moreover, recall that the moment generating function of a Gaussian
random variable, $Y\sim \mathcal{N}(0,\sigma^2)$ is $\mathbb{E}[e^{tY}]=\exp\left(\frac{t^2}{2}\sigma^2 \right)$. Thus the cumulant generating
function of $Y$ is $\log\mathbb{E}[e^{tY}]=\frac{t^2}{2}\sigma^2$.

By Levy’s Convergence Theorem (for a proof see \cite[Section 18.1]{williams}) in order to prove a central limit theorem we need to prove

    \begin{equation}
         \lim_{n\rightarrow \infty }C^{(n)}_2\left(X^{\theta_0}_{\widetilde{f},{n,\gamma}}\right) \text{ exists}   \end{equation}
   \begin{equation}
      \lim_{n\rightarrow \infty }C^{(n)}_k\left(X^{\theta_0}_{\widetilde{f},{n,\gamma}}\right)=0\quad \forall k>2.    \end{equation}

\subsection {Cumulants in terms of the CMV/GGT matrix} \label{sec:cueasdpp}

Let $f:\partial\mathbb{D}\to \mathbb {R}$ with $\text{supp}f\subset \left \{e^{i\theta} \mid \theta \in (-\pi,\pi)\right \}$ and let $\theta_0
\in (-\pi,\pi)$. For $0<\gamma<1$ let
$$
\widetilde{f}^{\theta_0}_{n,\gamma}
=f\left(\mathcal{M}\left(n^\gamma\left(\mathcal{M}^{-1}\left(e^{i\left(\cdot-\theta_0\right)}\right)\right)\right)\right).$$
Then, by the determinantal structure (\cite[Proposition 1]{soshnikov2}),
\beq\label{eq:momentgenFreddet}
\mathbb{E}\left[\exp\left({tX^{\theta_0}_{\widetilde{f},{n,\gamma}}}\right)\right]
=\det\left(1+\left(e^{t{\widetilde{f}}^{\theta_0}_{n,\gamma}}-1\right)\widetilde{P}_n\right)_{L^2(\mu)}
\eeq
where the determinant on the RHS is the Fredholm determinant of the operator $1+(e^{t{\widetilde{f}}^{\theta_0}_{n,\gamma}}-1)
\widetilde{P}_n$. Here $1$ stands for the identity operator, $(e^{t{\widetilde{f}}^{\theta_0}_{n,\gamma}}-1)$ is the multiplication operator
$g \mapsto  \left(e^{t{\widetilde{f}}^{\theta_0}_{n,\gamma}}-1\right)\times g$, and $\widetilde{P}_n$ can be taken as either the orthogonal projection with kernel $K_n$ or the one with kernel $\widetilde{K}_n$.

As in \cite{realope}, we use the formula $\det\left(1+A\right)=\exp\left(\sum_{j=1}^\infty\frac{(-1)^{j+1}}{j}\Tr A^j\right),$ (valid for a
trace class operator $A$ with norm$<1$, \cite[Section 2]{simontracecalss}), together with a Taylor expansion of the exponential inside the
trace to see that
\begin{equation*} \begin{split} &\quad
C^{(n)}_1\left(X^{\theta_0}_{\widetilde{f},{n,\gamma}}\right)=\textrm{Tr}\widetilde{P}_n\left({\widetilde{f}}^{\theta_0}_{n,\gamma} \right)\\
&\text{and for $m\geq2$:}\\ &\quad C^{(n)}_m\left(X^{\theta_0}_{\widetilde{f},{n,\gamma}}\right) \\ &\qquad= \sum_{j=1}^{m}\frac{(-1)^{j+1}}{j}
\sum_{l_1+\ldots+l_j=m, l_i \geq 1}\frac{\textrm{Tr}\left({\widetilde{f}}^{\theta_0}_{n,\gamma} \right)^{l_1}\widetilde{P}_n\cdots
\left({\widetilde{f}}^{\theta_0}_{n,\gamma} \right)^{l_j}\widetilde{P}_n -\textrm{Tr}\left({\widetilde{f}}^{\theta_0}_{n,\gamma} \right)^m
\widetilde{P}_n}{l_1! \cdots l_j!}.
\end{split} \end{equation*}

Now, let $P_n$ be the projection in $\ell^2(\mathbb{N})$ onto the first $n$ coordinates. A choice of a basis for $L_2(\mu)$ turns the above expressions for cumulants into restricted traces of appropriate matrices.

Explicitly, the choice of the CMV basis leads to
\begin{equation} \label{eq:cmvcumulant}
\begin{split}
&\quad C^{(n)}_m\left(X^{\theta_0}_{\widetilde{f},{n,\gamma}}\right) =\\ &\qquad \sum_{j=1}^{m}\frac{(-1)^{j+1}}{j} \sum_{l_1+\ldots+l_j=m, l_i \geq
1}\frac{\textrm{Tr}\left({\widetilde{f}}^{\theta_0}_{n,\gamma}(\mathcal{C}) \right)^{l_1}P_n\cdots
\left({\widetilde{f}}^{\theta_0}_{n,\gamma}(\mathcal{C})  \right)^{l_j}P_n
-\textrm{Tr}\left({\widetilde{f}}^{\theta_0}_{n,\gamma}(\mathcal{C})  \right)^m P_n}{l_1! \cdots l_j!}
\end{split}
\end{equation}
where $\mathcal{C}$ is the CMV matrix from \eqref{eq:CMV}, and the trace is taken on $\ell^2(\mathbb{N})$.

In the case that the orthogonal polynomials are dense in $L^2(\mu)$ we may choose $\left \{\varphi_n \right \}_n$ as a basis and we have
\begin{equation} \label{eq:cmvcumulantGGT}
\begin{split}
&\quad C^{(n)}_m\left(X^{\theta_0}_{\widetilde{f},{n,\gamma}}\right) =\\
&\qquad \sum_{j=1}^{m}\frac{(-1)^{j+1}}{j} \sum_{l_1+\ldots+l_j=m, l_i \geq 1}\frac{\textrm{Tr}\left({\widetilde{f}}^{\theta_0}_{n,\gamma}(G)
\right)^{l_1}P_n\cdots \left({\widetilde{f}}^{\theta_0}_{n,\gamma}(G)  \right)^{l_j}P_n
-\textrm{Tr}\left({\widetilde{f}}^{\theta_0}_{n,\gamma}(G)  \right)^m P_n}{l_1! \cdots l_j!}
\end{split}
\end{equation}
where $G$ is now the GGT matrix associated with $\mu$ from \eqref{eq:GGT}.

For simplicity, from now on, we abuse the notation and for an operator $\mathcal{A}$ defined on $\ell^2(\mathbb{N})$ we write
$$C^{(n)}_m(\mathcal{A})=\sum_{j=1}^{m}\frac{(-1)^{j+1}}{j} \sum_{l_1+\ldots+l_j=m, l_i \geq
1}\frac{\textrm{Tr}\left(\mathcal{A}\right)^{l_1}P_n\cdots \left(\mathcal{A}  \right)^{l_j}P_n -\textrm{Tr}\left(\mathcal{A}\right)^m
P_n}{l_1! \cdots l_j!},$$
and note that we have the formula
\beq \label{eq:momentgenFreddet1}
\exp\left(\sum_{m=2}^\infty t^m C^{(n)}_m \left(\mathcal{A}\right)  \right) \\
 =\det\left(I+ P_n\left({\rm e}^{ t \mathcal{A}}-I \right)P_n\right)  {\rm e}^{-t \Tr P_n \mathcal{A}},
 \eeq
which holds uniformly for $t$ in a sufficiently small neighborhood (depending on $\mathcal{A}$) of $0$.

\subsection{Toeplitz and Block Toeplitz Operators}
In Section 3, we reduce our analysis to the analysis of a Toeplitz operator. Here we recall a few results from the theory of Toeplitz and block Toeplitz operators. For a vector valued function, $\xi:\partial \mathbb{D} \to \mathbb{C}^d$ and $p\geq 1$, we let 
$$\left\|\xi\right\|_\mathbf{p}=\sum_{1\leq i\leq d}\left\|\xi_{i}\right\|_p.$$
The $p$'th norm of a \emph{matrix} valued function (MVF), $\zeta : \partial \mathbb{D} \to  \mathbb{M}_d\left(\mathbb{C}\right)$, is defined similarly, as a sum over the norms of its entries
$$\left\|\zeta\right\|_\mathbf{p}=\sum_{1\leq i,j\leq d}\left\|\zeta_{ij}\right\|_p.$$
We define the $k^{th}$ (scalar, vector, or matrix) \textit{Fourier} coefficient of a (scalar, vector, or matrix) valued function, $\psi$ by $$\widehat{\psi}_k=\int_{\partial\mathbb{D}}\psi(e^{i\theta})\cdot
e^{-ik\theta}\frac{d\theta}{2\pi}$$
whenever the integral is defined.

Let
\begin{equation*}
        L^2\left(\partial \mathbb{D};\mathbb{C}^d\right):= \left\{\xi:\partial\mathbb{D}\to \mathbb{C}^d:\: \left\|\xi\right\|_\textbf{2}<\infty\right\}
\end{equation*}
and let $\mathbb{H}^{(2;d)} \subseteq L^2\left(\partial \mathbb{D}; \mathbb{C}^d \right)$ be its subspace of functions whose negative Fourier coefficients are all zero. In addition, let
\begin{equation*}
        L^\infty \left( \partial \mathbb{D}; \mathbb{M}_d(\mathbb{C}) \right)= \left\{\zeta:\partial\mathbb{D}\to \mathbb{M}_d\left(\mathbb{C}\right):\:\|\zeta\|_\mathbf{\infty}<\infty \right\}.
\end{equation*}

For $\psi \in L^\infty \left(\partial \mathbb{D}; \mathbb{M}_d (\mathbb{C}) \right)$ we define the operator
$$M_\psi:\mathbb{H}^{(2;d)}\to \mathbb{H}^{(2;d)} $$ in the following way
$$M_\psi f=\widetilde{P}(\psi\cdot f)$$
where $\widetilde{P}$ is the orthogonal projection from $L^2\left(\partial \mathbb{D};\mathbb{C}^d\right)$ to $\mathbb{H}^{(2;d)}$.

The operator $M_\psi$ is represented with respect to the Fourier basis by the matrix
$$T(\psi)=\begin{pmatrix}
\widehat{\psi}_0 & \widehat{\psi}_{-1} & \widehat{\psi}_{-2} & \widehat{\psi}_{-3}&\ldots\\
\widehat{\psi}_1 & \widehat{\psi}_{0} & \widehat{\psi}_{-1} & \widehat{\psi}_{-2}&\ldots\\
\widehat{\psi}_2 & \widehat{\psi}_{1} & \widehat{\psi}_{0} & \widehat{\psi}_{-1} &\ldots\\
\widehat{\psi}_3 & \widehat{\psi}_{2} & \widehat{\psi}_{1} & \widehat{\psi}_{0} &\ldots\\
\vdots & \vdots & \vdots &\vdots &\ddots
\end{pmatrix}.$$
We say that $T(\psi)$ is the \textit{Block Toeplitz} (or just ``Toeplitz" if $d=1$) matrix associated with the symbol $\psi$.

The \textit{Block Hankel} (or simply \textit{Hankel} in the case $d=1$) matrix associated with the symbol $\psi$ is
$$H(\psi)=\begin{pmatrix}
\widehat{\psi}_0 & \widehat{\psi}_{1} & \widehat{\psi}_{2} & \widehat{\psi}_{3}&\ldots\\
\widehat{\psi}_1 & \widehat{\psi}_{2} & \widehat{\psi}_{3} & \widehat{\psi}_{4}&\ldots\\
\widehat{\psi}_2 & \widehat{\psi}_{3} & \widehat{\psi}_{4} & \widehat{\psi}_{5} &\ldots\\
\widehat{\psi}_3 & \widehat{\psi}_{4} & \widehat{\psi}_{5} & \widehat{\psi}_{6} &\ldots\\
\vdots & \vdots & \vdots &\vdots &\ddots
\end{pmatrix}.$$

By \cite[Section 6.1]{introtolttm}
\begin{equation}\label{eq:symbolisbdd}\begin{split}
        &\quad \left\|T(\psi)\right\|_\infty=\left\|\psi\right\|_\infty \\
        &\quad \left\|H(\psi)\right\|_\infty\leq\left\|\psi\right\|_\infty.
    \end{split}
\end{equation}

Finally, for later reference we recall the following identity \cite[Theorem 6.4]{introtolttm}. For $a,b\in L^\infty \left( \partial \mathbb{D}; \mathbb{M}_d (\mathbb{C}) \right)$ we have
\begin{equation}\label{eq:tophankid}
    T(a\cdot b)=T(a)T(b)+H(a)H\left(\widetilde{b}\right)
\end{equation}
where $\widetilde{b}(z)=b(1/z)$ and $H$ is a Hankel/block Hankel matrix. In particular, if $d=1$ we have that
\begin{equation} \label{eq:tophankid2} 
T(a)T(b)-T(b)T(a)=H(b)H\left(\widetilde{a}\right)-H(a)H\left(\widetilde{b}\right).
\end{equation}

\section{A Preliminary CLT for OPE of Geronimus measures}

Let $\alpha\in (-1,1)$. We define $I_\alpha=2\arcsin{(|\alpha|)}$. The measure
$$d\mu_\alpha(\theta)=\frac{\sqrt{\cos^2\left(I_\alpha/2\right)-\cos^2\left(\theta/2\right)}\chi_{\left(I_\alpha,2\pi-I_\alpha\right)}(\theta)d\theta}
{2\pi|1+\alpha|\sin(\theta/2)}$$
is a probability measure on the unit circle whose Verblunsky coefficients are the constant $\alpha$ (see \cite[Example 1.6.12]{simonopuc}). Such measures are also known as
``Geronimus measures''. Note that the case $\alpha=0$ corresponds to the measure $d\mu(\theta)=\frac{d\theta}{2\pi}$ on the circle.

Recall \eqref{eq:PoissonScaled} and \eqref{eq:PoissonLinStat}: for $\theta_0\in \left(I_\alpha,2\pi-I_\alpha\right)$, $0<\gamma<1$ and
$\eta\in \mathbb{C}, \text{Re}\eta>0$ we let
$$\Psi_{n,\gamma,\eta} \left(\theta \right)=\frac{1}{n^\gamma}\text{Re}\left(\frac{e^{i\theta}+\left(1-\frac{\eta}{n^\gamma}\right)}
{e^{i\theta}-\left(1-\frac{\eta}{n^\gamma}\right)}\right)$$
so the associated linear statistic is
$$X_{\Psi_{n,\gamma,\eta}}^{\theta_0}=\sum_{k=1}^n\Psi_{n,\gamma,\eta} \left(\theta_k-\theta_0 \right).$$

In this section we prove
\begin{prop} \label{prop:PoissonForConstant}

In the case of the OPE associated with the measure $\mu_\alpha$,
$$X_{\Psi_{n,\gamma,\eta}}^{\theta_0}-\bbE X_{\Psi_{n,\gamma,\eta}}^{\theta_0}\xrightarrow[n\to \infty]{
\mathcal{D}}\mathcal{N}\left(0,\sigma^2\right)$$
and $\sigma^2=\frac{2}{\left(\eta+\overline{\eta}\right)^2}$.

\end{prop}

  For the sake of simplicity we denote
  \begin{equation} \label{eq:omega_Definition}
\omega_n=\left(1-\frac{\eta}{n^\gamma}\right)e^{i\theta_0}
\end{equation} 
and we keep in mind that $\omega_n$ is also
$\eta,\gamma,\theta_0$ dependent so we get
$$\Psi_{n,\gamma,\eta} \left(\theta-\theta_0\right)=\Psi_{\omega_n,n}(\theta)=\frac{1}{n^\gamma}\text{Re}\left(\frac{e^{i\theta}+\omega_n} {e^{i\theta}-\omega_n}\right).$$

By the discussion in Section 2.3, in order to prove Proposition \ref{prop:PoissonForConstant}, it suffices to show that
\begin{equation} \label{eq:PoissonForConstant1}
         \lim_{n\rightarrow \infty }C^{(n)}_2\left(X^{\theta_0}_{\Psi_{\omega_n,n}}\right)=\sigma^2,
\end{equation}
\begin{equation} \label{eq:PoissonForConstant2}
      \lim_{n\rightarrow \infty }C^{(n)}_m\left(X^{\theta_0}_{\Psi_{\omega_n,n}}\right)=0\quad \forall m>2.
\end{equation}

We shall consider the cases  $\alpha=0$ and $\alpha \neq 0$ separately. Note that in the case $\alpha=0$ a very similar result was proven by Soshnikov in \cite{soshnikov}. However, since the function we consider
does not seem to fall into the category considered in \cite{soshnikov}, we treat it directly here.

\subsection{The case $\alpha \neq 0$.\\}
We start with the case $\alpha \neq 0$. Let $\mu_\alpha$ be a Geronimus measure associated with $\alpha_n=\alpha\not=0$. In this case, the set
of the orthogonal polynomials forms a complete orthonormal set in $L^2(\mu_\alpha)$, so we may use the form \eqref{eq:cmvcumulantGGT} of the
cumulant representation. The GGT matrix associated with $\mu_\alpha$ is
\begin{equation} \label{eq:Galpha}
G_\alpha=\begin{pmatrix} \alpha & \alpha\rho & \alpha\rho^2 & \alpha\rho^3 & \ldots \\
\rho & -\alpha^2 &  -\alpha^2\rho & -\alpha^2\rho^2 & \ldots \\
0 & \rho & -\alpha^2 & -\alpha^2\rho  & \ldots \\
0 & 0 & \rho &  -\alpha^2 & \ldots\\
\ldots & \ldots & \ldots & \ldots & \ddots
\end{pmatrix}.
\end{equation}
where $\rho=\sqrt{1-|\alpha|^2}$. Thus, in this case, \eqref{eq:PoissonForConstant1} and \eqref{eq:PoissonForConstant2} become
\begin{equation} \label{eq:PoissonForConstant3}
         \lim_{n\rightarrow \infty }C^{(n)}_2\left(\frac{1}{n^\gamma}\text{Re}\left(\frac{G_\alpha+\omega_n}{G_\alpha-\omega_n}\right)\right)=\sigma^2,
\end{equation}
\begin{equation} \label{eq:PoissonForConstant4}
      \lim_{n\rightarrow \infty }C^{(n)}_m\left(\frac{1}{n^\gamma}\text{Re}\left(\frac{G_\alpha+\omega_n}{G_\alpha-\omega_n}\right)\right)=0\quad \forall m>2.
\end{equation}

The proof of \eqref{eq:PoissonForConstant3} and \eqref{eq:PoissonForConstant4} consists of two steps. The first one is a reduction to the computation of the limiting cumulants of a Toeplitz operator and the second is the actual computation of these cumulants.

\subsubsection{The cumulants in terms of a Toeplitz operator}
Note that except for its first row and column, $G_\alpha$ has the shape of a Toeplitz operator. Indeed,
\begin{equation} \label{eq:Tphi}
T(\phi):=\begin{pmatrix}-{\alpha} & 0 & 0&\ldots\\
0 & 1 & 0&\ldots\\
0 & 0 & 1&\ldots\\
\vdots & \vdots & \vdots  &\ddots
\end{pmatrix}G_\alpha=\begin{pmatrix} -\alpha^2 & -\alpha^2\rho & -\alpha^2\rho^2 & -\alpha^2\rho^3 & \ldots \\
\rho & -\alpha^2 &  -\alpha^2\rho & -\alpha^2\rho^2 & \ldots \\
0 & \rho & -\alpha^2 & -\alpha^2\rho  & \ldots \\
0 & 0 & \rho &  -\alpha^2 & \ldots\\
\ldots & \ldots & \ldots & \ldots & \ddots
\end{pmatrix}
\end{equation}
is a Toeplitz operator corresponding to the symbol $\phi:\partial\mathbb{D}\to\mathbb{C}$ $$\phi(e^{i\theta})=\rho\cdot
e^{-i\theta}-\alpha^2\sum_{k=0}^\infty \rho^ke^{ik\theta}.$$

We want to show that for $m\geq 2$,
\begin{equation} \label{eq:CayleyToeplitzCompare}
\left|C^{(n)}_m\left(\frac{1}{n^\gamma}\text{Re}\left(\frac{G_\alpha+\omega_n}{G_\alpha-\omega_n}\right)\right)-C^{(n)}_m\left(\frac{1}{n^\gamma}\text{Re}\left(\frac{T\left(\phi\right)+\omega_n}{T\left(\phi\right)-\omega_n}\right)\right)\right|\to
0
\end{equation}
as $n\to \infty$. Our main tool for this is \cite[Theorem 3.1]{realope} (which we quote as Proposition \ref{th:comparison-general} below), so we need to show that both
$\frac{G_\alpha+\omega_n}{G_\alpha-\omega_n}$ and $\frac{T\left(\phi\right)+\omega_n}{T\left(\phi\right)-\omega_n}$ satisfy the conditions of that theorem.

We therefore start with
\begin{prop} \label{prop:CombsThomas}
Let $G=G_\alpha$ and $T(\phi)$ be as defined in \eqref{eq:Galpha} and \eqref{eq:Tphi}.
Then for any $z \in \mathbb{D}$
\beq \label{eq:CombesThomas}
\begin{split}
    &\left|\left(\frac{G+z}{G-z}\right)_{ij} \right|,\left|\left(\frac{T\left(\phi\right)+z}{T\left(\phi\right)-z}\right)_{ij}\right|\\
    &\quad \leq\frac{C}{d\left(z,\partial\mathbb{D}\right)}e^{-\Theta|i-j|}
\end{split}
\eeq
    where $\Theta=\min\left(\frac{d\left(z,\partial\mathbb{D}\right)}{2\left(\frac{12}{(1-\rho)^2}
    +3\right)},\ln\left(\frac{1+\rho}{2\rho}\right) \right)$ and $C>0$ is independent of $i,j$.
\end{prop}
\begin{proof}
We start by proving the claim for $\frac{G+z}{G-z}$.

Let  $\Theta=\min\left(\frac{d\left(z,\partial\mathbb{D}\right)}{2\left(\frac{12}{(1-\rho)^2} +3\right)},\ln\left(\frac{1+\rho}{2\rho}\right)
\right)$  and define the diagonal matrix $R$ by
\beq \no
R_{j,j}=e^{\Theta j}.
\eeq
We have
\beq \no
e^{\Theta (j-i)} \left( \left(G-z\right)^{-1} \right)_{i,j}=\left( R^{-1} \left(G-z\right)^{-1} R \right)_{i,j},
\eeq
which implies that
\beq \label{eq:CT1}
\left|   \left( \left(G-z\right)^{-1} \right)_{i,j} \right| \leq e^{-\Theta (j-i)} \left \|  R^{-1} \left(G-z\right)^{-1} R \ \right
\|_\infty,
\eeq
where $\| \cdot \|_\infty$ is the operator norm.

Now note that
\beq \no
R^{-1}\left(G-z\right)^{-1} R=\left(R^{-1}  \left(G-z\right) R\right)^{-1}
=\left( R^{-1} G R-z \right)^{-1}
\eeq
so that, by the resolvent formula
\beq \no
R^{-1}\left(G-z\right)^{-1} R=\left(G-z\right)^{-1}+R^{-1}\left(G-z\right)^{-1} R \left(G-R^{-1}GR \right) \left(G-z\right)^{-1}.
\eeq
Thus
\beq \label{eq:CT2}
\begin{split}
& \left \| R^{-1}\left(G-z\right)^{-1} R \right \|_\infty \\
&\quad \leq \left \|\left(G-z\right)^{-1} \right \|_\infty+ \left \|R^{-1}\left(G-z\right)^{-1} R \right \|_\infty \left \| G-R^{-1}GR \right
\|_\infty \left \| \left(G-z\right)^{-1} \right \|_\infty.
\end{split}
\eeq

Observe that for a matrix $A$ we have
$$\|A\|_\infty \leq \sum_{j \in \mathbb{Z}} \sup_{i} \left | A_{i,i+j} \right |$$
which leads us to
$$
\left \| G-R^{-1}GR \right \|_\infty \leq |\alpha| \sum_{k=1}^\infty\rho^k  \left|e^{k\Theta}-1 \right|+\rho\left|e^{-\Theta}-1\right|\leq
|\alpha| \sum_{k=1}^\infty\rho^k  \left|e^{k\Theta}-1 \right|+\rho e\left|{\Theta}\right|
$$
(note $0<\Theta<1$). Using $a^n-b^n=(a-b)(a^{n-1}+a^{n-2}b+\ldots+ab^{n-2}+b^{n-1})$ one can obtain $\left|e^{k\Theta}-1 \right|\leq \left|e^{\Theta}-1
\right|\cdot ke^{k\Theta} $ so we have
$$
\left \| G-R^{-1}GR \right \|_\infty \leq |\alpha| \left|e^{\Theta}-1 \right| \sum_{k=1}^\infty k \left(\rho e^{\Theta}\right)^k+\rho
e\left|{\Theta}\right|.
$$
Using the fact that $e^\Theta\leq \frac{1+\rho}{2\rho}$ and $0<\rho<1$ we see that $\frac{\rho e^\Theta}{\left(1-\rho e^\Theta\right)^2}\leq
\frac{4}{\left(1-\rho\right)^2} $. Together with $\sum_{k=1}^\infty k \left(\rho e^{\Theta}\right)^k=\frac{\rho e^\Theta}{\left(1-\rho e^\Theta\right)^2}$ this implies that
$$
 |\alpha| \left|e^{\Theta}-1 \right| \sum_{k=1}^\infty k \left(\rho e^{\Theta}\right)^k+\rho e\left|{\Theta}\right| \leq\Theta \left(
 \frac{4|\alpha|e}{(1-\rho)^2} +\rho e \right)<  \Theta\left(\frac{12}{(1-\rho)^2} +3\right).
$$

Also, by the unitarity of $G$ we have $\left \| \left(G-z\right)^{-1} \right \|_\infty \leq
\left(d\left(z,\partial\mathbb{D}\right)\right)^{-1}$.

Plugging these into \eqref{eq:CT2} we see that
\beq \no
\left \| R^{-1}\left(G-z\right)^{-1} R \right \|_\infty \leq \frac{1}{d\left(z,\partial\mathbb{D}\right)}+ \left \|R^{-1}\left(G-z\right)^{-1}
R \right \|_\infty \Theta\left(\frac{12}{(1-\rho)^2} +3\right)\frac{1}{d\left(z,\partial\mathbb{D}\right)},
\eeq
and since
$$\Theta\left(\frac{12}{(1-\rho)^2} +3\right)\frac{1}{d\left(z,\partial\mathbb{D}\right)} \leq 1/2$$
we have
\beq \no
\left \| R^{-1}\left(G-z\right)^{-1} R \right \|_\infty \leq \frac{2}{  d\left(z,\partial\mathbb{D}\right) }.
\eeq

Plugging this into \eqref{eq:CT1} yields for $j\geq i$
\beq \label{eq:SimpleCombesThomas}
\left|\left(G-z\right)^{-1}_{i,j}\right|\leq \frac{2e^{-\Theta|i-j|}}{d\left(z,\partial\mathbb{D}\right)}.
\eeq
For $j<i$ we can obtain the same bound similarly, by interchanging $R$ with $R^{-1}$ and observing that for $k\geq 0$ we have
$\left|e^{-k\Theta}-1\right|=1-e^{-k\Theta}\leq ke^{k\Theta}$.

Thus, we have that
\beq \begin{split}
     &\left|\left(\left(G+z\right)\left(G-z\right)^{-1}
     \right)_{i,j}\right|\\&\quad\leq\sum_{k=0}^\infty\left|\left(G+z\right)_{i,k}\right|\left|\left(G-z\right)^{-1}_{k,j}\right|\\
     &\quad\leq|\alpha|\sum_{k=i+1}^\infty\rho^{k-i}\left|\left(G-z\right)^{-1}_{k,j}\right|+\left(|\alpha|+|z|\right)\left|\left(G-z\right)^{-1}_{i,j}\right|+\rho\left|\left(G-z\right)^{-1}_{i-1,j}\right|\\
     &\quad\leq
     \frac{2}{d\left(z,\partial\mathbb{D}\right)}\left(|\alpha|\sum_{k=i+1}^\infty\rho^{k-i}e^{-\Theta|k-j|}+\left(|\alpha|+|z|\right)e^{-\Theta|i-j|}+\rho
     e^{-\Theta|i-j-1|}\right).\\
     \end{split} \eeq
     Using $1/3<e^{-\Theta}<1$ and $|\alpha|,\rho,|z|<1$ we can see that for $j\leq i$
     \beq \begin{split}
    &\frac{2}{d\left(z,\partial\mathbb{D}\right)}\left(|\alpha|\sum_{k=i+1}^\infty\rho^{k-i}e^{-\Theta|k-j|}+\left(|\alpha|+|z|\right)e^{-\Theta|i-j|}+\rho
    e^{-\Theta|i-j-1|}\right)\\
     &\quad\leq\frac{2e^{-\Theta|i-j|}}{d\left(z,\partial\mathbb{D}\right)}\left(\frac{|\alpha|\rho e^{-\Theta}}{1-\rho
     e^{-\Theta}}+|\alpha|+|z|+\rho e^\Theta\right)\\
     &\quad \leq2\left(\frac{1}{1-\rho}+5\right)\frac{e^{-\Theta|i-j|}}{d\left(z,\partial\mathbb{D}\right)}
     \end{split} \eeq
    and for $i< j$ we have (note $\frac{1}{1-\rho e^\Theta}<\frac{2}{1-\rho}$ and $\rho< e^{-\Theta}$)
     \beq \begin{split}
    &\frac{2}{d\left(z,\partial\mathbb{D}\right)}\left(|\alpha|\sum_{k=i+1}^\infty\rho^{k-i}e^{-\Theta|k-j|}+\left(|\alpha|+|z|\right)e^{-\Theta|i-j|}+\rho
    e^{-\Theta|i-j-1|}\right)\\
    &\quad\leq \frac{2}{d\left(z,\partial\mathbb{D}\right)}\left(|\alpha|\sum_{k=i+1}^j\rho^{k-i}e^{-\Theta|k-j|}+|\alpha|\sum_{k=j+1}^\infty
    \rho^{k-i}+\left(|\alpha|+|z|\right)e^{-\Theta|i-j|}+\rho e^{-\Theta|i-j-1|}\right)\\
     &\qquad\leq\frac{2}{d\left(z,\partial\mathbb{D}\right)}\left(\frac{|\alpha| \rho e^{\Theta} e^{-\Theta(j-i)}}{1-\rho
     e^\Theta}+|\alpha|\rho \frac{\rho^{j-i}}{1-\rho}+\left(|\alpha|+|z|\right)e^{-\Theta(j-i)}+\rho e^{-\Theta(j-i)}e^{-\Theta}\right)\\
     &\qquad\quad\leq\frac{2e^{-\Theta(j-i)}}{d\left(z,\partial\mathbb{D}\right)}\left(\frac{3}{1-\rho e^\Theta}+ \frac{1}{1-\rho}+3\right)\\
     &\qquad\qquad \leq 2\left(3+\frac{7}{1-\rho}\right)\frac{e^{-\Theta|i-j|}}{d\left(z,\partial\mathbb{D}\right)}.
     \end{split} \eeq
     This completes the proof for $\frac{G+z}{G-z}$.

     As for $\frac{T\left(\phi\right)+z}{T\left(\phi\right)-z}$, we first note that $$\phi(e^{i\theta})=\rho\cdot
     e^{i\theta}-\alpha^2\sum_{k=0}^\infty \rho^ke^{-ik\theta}=\frac{\rho e^{i\theta}-1}{1-\rho e^{-i\theta}}$$
    which means that $\text{Range}\left(\phi\right)\subset \partial \bbD$. It can be verified that $\phi\neq 1$, therefore the winding number
    of $\phi$ around any $z \not \in \text{Range}\left(\phi\right)$ is 0, which by \cite[Section 2]{introtolttm} shows that
    $\left(T\left(\phi\right)-z\right)^{-1}$ is well defined. By \cite[Section 2,3]{introtolttm}, the operator norm of the resolvent can be
    bounded  $$\left\|\left(T\left(\phi\right)-z\right)^{-1}\right\|_\infty\leq \frac{2}{d\left(z,\text{Range}\left(\phi\right)\right)}\leq
    \frac{2}{d\left(z,\partial \bbD\right)}.$$
    Which means that one can mimic the same procedure as above to get the desired result for
    $\frac{T\left(\phi\right)+z}{T\left(\phi\right)-z}$.

\end{proof}

\begin{remark} \label{rem:RealPart}
Since $\alpha \in \mathbb{R}$, the entries of $G_\alpha$ and  $T(\phi)$ are real. Thus, Proposition \ref{prop:CombsThomas} and its proof immediately imply that \eqref{eq:CombesThomas} and \eqref{eq:SimpleCombesThomas} still hold when replacing $G_\alpha$ by $G_\alpha^*$ and $T(\phi)$ by $T(\phi)^*$. In particular, \eqref{eq:CombesThomas} also holds for the real and imaginary parts of $\left(\frac{G+z}{G-z}\right)^*$ and $\left(\frac{T\left(\phi\right)+z}{T\left(\phi\right)-z}\right)^*$.
\end{remark}

The essence of Proposition \ref{prop:CombsThomas} is that the structure of $\frac{G+z}{G-z}$ is of an ``almost banded matrix'', meaning that far
from the diagonal the entries of the operator are exponentially small. In \cite{realope} it was shown that for such matrices there is a general comparison theorem.

\begin{prop} \cite[Theorem 3.1] {realope} \label{th:comparison-general}
Fix $0<\gamma<1$ and let $\{ \mathcal A^{(n)} \}_{n=1}^\infty$ $\{ \mathcal B^{(n)} \}_{n=1}^\infty$ be two matrix sequences which are uniformly bounded in operator norm. Assume that
$$\left|\mathcal A^{(n)}_{r,s}\right|,\left|\mathcal B^{(n)}_{r,s}\right|\leq Ce^{-\widetilde{d}\frac{|r-s|}{n^\gamma}}$$
for some $C,\widetilde{d}>0$.

For any two positive numbers $\ell_1<\ell_2$ let
$$P_{\{\ell_1,\ell_2\}}=P_{[\ell_2]+1}-P_{[\ell_1]-1}$$ $($recall $P_n$ is the orthogonal projection onto the first $n$ coordinates$)$.

Then, for any $\beta>\gamma$ and any $m \geq 2$,
\begin{equation} \label{eq:almost-band}
\begin{split}
& \left|C_m^{(n)} \left(\mathcal B^{(n)} \right) - C_m^{(n)} \left( \mathcal A^{(n)} \right) \right| \\
&\quad \leq C(m,\beta) \left\| P_{\left\{n-2mn^\beta,n+2mn^\beta \right\}} \left( \mathcal{B}^{(n)} -\mathcal{A}^{(n)}\right)
P_{\left\{n-2mn^\beta,n+2mn^\beta\right\}} \right\|_1+o(1)
\end{split}
\end{equation}
as $n \rightarrow \infty$, where $C(m,\beta)$ is a constant that depends on $m$ and $\beta$ but is independent of $n$.

\end{prop}

In our case, note that $\left\|\left(G-\omega_n\right)^{-1}\right\|_\infty\leq d\left(\omega_n,\partial\bbD\right)^{-1}=\mathcal{O} \left(n^\gamma\right)$ and
$\left\|\left(T\right(\phi\left)-\omega_n\right)^{-1}\right\|_\infty\leq 2d\left(\omega_n,\partial\bbD\right)^{-1}=\mathcal{O}
\left(n^\gamma\right)$ as $n\to \infty$. Thus, Proposition \ref{prop:CombsThomas} implies that $\mathcal{A}^{(n)}=\frac{1}{n^\gamma}\textrm{Re}\frac{G+\omega_n}{G-\omega_n}$ and $\mathcal{B}^{(n)}=\frac{1}{n^\gamma}\textrm{Re}\left(\frac{T\left(\phi\right)+\omega_n}{T\left(\phi\right)-\omega_n}\right)$ satisfy the conditions of Proposition \ref{th:comparison-general}.

Thus
 \beq \label{eq:cumulantComparisson}
 \begin{split}
    &\left|C^{(n)}_m\left(\frac{1}{n^\gamma}\textrm{Re}\left(\frac{G+\omega_n}{G-\omega_n}\right)\right)-C^{(n)}_m\left(\frac{1}{n^\gamma}\textrm{Re}\left(\frac{T\left(\phi\right)+\omega_n}{T\left(\phi\right)-\omega_n}\right)\right)\right|\\
    &\quad \leq \frac{1}{n^\gamma} C(m,\beta) \left\| P_{\left \{ n-2mn^\beta,n+2mn^\beta \right \}}
    \textrm{Re}\left(\left(\frac{G+\omega_n}{G-\omega_n}\right) -\left(\frac{T\left(\phi\right)+\omega_n}{T\left(\phi\right)-\omega_n}\right) \right)P_{\left \{
    n-2mn^\beta,n+2mn^\beta \right \}} \right\|_1+o(1)
    \end{split}
\eeq as $n\to \infty$.

\begin{prop} \label{prop:GGT2Toeplitz}
For any $m \geq 2$
\beq \label{eq:GGT2Toep1}
\frac{1}{n^\gamma}  \left\| P_{\left \{ n-2mn^\beta,n+2mn^\beta \right \}} \left(
    \left(\frac{G+\omega_n}{G-\omega_n}\right) -\left(\frac{T\left(\phi\right)+\omega_n}{T\left(\phi\right)-\omega_n}\right) \right)P_{\left \{
    n-2mn^\beta,n+2mn^\beta \right \}} \right\|_1\to 0
    \eeq
and
\beq \label{eq:GGT2Toep2}
\frac{1}{n^\gamma}  \left\| P_{\left \{ n-2mn^\beta,n+2mn^\beta \right \}} \textrm{Re}\left(
    \left(\frac{G+\omega_n}{G-\omega_n}\right) -\left(\frac{T\left(\phi\right)+\omega_n}{T\left(\phi\right)-\omega_n}\right) \right)P_{\left \{
    n-2mn^\beta,n+2mn^\beta \right \}} \right\|_1\to 0
    \eeq
    as $n\to \infty$
\end{prop}
\begin{proof}
We note that
\begin{equation} \nonumber
\frac{G+\omega_n}{G-\omega_n}-\frac{T\left(\phi\right)+\omega_n}{T\left(\phi\right)-\omega_n}=
2\omega_n(G-\omega_n)^{-1}(G-T\left(\phi\right))(T\left(\phi\right)-\omega_n)^{-1}
\end{equation}
and note that $\left(G-T(\phi)\right)_{jk}=0$ for any $1<j\in\mathbb{N}$. This implies that for $n$ sufficiently large
\beq \no
P_{\left \{ n-4mn^\beta,n+4mn^\beta \right \}}(G-T(\phi))=0.
\eeq

Thus, letting $\widetilde{P}_2=P_{\left \{ n-2mn^\beta,n+2mn^\beta \right \}}$ and $\widetilde{P}_4=P_{\left \{ n-4mn^\beta,n+4mn^\beta \right \}}$,
\begin{equation} \label{eq:cumulantComparisson2}
\begin{split}
& \widetilde{P}_2 \left(\left(\frac{G+\omega_n}{G-\omega_n}\right) -\left(\frac{T\left(\phi\right)+\omega_n}{T\left(\phi\right)-\omega_n}\right) \right)\widetilde{P}_2\\
&=2\omega_n \widetilde{P}_2 (G-\omega_n)^{-1}  \widetilde{P}_4(G-T\left(\phi\right))(T\left(\phi\right)-\omega_n)^{-1}\widetilde{P}_2 \\
&\quad+2\omega_n \widetilde{P}_2 (G-\omega_n)^{-1}\left(I-\widetilde{P}_4\right)(G-T\left(\phi\right))(T\left(\phi\right)-\omega_n)^{-1}\widetilde{P}_2 \\
&=2\omega_n \widetilde{P}_2 (G-\omega_n)^{-1}\left(I-\widetilde{P}_4\right)(G-T\left(\phi\right))(T\left(\phi\right)-\omega_n)^{-1}\widetilde{P}_2
\end{split}
\end{equation}
for all $n$ large enough.

Now note that by \cite[Lemma 3.2]{realope} and \eqref{eq:SimpleCombesThomas}
\beq \no
\left \|\widetilde{P}_2 (G-\omega_n)^{-1}\left(I-\widetilde{P}_4\right) \right\|_1 \leq D n^{\beta+2\gamma}e^{-D'n^{\beta-\gamma}}
\eeq
for some constants $D, D'>0$, and the same holds for $\left \|\widetilde{P}_2 (G^*-\overline{\omega_n})^{-1}\left(I-\widetilde{P}_4\right) \right\|_1$ by Remark \ref{rem:RealPart}. Using this together with the fact that $\left(\omega_n\right)_n$ and $\left \|\frac{1}{n^\gamma}(G-T\left(\phi\right))(T\left(\phi\right)-\omega_n)^{-1} \right \|$ are bounded, we get that
\begin{equation} \label{eq:cumulantComparisson3}
\left \|\frac{2\omega_n}{n^\gamma} \widetilde{P}_2 (G-\omega_n)^{-1}\left(I-\widetilde{P}_4\right)(G-T\left(\phi\right))(T\left(\phi\right)-\omega_n)^{-1}\widetilde{P}_2 \right \|_1 \leq \widetilde{D} n^{\beta+2\gamma}e^{-\widetilde{D}'n^{\beta-\gamma}}
\end{equation}
and the same holds for the real part. This, by \eqref{eq:cumulantComparisson2}, finishes the proof.
\end{proof}

By \eqref{eq:cumulantComparisson}, Proposition \ref{prop:GGT2Toeplitz} implies
\begin{corollary}
 For any $m\geq 2$ $$\left|C^{(n)}_m\left(\frac{1}{n^\gamma}\left(\frac{G+\omega_n}{G-\omega_n}\right)\right)-C^{(n)}_m\left(\frac{1}{n^\gamma}\left(\frac{T\left(\phi\right)+\omega_n}{T\left(\phi\right)-\omega_n}\right)\right)\right|\rightarrow 0$$ and
$$\left|C^{(n)}_m\left(\frac{1}{n^\gamma}\text{Re}\left(\frac{G+\omega_n}{G-\omega_n}\right)\right)-C^{(n)}_m\left(\frac{1}{n^\gamma}\text{Re}\left(\frac{T\left(\phi\right)+\omega_n}{T\left(\phi\right)-\omega_n}\right)\right)\right|\to 0$$
as $n\to \infty$.
\end{corollary}

We are now half-way through our reduction procedure (the first step of the proof of Proposition \ref{prop:PoissonForConstant}). The operator $\frac{T(\phi)+\omega_n}{T(\phi)-\omega_n}$ is an intermediate object on our way to the Toeplitz operator whose cumulants we will actually compute. This Toeplitz operator will arise naturally from inverting $T(\phi)-\omega_n$ using the \emph{Wiener-Hopf} factorization.

A symbol $a:\partial\mathbb{D}\to\mathbb{C}$ is said to have a Weiner-Hopf factorization if $a=a_+\cdot a_-$ where $a_+$ is a function such that both $a_+$ and $a_+^{-1}$ are analytic in $\mathbb{D}$ and both $a_-$ and $a_-^{-1}$ are analytic in $\mathbb{C}\setminus\overline{\mathbb{D}}$. This factorization is particularly handy when one wishes to invert a Toeplitz operator, since in this case (see e.g.\ \cite[Section 1]{introtolttm})
\beq \label{eq:WHGeneral}
\left(T(a)\right)^{-1}=T\left(a_+^{-1}\right)T\left(a_-^{-1}\right).
\eeq

    We start with observing that $T(\phi)-\omega_n I=T(\phi-\omega_n)=T(h)$ where
    $h:\partial\mathbb{D}\to\mathbb{C}$ is the function
    \begin{equation} \nonumber
    \begin{split}
    h(z)&=\rho \cdot z-(\alpha^2+\omega_n)-\alpha^2\sum_{k=1}^\infty\rho^kz^{-k}=\rho
    z-(\alpha^2+\omega_n)-\frac{\alpha^2\rho}{z}\left(\frac{1}{1-\frac{\rho}{z}}\right) \\
    &=\frac{\rho z^2-(1+\omega_n)z+\omega_n \rho}{z-\rho}.
    \end{split}
    \end{equation}

To obtain the Wiener-Hopf factorization of $h$ we define
\beq \label{eq:Pomega}
P_{\omega_n}(z)=z^2-\frac{1+\omega_n}{\rho}z+\omega_n
\eeq
so that
$$h(z)=\frac{\rho P_{\omega_n}(z)}{z-\rho}.$$
We denote the roots of $P_{\omega_n}$ by $$z^{\omega_n}_{\pm}=\frac{1+\omega_n}{2\rho}\pm
\frac{\sqrt{\left({1+\omega_n}\right)^2-4\omega_n\rho^2}}{2\rho}$$ where the branch cut of the square-root is the positive real line (so that for $z=Re^{i\theta}$ with $\theta \in (0,2\pi)$, $\sqrt{z}=\sqrt{R}e^{i\theta/2}$). We shall presently show that for large enough $n$ (recall that $\omega_n=\left(1-\frac{\eta}{n^\gamma}\right)e^{i\theta_0}$)
    \begin{equation} \label{eq:zomega}
    |z^{\omega_n}_+|>1>|z^{\omega_n}_-|,
    \end{equation}
    and so a Weiner-Hopf factorization of $h$, for $n$ large enough, is given by
    \begin{equation} \label{eq:WHFactorization}
    h(z)=  \left(\rho \left(z-z^{\omega_n}_{+}\right)\right)\left(\frac{z-z^{\omega_n}_{-}}{z-\rho}\right)
    \end{equation}
    (the first factor is clearly analytic in $\mathbb{D}$ and the second one in $\mathbb{C}\setminus \overline{\mathbb{D}}$).

   To show \eqref{eq:zomega} we expand the square root to see that
   \beq \label{eq:rootsofPomega}
   z^{\omega_n}_{\pm}=\frac{1+e^{i\theta_0}\pm\sqrt{\left(1+e^{i\theta_0}\right)^2-4\rho^2 e^{i\theta_0}}}{2\rho} \mp
\frac{\eta e^{i\theta_0}}{2\rho n^\gamma}\left(\frac{\left(1+e^{i\theta_0}-2\rho^2\right)}{\sqrt{\left(1+e^{i\theta_0}\right)^2-4\rho^2
e^{i\theta_0}}})\pm1\right)+o(n^{-\gamma})
\eeq
as $n\to \infty $. We denote
$$z^\infty_{\pm}=\frac{1+e^{i\theta_0}\pm\sqrt{\left(1+e^{i\theta_0}\right)^2-4\rho^2e^{i\theta_0}}}{2\rho}$$ so \eqref{eq:rootsofPomega}
becomes \beq z^{\omega_n}_{\pm}=z^\infty_{\pm}\mp\frac{\eta e^{i\theta_0}}{n^\gamma}\frac{z^\infty_{\pm}-
\rho}{\sqrt{\left(1+e^{i\theta_0}\right)^2-4\rho^2e^{i\theta_0}}} +o(n^{-\gamma}). \eeq
Using the relation $z^\infty_{\pm}\cdot z^\infty_{\mp}=
e^{i\theta_0}$ we can write $$z^{\omega_n}_{\pm}=z^\infty_{\pm}\left(1\mp \frac{\eta }{n^\gamma}\frac{e^{i\theta_0}-\rho
z^\infty_{\mp}}{\sqrt{\left(1+e^{i\theta_0}\right)^2-4\rho^2 e^{i\theta_0}}}\right) +o(n^{-\gamma}).$$
It is easy to verify that
$\left|z^\infty_{\pm}\right|=1$ so that
$$\left|z^{\omega_n}_{\pm}\right|=\left| 1\mp \frac{\eta }{n^\gamma}\frac{e^{i\theta_0}-\rho
z^\infty_{\mp}}{\sqrt{\left(1+e^{i\theta_0}\right)^2-4\rho^2 e^{i\theta_0}}} \right| +o(n^{-\gamma}).$$

We note that (recall that the branch cut is the positive real line, thus $e^{-i\frac{\theta_0}{2}}=-\sqrt{e^{-i\theta_0}}$)
\beq
\label{eq:exp(theta/2)}\frac{e^{-i\frac{\theta_0}{2}}}{e^{-i\frac{\theta_0}{2}}}\cdot \frac{e^{i\theta_0}-\rho
z^\infty_{\mp}}{\sqrt{\left(1+e^{i\theta_0}\right)^2-4\rho^2 e^{i\theta_0}}}=\frac{\sin\left(\frac{\theta_0}{2}\right)\mp
\sqrt{\rho^2-\cos^2\left(\frac{\theta_0}{2}\right)}}{2\sqrt{\rho^2-\cos^2\left(\frac{\theta_0}{2}\right)}}\eeq
so (recall $\theta_0\in \supp \mu_\alpha$ and hence $\cos^2\left(\frac{\theta_0}{2}\right)<\rho^2$)
\beq \label{eq:modolusofzpm}\left|z^{\omega_n}_{\pm}\right|^2=\left(1\pm\frac{\text{Re}\eta}{n^\gamma}\cdot\frac{\sin\left(\frac{\theta_0}{2}\right)\mp\sqrt{\rho^2-\cos^2\left(\frac{\theta_0}{2}\right)}}{2\sqrt{\rho^2-\cos^2\left(\frac{\theta_0}{2}\right)}}\right)^2+o(n^{-\gamma}).\eeq
Since $\sin\left(\frac{\theta_0}{2}\right)>\sqrt{\rho^2-\cos^2\left(\frac{\theta_0}{2}\right)}$ (square both sides) and $\text{Re}\eta>0$, we get \eqref{eq:zomega} for all $n$ sufficiently large.

    Now, from \eqref{eq:WHFactorization} and \eqref{eq:WHGeneral} we get
    $$(T(\phi)-\omega_n)^{-1}=\left(T(h)\right)^{-1}=T\left(\frac{1}{\rho
    \left(z-z^{\omega_n}_{+}\right)}\right)T\left(\frac{z-\rho}{z-z^{\omega_n}_{-}}\right).$$
We also have
    $$(T(\phi)+{\omega_n})=T\left(\frac{1}{z-\rho}\left(\rho z^2-(1-{\omega_n})z-\omega_n \rho\right)\right)=
    T\left(\frac{\rho P_{-\omega_n}(z)}{z-\rho} \right)$$
    and so
    $$(T(\phi)-{\omega_n})^{-1}(T(\phi)+\omega_n)=T\left(\frac{1}{\rho\left(z-z^{\omega_n}_{+}\right)}\right)T\left(\frac{z-\rho}{z-z^{\omega_n}_{-}}\right)T\left(\frac{\rho P_{-\omega_n}(z)}{z-\rho} \right),$$
    which by \eqref{eq:tophankid} and the fact that $H\left(\frac{z-\rho}{z-z^{\omega_n}_{-}} \right)=H\left(\frac{1-\frac{\rho}{z}}{1-\frac{z^{\omega_n}_-}{z}} \right)=0$ is \beq \label{eq:Toeplitzmulti}
    T\left(\frac{1}{\rho\left(z-z^{\omega_n}_{+}\right)}\right)T\left(\frac{\rho P_{-{\omega_n}}(z)}{z-z^{\omega_n}_{-}}\right). \eeq
Using \eqref{eq:tophankid} again we therefore see that
\beq \label{eq:toeplitz+henkel}
(T(\phi)-\omega_n)^{-1}(T(\phi)+\omega_n)=T\left(\frac{P_{-\omega_n}(z)}{P_{\omega_n}(z)}\right)-H\left(\frac{1}{\rho\left(z-z^{\omega_n}_{+}\right)}\right)H\left(\widetilde{\frac{\rho P_{-\omega_n}(z)}{z-z^{\omega_n}_{-}}}\right).
\eeq

To complete the first step of the proof of Proposition \ref{prop:PoissonForConstant} we want to show
\begin{prop}\label{prop:nohankel}
For $m\geq2$
\beq \label{eq:CayleyToep2Toep1}
\left|C^{(n)}_m\left(\frac{1}{n^\gamma}\left(\frac{T\left(\phi\right)+\omega_n}{T\left(\phi\right)-\omega_n}\right)\right)-C^{(n)}_m\left(\frac{1}{n^\gamma}T\left(\frac{P_{-\omega_n}}{P_{\omega_n}}\right)\right)\right|\xrightarrow[n\to
\infty ]{}0,
\eeq
and
\beq \label{eq:CayleyToep2Toep2}
\left|C^{(n)}_m\left(\frac{1}{n^\gamma}\textrm{Re}\left(\frac{T\left(\phi\right)+\omega_n}{T\left(\phi\right)-\omega_n}\right)\right)-C^{(n)}_m\left(\frac{1}{n^\gamma}\textrm{Re}\left(T\left(\frac{P_{-\omega_n}}{P_{\omega_n}}\right)\right)\right)\right|\xrightarrow[n\to
\infty ]{}0.
\eeq
\end{prop}

We would like to use Proposition \ref{th:comparison-general} again. In order to do so, we need to show that $T\left(\frac{P_{-\omega_n}}{P_{\omega_n}} \right)$ satisfies the conditions of that proposition.

\begin{lemma} \label{lemma:symbolanalysis}
Let $$\frac{P_{-\omega_n}(z)}{P_{\omega_n}(z)}=\sum_{k=-\infty}^\infty \hat{\varphi}^{(n,\omega)}_k\cdot z^k$$ be the Fourier expansion of $\frac{P_{-\omega_n}(z)}{P_{\omega_n}(z)}$. Then there exist $D,d>0$ $($independent of $n)$ such that for any $k\in \mathbb{Z}$
\beq \label{eq:FourierBound}
\left|\hat{\varphi}^{(n,\omega)}_k\right|\leq D
\cdot \exp{\left(-\frac{d}{n^\gamma}|k|\right)}.
\eeq
In particular
\beq \label{eq:0prop}
\left|T\left(\frac{P_{-\omega_n}}{P_{\omega_n}}\right)_{l,m}\right|\leq D
\cdot \exp{\left(-\frac{d}{n^\gamma}|l-m|\right)},
\eeq
and
\beq \label{eq:01prop}
\left|T\left(\text{Re}\frac{P_{-\omega_n}}{P_{\omega_n}}\right)_{l,m}\right|\leq D
\cdot \exp{\left(-\frac{d}{n^\gamma}|l-m|\right)}.
\eeq

In addition, there exist $C >0 $ such that
\beq \label{eq:1stprop}
\frac{1}{n^\gamma}\sum_{k\in \mathbb{Z}} \left|\widehat{\varphi}^{(n,\omega)}_k\right|<C,
\eeq
and $c >0$  such that for any $m\in \mathbb{N}$
\beq \label{eq:2ndprop}
\frac{1}{n^\gamma}\sum_{k\geq m} \left|\widehat{\varphi}^{(n,\omega)}_k\right|<\exp({-cmn^{-\gamma}}).
\eeq
Also,
\beq \label{eq:3rdprop}                  \lim_{n\to\infty}\frac{1}{n^{2\gamma}}\sum_{k\in\mathbb{Z}}|k|\left|\widehat{\varphi}_k^{(n,\omega)}\right|^2<\infty,
\eeq
and $\frac{1}{n^\gamma}\frac{P_{-\omega_n}(z)}{P_{\omega_n}(z)}$ is uniformly bounded in $\partial\bbD$ $($and hence, by \eqref{eq:symbolisbdd}, $\frac{1}{n^\gamma}T\left(\frac{P_{-\omega_n}}{P_{\omega_n}}\right)$ is uniformly bounded in $\partial\bbD)$.
\end{lemma}

\begin{proof}
Using (e.g.)
\beq \label{eq:Pmin-Pplus} P_{-\omega_n}(z)=P_{\omega_n}(z)+\frac{2\omega_n}{\rho}\left(z-\rho\right),\eeq
it is not hard to verify that
$$\hat{\varphi}^{(n,\omega)}_k=\begin{cases}&\frac{1}{\left(z_{-}^{\omega_n}-z_{+}^{\omega_n}\right)\left(z_{+}^{\omega_n}\right)^{k+1}}\left(P_{-{\omega_n}}\left(z_{+}^{\omega_n}\right)\right)
=\frac{\left(P_{-{\omega_n}}\left(z_{+}^{\omega_n}\right)\right)}{\left(z_{-}^{\omega_n}-z_{+}^{\omega_n}\right)}\left(z_{+}^{\omega_n}\right)^{-|k|-1}
\qquad k>1 \\
& \frac{1}{\left(z_{-}^{\omega_n}-z_{+}^{\omega_n}\right)\left(z_{-}^{\omega_n}\right)^{k+1}}\left(P_{-{\omega_n}}\left(z_{-}^{\omega_n}\right)\right) =\frac{\left(P_{-{\omega_n}}\left(z_{-}^{\omega_n}\right)\right)}{\left(z_{-}^{\omega_n}-z_{+}^{\omega_n}\right)}\left(z_{-}^{\omega_n}\right)^{|k|-1}\qquad k<-1
\end{cases}$$

By \eqref{eq:modolusofzpm} we have  $$|z_{-}^{\omega_n}|=\left|1-\frac{\text{Re}\eta}{n^{\gamma}}\frac{\sin\left(\frac{\theta_0}{2}\right)+\sqrt{\rho^2-\cos^2\left(\frac{\theta_0}{2}\right)}}{\sqrt{\rho^2-\cos^2\left(\frac{\theta_0}{2}\right)}}\right|+o(n^{-\gamma})$$
hence there exists $\widetilde{d}>0$ (independent of $n$) so that
\beq\label{eq:zminus_modulus}
|z_{-}^{\omega_n}|\leq \exp{\left(-\frac{\widetilde{d}}{n^\gamma}\right)}\eeq
Moreover,
\beq \label{zp-zm}
\lim_{n\to  \infty}\left|z_{-}^{\omega_n}-z_{+}^{\omega_n}\right|=\lim_{n\to  \infty}\left|e^{-i\theta_0/2}\left(z_{-}^{\omega_n}-z_{+}^{\omega_n}\right)\right|=\frac{2\sqrt{\rho^2-\cos^2\left(\frac{\theta_0}{2}\right)}}{\rho}>0\eeq
and by \eqref{eq:Pmin-Pplus}
$$\lim_{n\to  \infty}\left|P_{-{\omega_n}}\left(z_{\pm}^{\omega_n}\right)\right|=\frac{2}{\rho}\left|z^\infty_\pm-\rho\right|$$
therefore there exists $D>0$ such that  $$\left|\frac{P_{-{\omega_n}}\left(z_{\pm}^{\omega_n}\right)}{\left(z_{-}^{\omega_n}-z_{+}^{\omega_n}\right)}\right|\leq D$$.

We thus get
\beq
\left|\frac{\left(P_{-{\omega_n}}\left(z_{-}^{\omega_n}\right)\right)}{\left(z_{-}^{\omega_n}-z_{+}^{\omega_n}\right)}\left(z_{-}^{\omega_n}\right)^{|k|-1}\right|\leq D\cdot \exp{\left(-\frac{\widetilde{d}|k|}{n^\gamma}\right)}.
\eeq
In addition, we note that $z_{+}^{\omega_n}z_{-}^{\omega_n}=\omega_n$ (this is the constant coefficient of $P_{\omega_n})$, hence   $|z_{+}^{\omega_n}|^{-1}=|z_{-}^{\omega_n}|+\mathcal{O}(n^{-\gamma})$, which means, using \eqref{eq:zminus_modulus} that there exists $d'>0$ (independent of $n$) such that
\beq\label{eq:zp_modulus} |z_{+}^{\omega_n}|^{-1}\leq \exp{\left(-\frac{d'}{n^\gamma}\right)} \eeq
therefore
\beq
\left|\frac{\left(P_{-{\omega_n}}\left(z_{+}^{\omega_n}\right)\right)}{\left(z_{-}^{\omega_n}-z_{+}^{\omega_n}\right)}\left(z_{+}^{\omega_n}\right)^{-|k|-1}\right|\leq D\cdot \exp{\left(-\frac{\widetilde{d}|k|}{n^\gamma}\right)}.
\eeq
Which means that we can choose $d>0$ so that
$$\left|\hat{\varphi}^{(n,\omega)}_k\right|\leq D
\cdot \exp{\left(-\frac{d}{n^\gamma}|k|\right)}$$
as desired. The rest of the properties in the statement follow instantly from this estimate since $$\left|\frac{P_{-\omega_n}(z)}{P_{\omega_n}(z)}\right|\leq\sum_{k\in \mathbb{Z}} \left|\widehat{\varphi}^{(n,\omega)}_k\right|=\mathcal{O}\left(n^{\gamma}\right)$$
and
$$\sum_{k\in \mathbb{Z}}|k|\left|\widehat{\varphi}^{(n,\omega)}_k\right|=\mathcal{O}\left(n^{2\gamma}\right).$$

As for the real part, note that the $k$'th Fourier coefficient of $\text{Re}\frac{P_{-\omega_n}(\cdot)}{P_{\omega_n}(\cdot)}$ is $\frac{1}{2}\left(\hat{\varphi}^{(n,\omega)}_k+\overline{\hat{\varphi}^{(n,\omega)}_{-k}}\right)$
so \eqref{eq:01prop} follows from \eqref{eq:FourierBound} as well.
\end{proof}

We can now prove Proposition \ref{prop:nohankel}
\begin{proof}[Proof of Proposition \ref{prop:nohankel}]
We obtain for $0<\gamma<\beta<1$
\beq \begin{split}
  &
  \left|C^{(n)}_m\left(\frac{1}{n^\gamma}T\left(\frac{P_{-\omega_n}(\cdot)}{P_{\omega_n}(\cdot)}\right)\right)-C^{(n)}_m\left(\frac{1}{n^\gamma}\left(\frac{T\left(\phi\right)+\omega_n}{T\left(\phi\right)-\omega_n}\right)\right)\right|\\
  & \quad \leq \left\|P_{\left \{ n-2mn^\beta,n+2mn^\beta
  \right\}}\left(\frac{1}{n^\gamma}T\left(\frac{P_{-\omega_n}(\cdot)}{P_{\omega_n}(\cdot)}\right)-\frac{1}{n^\gamma}\left(\frac{T\left(\phi\right)+\omega_n}{T\left(\phi\right)-\omega_n}\right)\right)P_{\left
  \{ n-2mn^\beta,n+2mn^\beta \right\}}\right\|_1 +o\left(n^{-\gamma}\right)\\
  & \quad = \frac{1}{n^\gamma}\Big\|P_{\left \{ n-2mn^\beta,n+2mn^\beta \right\}}H\left(\frac{1}{\rho\left(\cdot-z^{\omega_n}_{+}\right)}\right)\cdot H\left(\widetilde{\frac{\rho P_{-\omega_n}(\cdot)}{\cdot-z^{\omega_n}_{-}}}\right)
  P_{\left\{ n-2mn^\beta,n+2mn^\beta \right\}}\Big\|_1\\
&\qquad \quad +o\left(n^{-\gamma}\right)
\end{split}
\eeq
as $n\to \infty$.

By Cauchy-Schwartz
\beq \label{eq:hankeltracenorm} \begin{split}
    & \left\|P_{\left \{ n-2mn^\beta,n+2mn^\beta
    \right\}}H\left(\frac{1}{n^\gamma}\frac{1}{\rho\left(\cdot-z^{\omega_n}_{+}\right)}\right)H\left(\widetilde{\frac{\rho P_{-\omega_n}(\cdot)}{\cdot-z^{\omega_n}_{-}}}\right)P_{\left \{ n-2mn^\beta,n+2mn^\beta \right\}}\right\|_1^2\\
    & \quad \leq \frac{1}{n^{2\gamma}}\left\|P_{\left \{ n-2mn^\beta,n+2mn^\beta
    \right\}}H\left(\frac{1}{\rho\left(\cdot-z^{\omega_n}_{+}\right)}\right)\right\|_2^2 \cdot\left\|H\left(\widetilde{\frac{\rho P_{-\omega_n}(\cdot)}{\cdot-z^{\omega_n}_{-}}}\right)\right\|_2^2.
\end{split} \eeq

It is easy to verify that $$\frac{\rho P_{-\omega_n}(z)}{z-z^{\omega_n}_{-}}=\rho z+\rho z^{\omega_n}_{-}-(1-{\omega_n})+\rho\sum_{k=1}^\infty \left(z^{\omega_n}_{-}\right)^{k-1}P_{-{\omega_n}}\left(z^{\omega_n}_{-}\right)z^{-k}$$
so that 
\beq \begin{split}
    &\left\|H\left(\widetilde{\frac{\rho P_{-{\omega_n}}(z)}{\cdot-z^{\omega_n}_{-}}}\right)\right\|_2^2\\
    &\quad=\rho^2 \left|P_{-{\omega_n}}\left(z^{\omega_n}_{-}\right)\right|^2\sum_{k=1}^\infty k\cdot  \left|z^{\omega_n}_{-}\right|^{2(k-1)}=\rho^2
    \left|P_{-{\omega_n}}\left(z^{\omega_n}_{-}\right)\right|^2\frac{1}{\left(1-\left|z^{\omega_n}_{-}\right|^2\right)^2}. \end{split}\eeq
In addition 
$$\left\|P_{\left \{ n-2mn^\beta,n+2mn^\beta \right\}}H\left(\frac{1}{\rho\left(\cdot-z^{\omega_n}_{+}\right)}\right)\right\|^2_2\leq
    \sum_{k=1}^\infty {\frac{1}{\rho^2} k\left|z^{\omega_n}_{+}\right|^{-2\left(k+1+\left[n-2mn^\beta\right]\right)}}.$$

For the sake of convenience, from now on we denote 
\begin{equation} \label{eq:A} 
A=\sqrt{\rho^2-\cos^2\left(\frac{\theta_0}{2}\right)}.
\end{equation}
We have by \eqref{eq:modolusofzpm}
$$\frac{1}{\left(1-
\left|z^{{\omega_n}}_{\pm}\right|^2\right)^2}=\frac{1}{\left(\frac{\text{Re}\eta}{n^\gamma}\cdot\frac{\sin\left(\frac{\theta_0}{2}\right)\mp
A}{A}\right)^2+o(n^{-2\gamma})}$$ 
which means 
$$\frac{1}{n^{2\gamma}}\cdot\left\|H\left(\widetilde{\frac{\rho P_{-{\omega_n}}(\cdot)}{\cdot-z^{\omega_n}_{-}}}\right)\right\|^2_2$$ is uniformly bounded (in $n$).\\
Therefore, in order to show that
    \beq \nonumber
\frac{1}{n^\gamma}\Big\|P_{\left \{ n-2mn^\beta,n+2mn^\beta \right\}}H\left(\frac{1}{\rho\left(\cdot-z^{\omega_n}_{+}\right)}\right)
    \cdot H\left(\widetilde{\frac{\rho P_{-{\omega_n}}(\cdot)}{\cdot-z^{\omega_n}_{-}}}\right)P_{\left \{ n-2mn^\beta,n+2mn^\beta \right\}}\Big\|_1
    \eeq
    vanishes as $n\to \infty $, it suffices to prove $$\sum_{k=1}^\infty {\frac{1}{\rho^2} k\left|z^{\omega_n}_{+}\right|^{-2\left(k+1+\left[n-2mn^\beta\right]\right)}}\to
0$$ as $n\to \infty.$
To see that, using \eqref{eq:zp_modulus}, we see
$$\sum_{k=1}^\infty {\frac{1}{\rho^2}
k\left|z^{\omega_n}_{+}\right|^{-2\left(k+1+\left[n-2mn^\beta\right]\right)}}\leq\sum_{k=1}^\infty \frac{1}{\rho^2}
k\exp\left(\frac{2\widetilde{d}}{n^\gamma}\cdot \left(k+1+\left[n-2mn^\beta\right]\right)\right)$$
which by standard estimates is $\mathcal{O}\left(n^{2\gamma}\exp\left(-\widetilde{d}\cdot\left( n^{1-\gamma}-2mn^{\beta-\gamma}\right)\right)\right)$ as $n \to \infty$.
This completes the proof. \end{proof}


\subsubsection{Cumulant Analysis of the Toeplitz Operator}
We have reduced our problem to the analysis of a Toeplitz operator. In order to prove Proposition \ref{prop:PoissonForConstant} for $\alpha \neq 0$ we need to show
$$\lim_{n\rightarrow \infty}C^{(n)}_m\left(\frac{1}{n^\gamma}T\left(\text{Re}\left(\frac{P_{-{\omega_n}}}{P_{{\omega_n}}}\right)\right)\right) =\left \{ \begin{array}{cc} \sigma^2 & \textrm{ if } m=2 \\ 0 & \textrm{ if } m>2 \end{array} \right..$$ 

We write 
$$\kappa^{(n,\omega)}(z)=\sum_{k=-\infty}^\infty
\frac{\hat{\varphi}_k+\overline{\hat{\varphi}_{-k}}}{2n^\gamma}\cdot z^k$$
and we get
$$\frac{1}{n^\gamma}T\left(\text{Re}\left(\frac{P_{-{\omega_n}}}{P_{{\omega_n}}}\right)\right)=T\left(\kappa^{(n,\omega)}\right).$$

We want to use
\begin{multline}\label{eq:cumultoepl}
\exp\left(\sum_{m=2}^\infty t^m C_m \left(T \left(\kappa^{(n,\omega)}\right)\right)  \right) \\
 =\det\left(I+ P_n\left({\rm e}^{ t T\left(\kappa^{(n,\omega)}\right)}-I \right)P_n\right)  {\rm e}^{-t \Tr P_n T
 \left(\kappa^{(n,\omega)}\right)}
\end{multline}
(recall \eqref{eq:momentgenFreddet1}).

Now, let $\kappa^{(n,\omega)}_-=\sum_{k=-\infty}^{-1}
\frac{\hat{\varphi}_k+\overline{\hat{\varphi}_{-k}}}{2n^\gamma}\cdot z^k$, and $\kappa^{(n,\omega)}_+=\sum_{k=0}^{\infty}
\frac{\hat{\varphi}_k+\overline{\hat{\varphi}_{-k}}}{2n^\gamma}\cdot z^k$, and note that \eqref{eq:tophankid2}, \eqref{eq:3rdprop} imply that $$\left[T\left(\kappa^{(n,\omega)}_+\right),T\left(\kappa^{(n,\omega)}_-\right)\right]=-H(\kappa^{(n,\omega)})H(\widetilde \kappa^{(n,\omega)})$$ is trace class. Thus, by invoking \cite[Lemma 4.4]{realope} (with $\phi=\kappa^{(n,\omega)}$; note that \eqref{eq:3rdprop} says the lemma applies) we see that
\beq\begin{split}\label{eq:Toeplitztickdet1}
&\det\left(I+ P_n({\rm e}^{ t T\left(\kappa^{(n,\omega)}\right)}-I)P_n\right)  {\rm e}^{-\mathrm{Tr} P_ntT\left(\kappa^{(n,\omega)}\right)}\\
&\quad ={\rm e}^{\frac{t^2}{2}\Tr H(\kappa^{(n,\omega)})H(\widetilde \kappa^{(n,\omega)})}\det(I+Q_n (R(t,\kappa^{(n,\omega)})^{-1}-I)).
\end{split}
\eeq
 where $$R\left(t,\kappa^{(n,\omega)}\right)={\rm e}^{-t T\left(\kappa^{(n,\omega)}_{+}\right) }{\rm e}^{ t
 T\left(\kappa^{(n,\omega)}\right)}{\rm e}^{-t T\left(\kappa^{(n,\omega)}_{-}\right)}.$$ and $$Q_n=I-P_n.$$
 
The first factor on the RHS of \eqref{eq:Toeplitztickdet1} is the moment generating function of a Gaussian variable. Since the Fredholm determinant is continuous with respect to the trace norm, we thus only need to show  $$\left\|Q_n (R(t,\kappa^{(n,\omega)})^{-1}-I)\right\|_1\to 0$$ as $n\to \infty$, uniformly for $t$ in some neighborhood of $0$, and
 \beq\label{eq:trHankels}
 \lim_{n\to \infty}\text{Tr}H\left(\kappa^{(n,\omega)}\right)H\left(\widetilde \kappa^{(n,\omega)}\right)=\sigma^2=\frac{2}{\left(\eta+\overline{\eta}\right)^2}.\eeq
 
The fact that $$\left\|Q_n (R(t,\kappa^{(n,\omega)})^{-1}-I)\right\|_1\to 0$$ as $n\to \infty$ uniformly in $t$ in a neighborhood of $0$, follows from the proof of \cite[Lemma 4.5]{realope}. Indeed, replacing $\phi^{(n)}$ in that lemma with $\kappa^{(n,\omega)}$ we see we only need
 $$
 \lim_{n\to\infty}\sum_{k\in
 \mathbb{Z}}|k|\left|\widehat{\kappa}_k^{(n,\omega)}\right|^2<\infty
 $$
 and
 $$
 \sum_{k\in \mathbb{Z}}
 \left|\widehat{\kappa}^{(n,\omega)}_k\right|<C
 $$
which hold in our case by \eqref{eq:1stprop}, and \eqref{eq:3rdprop}.

We are thus left with showing that 
$$
\lim_{n\to \infty}\text{Tr}\: H\left(\kappa^{(n,\omega)}\right)H\left(\widetilde \kappa^{(n,\omega)}\right)=\sigma^2.
$$
We compute:
\beq 
\begin{split}
&\Tr H\left(\kappa^{(n,\omega)}\right)H\left(\widetilde \kappa^{(n,\omega)}\right)\\
&\quad =\frac{1}{4n^{2\gamma}}\sum_{k=1}^\infty k\cdot \left|\hat{\varphi}_k+\overline{\hat{\varphi}_{-k}}\right|^2\\
&\quad = \frac{1}{4n^{2\gamma}} \sum_{k=1}^\infty k\cdot \left|\hat{\varphi}_k\right|^2+\frac{1}{4n^{2\gamma}}\sum_{k=1}^\infty k\cdot
\left|\hat{\varphi}_{-k}\right|^2+\frac{1}{2n^{2\gamma}}\sum_{k=1}^\infty k\cdot \text{Re}\left(\hat{\varphi}_k{\hat{\varphi}_{-k}}\right).
\end{split}
\eeq 
One can see that  
\beq 
\begin{split}
& \sum_{k=1}^\infty k\cdot\hat{\varphi}_k\hat{\varphi}_{-k} \\
& \quad =\sum_{k=1}^\infty
k\cdot\frac{P_{-{\omega_n}}(z_{-}^{\omega_n})P_{-{\omega_n}}(z_{-}^{\omega_n})}{\left(z_{-}^{\omega_n}-z_{+}^{\omega_n}\right)^2{\omega_n}}\cdot\frac{\left(z_{-}^{\omega_n}\right)^{2k}}{{\omega_n}^{k}}\\
& \quad = \frac{P_{-{\omega_n}}(z_{-}^{\omega_n})P_{-{\omega_n}}(z_{-}^{\omega_n})(z_{-}^{\omega_n})^2\cdot }{\left(z_{-}^{\omega_n}-z_{+}^{\omega_n}\right)^2{\omega_n}^2}
\cdot \frac{1}{\left(1-\frac{(z_{-}^{\omega_n})^2}{{\omega_n}}\right)^2}.
\end{split} 
\eeq
Taking $n\to \infty$, we have $\frac{(z_{-}^{\omega_n})^2}{{\omega_n}}\to \frac{(z_{-}^\infty)^2}{e^{i\theta_0}}$. Recall that $\theta_0\in \supp
\mu_\alpha$, which means that $\frac{(z_{-}^\infty)^2}{e^{i\theta_0}}\not=1$ (write $\frac{(z_{-}^\infty)^2}{e^{i\theta_0}}=\left(
\frac{z_{-}^\infty}{e^{i\theta_0/2}}\right)^2$ and perform a similar computation as in \eqref{eq:exp(theta/2)}). Thus
$$
\frac{1}{2n^{2\gamma}}\sum_{k=1}^\infty k\cdot \text{Re}\left(\hat{\varphi}_k{\hat{\varphi}_{-k}}\right)=o(1)
$$ 
as $n\to \infty$, and so
\beq 
\begin{split}
&\lim_{n\to \infty}\Tr H\left(\kappa^{(n,\omega)}\right)H\left(\widetilde \kappa^{(n,\omega)}\right)\\
&\quad = \lim_{n\to \infty}\left(\frac{1}{4n^{2\gamma}} \sum_{k=1}^\infty k\cdot
\left|\hat{\varphi}_k\right|^2+\frac{1}{4n^{2\gamma}}\sum_{k=1}^\infty k\cdot \left|\hat{\varphi}_{-k}\right|^2\right).
\end{split}
\eeq 

Now, a straightforward computation shows that  
\beq 
\begin{split}
& \frac{1}{4n^{2\gamma}} \sum_{k=1}^\infty k\cdot \left|\hat{\varphi}_k\right|^2+\frac{1}{4n^{2\gamma}}\sum_{k=1}^\infty k\cdot
\left|\hat{\varphi}_{-k}\right|^2\\
& \quad =\frac{1}{4n^{2\gamma}\cdot\left|z_{-}^{\omega_n}-z_{+}^{\omega_n}\right|^2}\cdot \left(\sum_{k=1}^\infty
k\cdot\frac{\left|P_{-{\omega_n}}(z_{+}^{\omega_n})\right|^2}{\left|z_{+}^{\omega_n}\right|^{2(k+1)}}+\sum_{k=1}^\infty
k\cdot\left|P_{-{\omega_n}}(z_{-}^{\omega_n})\right|^2\cdot \left|z_{-}^{\omega_n}\right|^{2(k-1)}\right)\\
& \quad=\frac{\left|P_{-{\omega_n}}(z^{\omega_n}_+)\right|^2}{4n^{2\gamma}\cdot\left|z_{-}^{\omega_n}-z_{+}^{\omega_n}\right|^2}\cdot \left(\sum_{k=1}^\infty
    k\cdot\frac{1}{\left|z_{+}^{\omega_n}\right|^{2(k+1)}}\right)+\frac{\left|P_{-{\omega_n}}(z^{\omega_n}_-)\right|^2}{4n^{2\gamma}\cdot\left|z_{-}^{\omega_n}-z_{+}^{\omega_n}\right|^2}\cdot
    \left(\sum_{k=1}^\infty k\cdot{\left|z_{-}^{\omega_n}\right|^{2(k-1)}}\right)\\
    &\quad=\frac{\left|P_{-{\omega_n}}(z^{\omega_n}_+)\right|^2}{4n^{2\gamma}\cdot\left|z_{-}^{\omega_n}-z_{+}^{\omega_n}\right|^2}\cdot\frac{1}{\left(\left|z_{+}^{\omega_n}\right|^2-1\right)^2}+\frac{\left|P_{-{\omega_n}}(z^{\omega_n}_-)\right|^2}{4n^{2\gamma}\cdot\left|z_{-}^{\omega_n}-z_{+}^{\omega_n}\right|^2}\cdot\frac{1}{\left(\left|z_{-}^{\omega_n}\right|^2-1\right)^2}.
 \end{split}
\eeq
From \eqref{eq:Pmin-Pplus} we obtain (recall \eqref{eq:A})
$$
\lim_{n\to
\infty}\left|P_{-{\omega_n}}(z^{\omega_n}_\pm)\right|^2=\frac{4}{\rho^2}\left|z^\infty_\pm-\rho\right|^2=
\frac{4}{\rho^2}\left|e^{i\frac{-\theta_0}{2}}\left(z^\infty_\pm-\rho\right)\right|^2=\frac{4}{\rho^2}\left(A\mp\sin\left(\frac{\theta_0}{2}\right)\right)^2,
$$
so by \eqref{eq:modolusofzpm} and \eqref{zp-zm}, we see that 
\beq \no
\begin{split}
    &\frac{\left|P_{-{\omega_n}}(z^{\omega_n}_+)\right|^2}{4n^{2\gamma}\cdot\left|z_{-}^{\omega_n}-z_{+}^{\omega_n}\right|^2}\cdot\frac{1}{\left(\left|z_{+}^{\omega_n}\right|^2-1\right)^2}+\frac{\left|P_{-{\omega_n}}(z^{\omega_n}_-)\right|^2}{4n^{2\gamma}\cdot\left|z_{-}^{\omega_n}-z_{+}^{\omega_n}\right|^2}\cdot\frac{1}{\left(\left|z_{-}^{\omega_n}\right|^2-1\right)^2}\\
 &\quad =\frac{\left|P_{-{\omega_n}}(z^{\omega_n}_+)\right|^2}{4\left|z_{-}^{\omega_n}-z_{+}^{\omega_n}\right|^2}\cdot
 \frac{1}{\frac{\left(\text{Re}\eta\right)^2\left(A-\sin{\left(\frac{\theta_0}{2}\right)}\right)^2}{A^2}+o(1)}+\frac{\left|P_{-{\omega_n}}(z^{\omega_n}_-)\right|^2}{4\left|z_{-}^{\omega_n}-z_{+}^{\omega_n}\right|^2}\cdot\frac{1}{\frac{\left(\text{Re}\eta\right)^2\left(A+\sin{\left(\frac{\theta_0}{2}\right)}\right)^2}{A^2}+o(1)}\\
 &\quad =\frac{\frac{4}{\rho^2}\left(A-\sin\left(\frac{\theta_0}{2}\right)\right)^2}{\frac{16A^2}{\rho^2}}\cdot
 \frac{1}{\frac{\left(\text{Re}\eta\right)^2\left(A-\sin{\left(\frac{\theta_0}{2}\right)}\right)^2}{A^2}}+\frac{\frac{4}{\rho^2}\left(A+\sin\left(\frac{\theta_0}{2}\right)\right)^2}{\frac{16A^2}{\rho^2}}\cdot\frac{1}{\frac{\left(\text{Re}\eta\right)^2\left(A+\sin{\left(\frac{\theta_0}{2}\right)}\right)^2}{A^2}}+o(1)\\
 & \qquad \xrightarrow[n\to\infty]{} \frac{1}{2\cdot \left(\text{Re}\eta\right)^2}=\frac{2}{\left(\eta+\overline{\eta}\right)^2}.
\end{split}
\eeq
Combining the above computations, we conclude 
\beq \label{eq:limitofvar}
\lim_{n \to \infty}\text{Tr}H\left(\kappa^{(n,{\omega})}\right)H\left(\widetilde \kappa^{(n,\omega)}\right)=\sigma^2
\eeq

We summarize the discussion above in the 
\begin{proof}[Proof of Proposition \ref{prop:PoissonForConstant}, the case $\alpha \neq 0$]

Recall that
$$
C^{(n)}_m\left(X^{\theta_0}_{\Psi_{n,\omega}}\right)=C^{(n)}_m\left(\frac{1}{n^\gamma}\left(\text{Re}\frac{G+{\omega_n}}{G-{\omega_n}}\right)\right)
$$
and by Proposition \ref{prop:GGT2Toeplitz} and Proposition \ref{prop:nohankel} we have 
$$
\lim_{n\to\infty}C^{(n)}_m\left(\frac{1}{n^\gamma}\left(\text{Re}\frac{G+{\omega_n}}{G-{\omega_n}}\right)\right)
=\lim_{n\to\infty}C^{(n)}_m\left(\frac{1}{n^\gamma}\left(\text{Re}\frac{T\left(\phi\right)+{\omega_n}}{T\left(\phi\right)-{\omega_n}}\right)\right)
=\lim_{n\to\infty}C_m\left(\frac{1}{n^\gamma}T\left(\text{Re}\frac{P_{-{\omega_n}}(z)}{P_{{\omega_n}}(z)}\right)\right).
$$
But the computations following \eqref{eq:Toeplitztickdet1} show that
\beq \no 
\begin{split}
\lim_{n\to\infty}\det\left(I+ P_n\left({\rm e}^{ t T\left(\kappa^{(n,\omega)}\right)}-I \right)P_n\right)  {\rm e}^{-t \Tr P_n T
\left(\kappa^{(n,\omega)}\right)}=\exp\left(\frac{t^2}{2}\cdot \frac{2}{\left(\eta+\overline{\eta}\right)^2}\right)
\end{split}
\eeq
which means, by \eqref{eq:momentgenFreddet1}, that 
\beq \no
\lim_{n\to\infty}C^{(n)}_m\left(\frac{1}{n^\gamma}T\left(\text{Re}\frac{P_{-{\omega_n}}(z)}{P_{{\omega_n}}(z)}\right)\right)=\left \{\begin{array}{cc}
\frac{2}{\left(\eta+\overline{\eta}\right)^2} &\quad m=2 \\
0 &\quad m>2 \end{array} \right ..
\eeq
Thus, we have
$$
X^{\theta_0}_{\Psi_{n,\omega}}-\bbE
X^{\theta_0}_{\Psi_{n,\omega}}\xrightarrow[n\to \infty]{\mathcal{D}}\mathcal{N}\left(0,\frac{2}{\left(\eta+\overline{\eta}\right)^2}\right),
$$
which concludes the proof.
\end{proof}


\subsection{The case $\alpha=0$.\\}
As mentioned above The case of $\alpha_n \equiv 0$ corresponds to $d\mu(\theta)=d\theta$, in which case the OPE is the eigenvalue process of the CUE. Mesoscopic fluctuations for Schwartz functions were already studied in this case by Soshnikov \cite{soshnikov}. However, Soshnikov's work does not contain a CLT for the linear statistic $X_{\Psi_{n,\gamma,\eta}}^{\theta_0}$, so for completeness we consider it here.

We want to use a strategy similar to the one we used in the $\alpha \neq 0$ case, but by the discussion in Section 2.1 the orthogonal polynomials are not dense in $L_2(d\theta)$. Thus, we use the CMV basis in order to get the matrix representation \eqref{eq:cmvcumulant} for the cumulants. 

Recall $\omega_n=\omega_n(\eta, \gamma, \theta_0)$ was defined in \eqref{eq:omega_Definition}. Letting $\mathcal{C}$ be the corresponding CMV matrix, a direct computation shows that
\begin{equation} \label{eq:CMVLinearStatistic}
\begin{split}
\frac{1}{n^\gamma}\textrm{Re}\frac{\mathcal{C}+\omega_n}{\mathcal{C}-\omega_n}&=\frac{1}{n^\gamma}\begin{pmatrix}
1&\omega_n&\overline{\omega_n}&\omega_n^2&\overline{\omega_n}^2&\omega_n^3&\overline{\omega_n}^3&\omega_n^4&\ldots\\
\overline{\omega_n}&1&\overline{\omega_n}^2&\omega_n&\overline{\omega_n}^3&\omega_n^2&\overline{\omega_n}^4&\omega_n^3&\ldots\\
\omega_n&\omega_n^2&1&\omega_n^3&\overline{\omega_n}&\omega_n^4&\overline{\omega_n}^2&\omega_n^5&\ldots\\
\overline{\omega_n}^2&\overline{\omega_n}&\overline{\omega_n}^3&1&\overline{\omega_n}^4&\omega_n&\overline{\omega_n}^5&\omega_n^2&\ldots\\
\omega_n^2&\omega_n^3&\omega_n&\omega_n^4&1&\omega_n^5&\overline{\omega_n}&\omega_n^6&\ldots\\
\overline{\omega_n}^3&\overline{\omega_n}^2&\ \overline{\omega_n}^4&\overline{\omega_n}&\overline{\omega_n}^5&1&\overline{\omega_n}^6&\omega_n&\ldots\\
\omega_n^3&\omega_n^4&\omega_n^2&\omega_n^5&\omega_n&\omega_n^6&1&\omega_n^7&\ldots\\
\overline{\omega_n}^4&\overline{\omega_n}^3&\overline{\omega_n}^5&\overline{\omega_n}^2&\overline{\omega_n}^6&\overline{\omega_n}&\overline{\omega_n}^7&1&\ldots \\
\vdots&\vdots&\vdots&\vdots&\vdots&\vdots&\vdots&\vdots&\ddots
\end{pmatrix} \\
&=\mathcal{T}^{(n)}+\mathcal{H}^{(n)},
\end{split}
\end{equation}
where
\begin{equation} \nonumber
\begin{split}
\mathcal{T}_n  = \frac{1}{n^\gamma}\begin{pmatrix}
1&0&\overline{\omega_n}&0&\overline{\omega_n}^2&0&\overline{\omega_n}^3&0&\ldots\\
0&1&0&\omega_n&0&\omega_n^2&0&\omega_n^3&\ldots\\
\omega_n&0&1&0&\overline{\omega_n}&0&\overline{\omega_n}^2&0&\ldots\\
0&\overline{\omega_n}&0&1&0&\omega_n&0&\omega_n^2&\ldots\\
\omega_n^2&0&\omega_n&0&1&0&\overline{\omega_n}&0&\ldots\\
0&\overline{\omega_n}^2&0&\overline{\omega_n}&0&1&0&\omega_n&\ldots\\
\omega_n^3&0&\omega_n^2&0&\omega_n&0&1&0&\ldots\\
0&\overline{\omega_n}^3&0&\overline{\omega_n}^2&0&\overline{\omega_n}&0&1&\ldots \\
\vdots&\vdots&\vdots&\vdots&\vdots&\vdots&\vdots&\vdots&\ddots
\end{pmatrix},\\
\end{split}
\eeq
and
\beq \nonumber
\begin{split}
&\mathcal{H}_n=\frac{1}{n^\gamma}\begin{pmatrix}
0&\omega_n&0&\omega_n^2&0&\omega_n^3&0&\omega_n^4&\ldots\\
\overline{\omega_n}&0&\overline{\omega_n}^2&0&\overline{\omega_n}^3&0&\overline{\omega_n}^4&0&\ldots\\
0&\omega_n^2&0&\omega_n^3&0&\omega_n^4&0&\omega_n^5&\ldots\\
\overline{\omega_n}^2&0&\overline{\omega_n}^3&0&\overline{\omega_n}^4&0&\overline{\omega_n}^5&0&\ldots\\
0&\omega_n^3&0&\omega_n^4&0&\omega_n^5&0&\omega_n^6&\ldots\\
\overline{\omega_n}^3&0&\ \overline{\omega_n}^4&0&\overline{\omega_n}^5&0&\overline{\omega_n}^6&0&\ldots\\
0&\omega_n^4&0&\omega_n^5&0&\omega_n^6&0&\omega_n^7&\ldots\\
\overline{\omega_n}^4&0&\overline{\omega_n}^5&0&\overline{\omega_n}^6&0&\overline{\omega_n}^7&0&\ldots \\
\vdots&\vdots&\vdots&\vdots&\vdots&\vdots&\vdots&\vdots&\ddots
\end{pmatrix} \\
\end{split}
\end{equation}
so $\mathcal{T}^{(n)}$ is a $2\times 2$ Block Toeplitz Matrix (see Section 2) and $\mathcal{H}^{(n)}$ is a $2\times 2$ Block Hankel Matrix. 

We denote
$$\Xi_+=\begin{bmatrix}
\omega_n & 0 \\
0 & \overline{\omega_n}
\end{bmatrix}, \quad 
\Xi_-=\begin{bmatrix}
\overline{\omega_n} & 0 \\
0 & \omega_n
\end{bmatrix}, \quad \textbf{0}=\begin{bmatrix}
0& 0 \\
0 & 0
\end{bmatrix},
$$
and 
$$\widehat{\phi}^{(n,\alpha,\eta,\theta_0)}_j=   \left\{
\begin{array}{ll}
      \frac{1}{n^{\alpha}}\left(\Xi_{-}\right)^{-j}=\frac{1}{n^\gamma}\begin{bmatrix}
\overline{\omega_n}^{-j} & 0 \\
0 & \omega_n^{-j} \\
\end{bmatrix} & j\leq 0 \\
\quad\\
     \frac{1}{n^{\alpha}}\left(\Xi_{+}\right)^j= \frac{1}{n^\gamma}\begin{bmatrix}
\omega_n^j & 0 \\
0 & \overline{\omega_n}^j
\end{bmatrix} & j> 0 \\
\end{array}
\right.
$$
and the matrix valued symbols
\begin{equation} \nonumber
\phi^{(n,\alpha,\eta,\theta_0)}_+(\theta)=\sum_{j=1}^\infty
\widehat{\phi}_j^{(n,\alpha,\eta,\theta_0)}e^{ij\theta}, \quad 
\phi^{(n,\alpha,\eta,\theta_0)}_-(\theta)=\sum_{j=0}^{\infty}
\widehat{\phi}_{-j}^{(n,\alpha,\eta,\theta_0)}e^{-ij\theta},
\end{equation}
and
$$
\phi^{(n,\alpha,\eta,\theta_0)}=\phi^{(n,\alpha,\eta,\theta_0)}_++\phi^{(n,\alpha,\eta,\theta_0)}_-.
$$
Thus
\begin{equation} \nonumber
\mathcal{T}^{(n)}=T\left(\phi^{(n)}\right)=T\left(\phi^{(n)}_-\right)+T\left(\phi^{(n)}_+\right),
\end{equation}
and 
$$\mathcal{H}^{(n)}=H\left(\phi^{(n)} \times \begin{pmatrix}
0&1\\
1&0
\end{pmatrix} \right),
$$
where we have omitted the $(\alpha,\eta,\theta_0)$ notation for simplicity as we shall do henceforth. 

\begin{lemma}\label{lem:3lemma}
The following properties hold:
\begin{enumerate}

\item There exists $C_\eta \in \mathbb{R}_{>0} $ such that
$\sum_{k\in \mathbb{Z}} \left\|\widehat{\phi}^{(n)}_k\right\|_1<C_\eta$. 
In particular, the sequence of functions $\left\{\phi^{(n)} \right \}$ is uniformly bounded on the unit circle.

\item There exist $ d, \widetilde{d} \in \mathbb{R}_{>0}$ such that
$\left\|\widehat{\phi}^{(n)}_k\right\|_1<\frac{\widetilde{d}}{n^\gamma}\exp({-d|k|n^{-\gamma}})$. In particular there exist constants $C,\widetilde{d}\in
\mathbb{R}$ which are independent of n,r,s such that  
$$\left|\left(\frac{1}{n^\gamma}\textrm{Re}\frac{\mathcal{C}+\omega_n}{\mathcal{C}-\omega_n} \right)_{r,s}\right|\leq
Ce^{\widetilde{d}\frac{|r-s|}{n^\gamma}}.$$
\item There exist $ D_1,D_2 \in \mathbb{R}_{>0}$  such that  for any $m\in \mathbb{N}$, 
$$\sum_{k\geq m}
    \left\|\widehat{\phi}^{(n)}_k\right\|_1<D_1\exp({-D_2mn^{-\gamma}}).$$
\item  $\lim_{n\to\infty}\sum_{k\in \mathbb{Z}}|k|\left\|\widehat{\phi}_k^{(n)}\right\|^2_1=\frac{2}{(\text{Re}(\eta))^2}.$
\end{enumerate}
\end{lemma}

\begin{proof}
\begin{enumerate}
    \item We write (for $n$ sufficiently large so that $|\omega_n|<1$)
\begin{equation*}
    \begin{split}
&\quad \sum_{k\in \mathbb{Z}} \left\|\widehat{\phi}^{(n)}_k\right\|_1 \\
&\qquad \leq \frac{4}{n^\gamma}\sum_{k\geq 0}
|\omega_n|^k=\frac{4}{n^\gamma}\left(\frac{1}{1-|\omega_n|}\right)=\frac{4}{n^\gamma}\left(\frac{1+|\omega_n|}{1-|\omega_n|^2}\right) \\
& \qquad =\frac{4}{n^\gamma}\left(\frac{2+\mathcal{O}(n^{-\gamma})}{
    2\frac{\text{Re}(\eta)}{n^\gamma}+o(n^{-\gamma})}\right)=\frac{4 +\mathcal{O}(n^{-\gamma})}{\text{Re}(\eta)+o(1)}<C_\eta
\end{split}
\end{equation*}
which is independent of $n$ (but $\eta$ dependent). 
\item
Since
$\left\|\widehat{\phi}^{(n)}_k\right\|_1=\frac{1}{n^\gamma}2|\omega_n|^{|k|},$
the claim follows from the fact that  $$|\omega_n|<1-\frac{d_1}{n^\gamma}<\exp(-d_2n^{-\gamma})$$ for some $n$ independent constants $d_1,d_2>0$ (recall $\text{Re}\eta>0$).
The second part follows immediately from \eqref{eq:CMVLinearStatistic}.
\item This follows immediately from (2) of this lemma.
\item
We compute: $$\left\|\widehat{\phi}_{k}^{(n)}\right\|^2_1=\frac{4}{n^{2\alpha}}|\omega_n|^{2|k|}.$$
For $n$ sufficiently large so that $|\omega_n|<1$, we get
\begin{equation*}
 \begin{split}
&\quad \sum_{k\in \mathbb{Z}}|k|\left\|\widehat{\phi}_k^{(n)}\right\|^2_1 \\
&\qquad =\sum_{k\in \mathbb{Z}}|k|\frac{4}{n^{2\alpha}}\left(|\omega_n|^2\right)^{|k|}=\frac{8|\omega_n|^2}{n^{2\alpha}}\sum_{k\geq
1}k\left(|\omega_n|^2\right)^{k-1}=
    \frac{8|\omega_n|^2}{n^{2\alpha}}\left(\frac{1}{(1-|\omega_n|^2)^2}\right) \\
& \qquad
=\frac{8\left[\left(1-\frac{\text{Re}(\eta)}{n^\gamma}\right)^2+\left(\frac{\text{Im}(\eta)}{n^\gamma}\right)^2\right]}{n^{2\alpha}}\left(\frac{1}{\left(1-\left(1-\frac{\text{Re}(\eta)}{n^\gamma}\right)^2-\left(\frac{\text{Im}(\eta)}{n^\gamma}\right)^2\right)^2}\right)\\
& \qquad
=\frac{8\left[\left(1-\frac{\text{Re}(\eta)}{n^\gamma}\right)^2+\left(\frac{\text{Im}(\eta)}{n^\gamma}\right)^2\right]}{n^{2\alpha}}\left(\frac{1}{\left(2\frac{\text{Re}(\eta)}{n^\gamma}-\left(\frac{\text{Re}(\eta)}{n^\gamma}\right)^2-\left(\frac{\text{Im}(\eta)}{n^\gamma}\right)^2\right)^2}\right)\\
& \qquad =
8\left[\left(1-\frac{\text{Re}(\eta)}{n^\gamma}\right)^2+\left(\frac{\text{Im}(\eta)}{n^\gamma}\right)^2\right]\frac{1}{\left(2\text{Re}(\eta)-\frac{\text{Re}(\eta)^2}{n^\gamma}-\frac{\text{Im}(\eta)^2}{n^\gamma}\right)^2}.
\end{split}
\end{equation*}
and the result follows from taking $n\rightarrow \infty$.
 \end{enumerate}
\end{proof}

Lemma \ref{lem:3lemma} now says that we may use Proposition \ref{th:comparison-general} with $\beta=\gamma+\epsilon$ (where $0<\epsilon<1-\gamma$) to deduce for all $m\geq 2$
\begin{equation} \label{eq:reducingtotop}
\begin{split}
& \left|C_m^{(n)} \left(X^{\theta_0}_{\Psi_{n,\alpha,\eta}} \right) - C_m^{(n)} \left( T\left(\phi^{(n)}\right) \right) \right| \\
&\quad \leq C(m,\beta) \left \| P_{\left \{n-2mn^\beta,n+2mn^\beta \right\}} \left( H\left(\phi^{(n)} \times \begin{pmatrix}
0&1\\
1&0
\end{pmatrix} \right)\right) P_{\left \{n-2mn^\beta,n+2mn^\beta\right \}} \right \|_1+o(1)\\
& \quad \leq \widetilde{C}(m)n^{2\beta}e^{d_2n^\epsilon}e^{-d_2n^{1-\gamma}}+o(1)=o(1).
\end{split}
\end{equation}

Having reduced the analysis to that of a block Toeplitz operator, the proof follows essentially the same lines of \cite[Section 4]{realope}. We sketch the procedure beginning with a fundamental
\begin{lemma}
The commutator $\left[T\left(\phi^{(n)}_{+}\right) ,T\left(\phi^{(n)}_{-}\right) \right]$ is a trace class operator with a uniform (in $n$) bound on its trace.
\end{lemma}
\begin{proof}
We note that both $\phi^{(n)}_{+}$ and $\phi^{(n)}_{-}$ are matrix valued functions whose matrix valued Fourier coefficients $\widehat{\phi}_{k}^{(n)}$, are diagonal matrices. Thus $$\phi^{(n)}_{+}\cdot\phi^{(n)}_{-}=\phi^{(n)}_{-}\cdot\phi^{(n)}_{+}$$ so
$$T\left(\phi^{(n)}_{+}\cdot\phi^{(n)}_{-}\right)=T\left(\phi^{(n)}_{-}\cdot\phi^{(n)}_{+}\right).$$
This, together with \eqref{eq:tophankid} leads us to
$$\left[T\left(\phi^{(n)}_{+}\right) ,T\left(\phi^{(n)}_{-}\right) \right]=-H\left(\phi^{(n)}\right)H\left(\widetilde{\phi}^{(n)}\right).$$
which implies
\begin{equation*}
\begin{split}
 \left\|\left[T\left(\phi^{(n)}_{+}\right) ,T\left(\phi^{(n)}_{-}\right) \right]\right\|_1 & \leq  \left\|-H\left(\phi^{(n)}\right)\right\|_2\left\|H\left(\widetilde{\phi}^{(n)}\right)\right\|_2\\
 &\quad=\left(\sum_{k=1}^\infty k\left\|\widehat{\phi}^{(n)}_{-k}\right\|_2^2\right)^{1/2}\left(\sum_{k=1}^\infty
 k\left\|\widehat{\phi}^{(n)}_{k}\right\|_2^2\right)^{1/2}
\end{split}
\end{equation*}
which, by (4) of Lemma \ref{lem:3lemma}, is uniformly bounded. 
\end{proof}

As in the argument below \eqref{eq:cumultoepl}, we consider the moment generating function and note that by \eqref{eq:momentgenFreddet1}, Lemma \ref{lem:3lemma} , and \cite[Lemma~4.4]{realope} (note that although we are dealing with block Toeplitz matrices, the blocks here are diagonal so all computations follow as in \cite[Section 4]{realope})
\begin{multline}\label{eq:boroku}
\exp\left(\sum_{m=2}^\infty t^m C^{(n)}_m \left(T \left(\phi^{(n)}\right)\right)  \right) \\
 =\det\left(I+ P_n\left({\rm e}^{ t T\left(\phi^{(n)}\right)}-I \right)P_n\right)  {\rm e}^{-t \Tr P_n T \left(\phi^{(n)}\right)}\\
={\rm e}^{\frac{t^2}{2}\Tr
 H(\phi^{(n)})H(\widetilde \phi^{(n)})}\det(I+Q_n (R(t,\phi^{(n)})^{-1}-I))
\end{multline}
 where $$R(t,\phi^{(n)})={\rm e}^{-t T(\phi^{(n)}_{+} ) }{\rm e}^{ t T(\phi^{(n)})}{\rm e}^{-t T(\phi^{(n)}_{-} ) }$$
 and $$Q_n=I-P_n.$$

Note that by \cite[Lemma~4.3]{realope}
\begin{multline}\label{eq:expexpand}
    {\rm e}^{-t T(\phi^{(n)}_{+} ) }{\rm e}^{ t T(\phi^{(n)})}{\rm e}^{-t T(\phi^{(n)}_{-} ) }-I \\
    \qquad= \sum_{m_1,m_2,m_3=0}^\infty \sum_{j=0}^{m_2-1} \frac{(-1)^{m_1+m_3} T(\phi^{(n)}_{+}) ^{m_1} (T(\phi^{(n)} )^j [T(\phi^{(n)}_{+})
    ,T(\phi^{(n)}_{-}) ] (T(\phi^{(n)} )^{m_2-j-1} T(\phi^{(n)}_{-}) ^{m_3}}{m_1! m_2! m_3! (m_1+m_2+m_3+1)},
\end{multline}
and since $\left[T\left(\phi^{(n)}_{+}\right) ,T\left(\phi^{(n)}_{-}\right) \right]$ is a trace class operator, so is
${\rm e}^{-t T\left(\phi^{(n)}_{+} \right) }{\rm e}^{ t T\left(\phi^{(n)}\right)}{\rm e}^{-t T\left(\phi^{(n)}_{-} \right) }-I $.

As above, by the continuity of the Fredholm determinant with respect to the trace norm, convergence to the moment generating function to that of a normal random variable now follows from
\begin{prop} \label{prop:qnto0}
$$\|Q_n (R(t,\phi^{(n)})^{-1}-I)\|_1\to 0 \text{ as } n\to \infty$$
\end{prop}
\begin{proof}
The proof uses the same strategy of the proof of \cite[Lemma 4.5]{realope}. We sketch the main points. 
By using the expansion in \eqref{eq:expexpand} we get
\begin{equation*}
\begin{split}
&\left \|Q_n\left({\rm e}^{-t T\left(\phi_+^{(n)}\right ) }{\rm e}^{t T\left(\phi^{(n)}\right))}{\rm e}^{-t
T\left(\phi^{(n)}_{-}\right)}-I\right)\right \|_1\\
&\quad=\sum_{m_1,m_2,m_3=0}^\infty \sum_{j=0}^{m_2-1} \frac{(-1)^{m_1+m_3} t^{m_1+m_2+m_3+1}  }{m_1! m_2! m_3! (m_1+m_2+m_3+1)}\\
&\qquad\times \left \| Q_n T\left(\phi_+^{(n)}\right)^{m_1}T\left(\phi^{(n)}\right)^j
\left[T\left(\phi_+^{(n)}\right),T\left(\phi_-^{(n)}\right)\right] T\left(\phi^{(n)}\right)^{m_2-j-1}
T\left(\phi_-^{(n)}\right)^{m_3}\right\|_1.
\end{split}
\end{equation*}
By Lemma \ref{lem:3lemma}
\begin{equation}\label{eq:dominatedconv}
\begin{split}
   &\left \| Q_n T\left(\phi_+^{(n)}\right)^{m_1}T\left(\phi^{(n)}\right)^j \left[T\left(\phi_+^{(n)}\right),T\left(\phi_-^{(n)}\right)\right]
   T\left(\phi^{(n)}\right)^{m_2-j-1} T\left(\phi_-^{(n)}\right)^{m_3}\right\|_1 \\
    &\quad \leq \left \| Q_n \right \|_\infty  \left\|T\left(\phi_+^{(n)}\right)\right \|_\infty^{m_1} \left \|T\left(\phi^{(n)}\right)\right
    \|_\infty^j   \left\|\left[T\left(\phi_+^{(n)}\right),T\left(\phi_-^{(n)}\right)\right]\right\|_1 \left \|T\left(\phi^{(n)}\right)\right
    \|_\infty^{m_2-j-1} \left \|T\left(\phi_-^{(n)}\right)\right \|_\infty^{m_3} \\
& \quad \leq C^{m_1+m_2+m_3}
\end{split}
\end{equation}
for some constant $C>0$. Hence, by the dominated convergence theorem, in order to prove 
$$\lim_{n\to \infty}\|Q_n (R(t,\phi^{(n)})^{-1}-I)\|_1=0,$$ 
it is sufficient to show that for all $\{m_1,m_2,m_3\}\subseteq \mathbb{N}\cup \{0\}$ and \mbox{$j\in \{1,\ldots,m_2-1\}$}
    $$\lim_{n\to \infty}\left \| Q_n T\left(\phi_+^{(n)}\right)^{m_1}T\left(\phi^{(n)}\right)^j
    \left[T\left(\phi_+^{(n)}\right),T\left(\phi_-^{(n)}\right)\right] T\left(\phi^{(n)}\right)^{m_2-j-1}
    T\left(\phi_-^{(n)}\right)^{m_3}\right\|_1 =0.$$

    Let $\gamma<\beta<1$, for $n$ such that   $n^\beta(m_1+j)<n$ we split $Q_n T\left(\phi_+^{(n)}\right)$ into $$Q_n
    T\left(\phi_+^{(n)}\right)=Q_n T\left(\phi_+^{(n)}\right)P_{n-{n^\beta}}+Q_n T\left(\phi_+^{(n)}\right)Q_{n-n^\beta}$$
    where we abuse the notation and write for $M\in \mathbb{R}$, $P_M=P_{\lceil M \rceil}$.\\
    We now may use Lemma \ref{lem:3lemma} to deduce
    $$\left\|Q_n T\left(\phi_+^{(n)}\right)P_{n-n^\beta}\right\|_\infty\leq \sum_{\frac{n^\beta}{2} \leq k}
    \left\|\widehat{\phi}^{(n)}_k\right\|_1\leq D_1 \exp\left(-\frac{D_2}{2}n^{\beta-\gamma}\right) $$
    and thus
    \begin{equation*}
    \begin{split}
       & \left \| Q_n T\left(\phi_+^{(n)}\right)^{m_1}T\left(\phi^{(n)}\right)^j
       \left[T\left(\phi_+^{(n)}\right),T\left(\phi_-^{(n)}\right)\right] T\left(\phi^{(n)}\right)^{m_2-j-1}
       T\left(\phi_-^{(n)}\right)^{m_3}\right\|_1\\
        &\quad\leq C\left \|Q_{n-n^\beta}T\left(\phi_+^{(n)}\right)^{m_1-1}T\left(\phi^{(n)}\right)^j
         \left[T\left(\phi_+^{(n)}\right),T\left(\phi_-^{(n)}\right)\right] T\left(\phi^{(n)}\right)^{m_2-j-1}
         T\left(\phi_-^{(n)}\right)^{m_3}\right\|_1\\
         &\qquad+\mathcal{O}\left(\exp\left(-\frac{D}{2}n^{\beta-\gamma}\right)\right).
     \end{split}
    \end{equation*}
    We iterate this procedure and use standard inequalities as above to obtain
    \begin{equation}
    \begin{split}
    & \left \| Q_n T\left(\phi_+^{(n)}\right)^{m_1}T\left(\phi^{(n)}\right)^j
    \left[T\left(\phi_+^{(n)}\right),T\left(\phi_-^{(n)}\right)\right] T\left(\phi^{(n)}\right)^{m_2-j-1}
    T\left(\phi_-^{(n)}\right)^{m_3}\right\|_1\\
&\leq \quad C^{m_1+m_2+m_3-1}\left \|Q_{n-(m_1+j)n^\beta} \left[T\left(\phi_+^{(n)}\right),T\left(\phi_-^{(n)}\right)\right]\right\|_1\\
    &\quad+\mathcal{O}\left(\exp\left(-\frac{D}{2}n^{\beta-\gamma}\right)\right).\\
    \end{split}
\end{equation}

It is therefore sufficient to prove that
$$\lim _{n\to \infty}\left \|Q_{n-(m_1+j)n^\beta} \left[T\left(\phi_+^{(n)}\right),T\left(\phi_-^{(n)}\right)\right]\right\|_1=0.$$
For this, write
\begin{equation*}
\begin{split}
        &\left\|Q_{n-(m_1+j)n^\beta}\left[T\left(\phi_+^{(n)}\right),T\left(\phi_-^{(n)}\right)\right]\right\|_1 \leq\left\|Q_{n-(m_1+j)n^\beta}H\left(\phi^{(n)}\right)\right\|_2\left\|H\left(\widetilde{\phi}^{(n)}\right)\right\|_2,
\end{split}
\end{equation*}
and note that 
$$\left\|H\left(\widetilde{\phi}^{(n)}\right)\right\|^2_2=\sum_{k\in \mathbb{N}}|k|\left\|\widehat{\phi}_k^{(n)}\right\|^2_2$$
is bounded by Lemma \ref{lem:3lemma}.
Finally, by the estimates in the proof of that lemma,
$$\left\|Q_{n-(m_1+j)n^\beta}H\left(\phi^{(n)}\right)\right\|^2_2 \leq \sum_{h=1}^\infty h \left
\|\widehat{\phi}^{(n)}_{\frac{n-1-(m_1+j)n^\beta}{2}+h}\right\|^2_2\rightarrow 0$$ 
 as $n\to \infty$. This completes the proof of the proposition.
\end{proof}

\begin{proof}[Proof of Proposition \ref{prop:PoissonForConstant}, the case $\alpha= 0$]
From \eqref{eq:boroku} and Proposition \ref{prop:qnto0} we have uniformly for $t$ in a neighborhood of 0
\begin{equation} \nonumber
 \lim_{n\to \infty}\det\left(I+ P_n({\rm e}^{ t T(\phi^{(n)})}-I)P_n\right)  {\rm e}^{-\mathrm{Tr}
    P_ntT(\phi^{(n)})}=\lim_{n \to \infty}{\rm e}^{\frac{t^2}{2}\Tr
 H(\phi^{(n)})H(\widetilde \phi^{(n)})}.
\end{equation}

Recall
$$H\left(\phi^{(n)}\right)=\begin{pmatrix}
\widehat{\phi}^{(n)}_1 & \widehat{\phi}^{(n)}_2 & \widehat{\phi}^{(n)}_3&\ldots\\
\widehat{\phi}^{(n)}_2 & \widehat{\phi}^{(n)}_3 & \widehat{\phi}^{(n)}_4&\ldots\\
\widehat{\phi}^{(n)}_3 & \widehat{\phi}^{(n)}_4 & \widehat{\phi}^{(n)}_5&\ldots\\
\vdots&\vdots&\vdots&\ddots
\end{pmatrix} \quad H\left(\widetilde{\phi}^{(n)}\right)=\begin{pmatrix}
\widehat{\phi}^{(n)}_{-1} & \widehat{\phi}^{(n)}_{-2} & \widehat{\phi}^{(n)}_{-3}&\ldots\\
\widehat{\phi}^{(n)}_{-2} & \widehat{\phi}^{(n)}_{-3} & \widehat{\phi}^{(n)}_{-4}&\ldots\\
\widehat{\phi}^{(n)}_{-3} & \widehat{\phi}^{(n)}_{-4} & \widehat{\phi}^{(n)}_{-5}&\ldots\\
\vdots&\vdots&\vdots&\ddots
\end{pmatrix}$$
so by a straightforward computation
\begin{equation*}
    \begin{split}
&\text{Tr}\left(H\left(\phi^{(n)}\right)H\left(\widetilde{\phi}^{(n)}\right)\right)\\
&\quad =\sum_{j=1}^\infty\text{Tr}\left(\sum_{k=0}^\infty\widehat{\phi}^{(n)}_{k+j}\widehat{\phi}^{(n)}_{-k-j}\right)\\
&\quad=\frac{2|\omega_n|^2}{n^{2\alpha}\left(1-|\omega_n|^2\right)^2}=\frac{2\left(\left(1-\frac{\text{Re}(\eta)}{n^\gamma}\right)^2+\left(\frac{\text{Im}(\eta)}{n^\gamma}\right)^2\right)}{\left(2\text{Re}(\eta)-\frac{(\text{Re}(\eta))^2}{n^\gamma}-\frac{(\text{Im}(\eta))^2}{n^\gamma}\right)^2}\xrightarrow[n\to
\infty]{} \frac{2}{\left(\eta+\overline{\eta}\right)^2}.
  \end{split}
\end{equation*}
Combining the above with  \eqref{eq:reducingtotop} implies that
$$\lim_{n\to \infty}C^{(n)}_m\left(X^{\theta_0}_{\Psi_{n,\alpha,\eta}}\right)=\lim_{n\to \infty}C^{(n)}_m\left(T\left(\phi^{(n)}\right)\right)=\left\{
\begin{array}{ll}
      \frac{2}{(\eta+\overline{\eta})^2} & m= 2 \\
\quad\\
      0 & m> 2 \\
\end{array}
\right.$$
which yields
$$X^{\theta_0}_{\Psi_{n,\alpha,\eta}}-\mathbb{E}\left[X^{\theta_0}_{\Psi_{n,\alpha,\eta}}\right]\xrightarrow[n\to \infty]{\mathcal{D}}
\mathcal{N}\left(0,\sigma^2\right)$$
where $$\sigma^2=\frac{2}{(\eta+\overline{\eta})^2}.$$
This completes the proof.

\end{proof}


\section{Proof of Theorems \ref{thm:UniversalityPoissonIntro} and \ref{thm:CLTPoissonIntro}}

In this section we shall prove Theorem \ref{thm:UniversalityPoissonIntro} and its corollary Theorem \ref{thm:CLTPoissonIntro}. These two theorems deal with the linear statistic $\Psi_{n, \gamma, \eta}^{(\theta_0)}$, but we shall extend our discussion here to consider also linear combinations. Namely, for
$\zeta_1,\ldots,\zeta_l\in\mathbb{R}$, $0<\gamma<1$ and $\eta_1,\ldots,\eta_l\in\mathbb{C}$ with $\text{Re}\eta_j>0$ we would like to consider
\begin{equation} \label{eq:LinearCombinations}
\Psi_{n,\Omega}\left(e^{i\theta}\right)=\frac{1}{n^\gamma}\sum_{j=1}^l\zeta_j\text{Re}\left(\frac{e^{i\theta}+\left(1-\frac{\eta_j}{n^\gamma}\right)e^{i\theta_0}}
{e^{i\theta}-\left(1-\frac{\eta_j}{n^\gamma}\right)e^{i\theta_0}}\right)=\sum_{j=1}^l\zeta_j\Psi_{n,\omega_j}
\end{equation}
where $\omega_j=\left(1-\frac{\eta}{n^\gamma}\right)e^{i\theta_0}$ and $\Omega=\left(\omega_1,\ldots,\omega_l\right)$.

First, it is not hard to see that the analysis of Section 3 extends to linear combinations as well so that in the case of $\alpha_n \equiv \alpha \in (-1,1)$,
\beq
X^{\theta_0}_{\sum_{k=1}^l\zeta_k\Psi_{n,\Omega_k}}-\bbE X^{\theta_0}_{\sum_{k=1}^l\zeta_k\Psi_{n,\Omega_k}}\xrightarrow[n\to
\infty]{\mathcal{D}}\mathcal{N}\left(0,\sigma^2_{\eta_1,\ldots,\eta_l}\right).
\eeq
where $$\sigma^2_{\eta_1,\ldots,\eta_l}=\lim_{n\to \infty}\text{Tr}H\left(\sum_{j=1}^n\zeta_j\kappa^{(n,\omega_j)}\right)H\left(
\sum_{j=1}^n\zeta_j\widetilde\kappa^{(n,\omega_j)}\right)=\sum_{1\leq k,j\leq
l}\zeta_j\zeta_k\text{Re}\frac{2}{\left(\eta_k+\overline{\eta_j}\right)^2}.$$

We want to show that under the assumption of Theorem \ref{thm:UniversalityPoissonIntro}, we have for any $m\in\bbN$
\beq\label{eq:conv-in-moments}\left|\bbE\left(X^{\theta_0}_{\Psi_{n,\Omega}}-\bbE
X^{\theta_0}_{\Psi_{n,\Omega}}\right)^m-\bbE_0\left(X^{\theta_0}_{\Psi_{n,\Omega}}-\bbE_0X^{\theta_0}_{\Psi_{n,\Omega}}\right)^m\right|\xrightarrow[n\to\infty]{}0\eeq
where recall that $\mathbb{E}$ (respectively $\mathbb{E}_0$) signify taking expectations with respect to OPE's with the measure $\mu$ (resp.\ $\mu_0$). For this, it suffices to show that 
\begin{equation} \nonumber
\left|C_m^{(n)}\left(X^{\theta_0}_{\Psi_{n,\Omega}}\right)-C^{(n)}_{m,0}\left(X^{\theta_0}_{\Psi_{n,\Omega}}\right)\right|\xrightarrow[n\to\infty]{}0.
\end{equation}

We will work in this section exclusively with the CMV representation of the cumulants. In particular,
we study the operator 
$$\Psi_{n,\omega}\left(\mathcal{C}\right)=\frac{1}{n^\gamma}\text{Re}\left(\mathcal{C}+\omega\right)\left(\mathcal{C}-\omega\right)^{-1}$$
so on the operator theory side, we would like to prove
$$\left|C_m^{(n)}\left(\Psi_{n,\omega}\left(\mathcal{C}\right)\right)-C^{(n)}_{m,0}\left(\Psi_{n,\omega}\left(\mathcal{C}_0\right)\right)\right|\xrightarrow[n\to\infty]{}0.$$
where $\mathcal{C},\mathcal{C}_0$ are the CMV operators associated with $\mu$ and $\mu_0$ respectively.

Our first step in this analysis is the observation that the ``almost banded" structure that we have obtained in Section 3 for
$\frac{1}{n^\gamma}\text{Re}\left(G+\omega\right)\left(G-\omega\right)^{-1}$ (i.e\ Proposition \ref{prop:CombsThomas}) is a known fact for CMV
matrices, that is, by \cite[Theorem 10.14.1]{simonopuc}, we have 
$$\left(\mathcal{C}-z\right)^{-1}_{nm}\leq D\exp\left({-\Delta|n-m|}\right)$$
where $\Delta=\min\left\{\frac{1}{3},\frac{d(z,\partial\mathbb{D})}{6e}\right\}$ and $D>0$ is independent of $n,m,z$.

The five diagonal shape of $\mathcal{C}+z$ implies that
\beq 
\label{eq:comb-tom}\left(\mathcal{C}+z\right)\left(\mathcal{C}-z\right)^{-1}_{nm} \leq D'\exp\left({-\Delta|n-m|}\right)
\eeq 
and in particular $$\left|\Psi_{n,\Omega}\left(\mathcal{C}\right)\right|_{kj}\leq D''\exp\left(-\frac{d|k-j|}{n^\gamma}\right)$$
for $D'',d>0$ independent of $k,j$, where $\Psi_{n,\Omega}\left(\mathcal{C}\right)$ is the obvious linear combination. 

We need the following unit circle analog of \cite[Proposition 3.3]{realope}.
\begin{prop} \label{prop:CMV-perturbation}
Let $\mathcal{C}_0$ and $\mathcal{C}$ be two bounded CMV matrices and fix $\eta\in\bbC$ with $\text{Re}\eta>0$ and $\theta_0 \in (-\pi,\pi)$.
For $1>\gamma>0$ let
\begin{equation*}
    \omega_n=\left(1-\frac{\eta}{n^\gamma}\right)e^{i\theta_0}
\end{equation*}
and let $F_0(\omega_n)=\left(\mathcal{C}_0+\omega_n\right)\left(\mathcal{C}_0-\omega_n\right)^{-1}$ and
$F(\omega_n)=\left(\mathcal{C}+\omega_n\right)\left(\mathcal{C}-\omega_n\right)^{-1}$.
Fix $m$ and assume that there exists $1>\beta>\gamma$ such that
\begin{equation}\label{eq:measure-condition}
\begin{split}
\left(\sum_{j=[n-4mn^\beta]}^{[n+4mn^\beta]}\text{Re} \left( F_0(\omega_n) \right)_{j,j} \right)=\mathcal{O}(n^\beta),
\end{split}
\end{equation}
and
\begin{equation}
    \label{eq:jacobi-perturb-cond2}
\left \| \left(\mathcal{C}-\mathcal{C}_0 \right) P_{\left \{ n-6mn^\beta,n+6mn^\beta \right \} } \right \|_\infty = o \left (n^{-\beta}
\right).
\end{equation}
Then
\begin{equation}
    \label{eq:preconclusion1}
\frac{1}{n^\gamma} \left \|  P_{\left \{ n-2mn^\beta,n+2mn^\beta\right\}} \left( F(\omega_n)-F_0(\omega_n)\right) P_{\left
\{n-2mn^\beta,n+2mn^\beta\right \}} \right \|_1 \rightarrow 0
\end{equation}
as $n \rightarrow \infty$.
\end{prop}
\begin{proof}

Using the following identity
\begin{equation} \label{eq:resolvent-formula}
F(\omega_n) -F_0(\omega_n)=2\omega_n\mathcal{G}(\omega_n)\left(\mathcal{C}_0-\mathcal{C} \right)\mathcal{G}_0(\omega_n)
\end{equation} where \begin{equation*}
    \mathcal{G}(\omega_n)=\left(\mathcal{C}-\omega_n\right)^{-1}, \quad \mathcal{G}_0(\omega_n)=\left(\mathcal{C}_0-\omega_n\right)^{-1}
\end{equation*} we obtain \begin{multline} \label{eq:resolvent-consequence}
\widetilde{P}_n^1  \left(F(\omega_n) -
F_0(\omega_n)\right)\widetilde{P}_n^1  =2\omega_n\widetilde{P}_n^1   \mathcal{G}(\omega_n) \left(\mathcal{C}_0-\mathcal{C} \right) \mathcal{G}_0(\omega_n)
\widetilde{P}_n^1 \\
=2\omega_n\widetilde{P}_n^1\mathcal{G}(\omega_n) \widetilde{P}_n^2 \left(\mathcal{C}_0-\mathcal{C} \right) \widetilde{P}_n^2 \mathcal{G}_0(\omega_n)
\widetilde{P}_n^1 +\mathcal{R}_n.
\end{multline} where again, as in Proposition \ref{prop:GGT2Toeplitz}, we denote 
\begin{equation} \nonumber
\widetilde{P}_n^1=P_{\left\{n-2mn^\beta,n+{2mn}^\beta\right\}}, \quad
\widetilde{P}_n^2=P_{\left\{n-4mn^\beta,n+{4mn}^\beta\right\}}.
\end{equation}

By \cite[Lemma 3.2]{realope}
\beq \label{eq:remainder-estimate}
\left \| \mathcal{R}_n \right \|_1 \leq D_1 n^{D_2} e^{-D_3 n^{\beta-\gamma}}.
\eeq 
for some constants $D_1,D_2,D_3>0$ independent of $n$, so we are left with showing 
\beq \label{eq:onenormresolvent1}
\begin{split}
& \frac{1}{n^\gamma}\left \| \widetilde{P}_n^1 \mathcal{G}(\omega_n) \widetilde{P}_n^2 \left(\mathcal{C}_0-\mathcal{C} \right) \widetilde{P}_n^2
\mathcal{G}_0(\omega_n)\widetilde{P}_n^1 \right \|_1 \rightarrow 0
\end{split}
\eeq
as $n \rightarrow \infty$.

By the resolvent identity we have
\begin{multline}\label{eq:2ndresolventAA}
 \widetilde{P}_n^1 \mathcal{G}(\omega_n) \widetilde{P}_n^2 \left(\mathcal{C}_0-\mathcal{C} \right) \widetilde{P}_n^2\ \mathcal{G}_0(\omega_n)\widetilde{P}_n^1\\=
  \widetilde{P}_n^1 \mathcal{G}_0(\omega_n)(I+\left(\mathcal{C}_0-\mathcal{C} \right) \mathcal{G}(\omega_n)) \widetilde{P}_n^2 \left(\mathcal{C}_0-\mathcal{C}
  \right) \widetilde{P}_n^2 \mathcal{G}_0(\omega_n)\widetilde{P}_n^1
\end{multline}
Therefore, by Cauchy-Schwarz,
\begin{multline}\label{eq:2ndresolventA}
 \left\|\widetilde{P}_n^1 \mathcal{G}(\omega_n) \widetilde{P}_n^2 \left(\mathcal{C}_0-\mathcal{C} \right) \widetilde{P}_n^2\
 \mathcal{G}_0(\omega_n)\widetilde{P}_n^1\right\|_1\\\qquad \leq
 \left\| \widetilde{P}_n^1 \mathcal{G}_0(\omega_n)\right \|_2 \left\|(I+\left(\mathcal{C}_0-\mathcal{C} \right)\mathcal{G}(\omega_n))
 \widetilde{P}_n^2\right\|_\infty \left\| \left(\mathcal{C}_0-\mathcal{C} \right) \widetilde{P}_n^2\right\|_\infty\left\|  \mathcal{G}_0(\omega_n)
 \widetilde{P}_n^1\right\|_2.
\end{multline}

We recall $\widetilde{P}_n^3=P_{\left\{n-6mn^\beta,n+{6mn}^\beta\right\}}$ and write
\beq \no \left(\mathcal{C}_0-\mathcal{C} \right)G(\omega_n) \widetilde{P}_n^2=\left(\mathcal{C}_0-\mathcal{C}
\right)\widetilde{P}_n^3G(\omega_n) \widetilde{P}_n^2+ \widetilde{ \mathcal R_n},
\eeq
 where, again by \cite[Lemma 3.2]{realope}, $$\left \|\widetilde{\mathcal R}_n \right \|_\infty \leq D_1 n^{D_2} e^{-D_3 n^{\beta-\gamma}}.$$  
This implies 
\begin{equation}  \no \left\|\left(\mathcal{C}_0-\mathcal{C} \right)G(\omega_n) \widetilde{P}_n^2\right\|_\infty \leq
\left\|\left(\mathcal{C}_0-\mathcal{C} \right)\widetilde{P}_n^3\right\|_\infty \left\|G(\omega_n)
\widetilde{P}_n^2\right\|_\infty+o(1)=o(1),
\end{equation}
as $n \to \infty$.

Now, since $\mathcal{C},\mathcal{C}_0$ are unitary 
\begin{equation*}
    \left\|\mathcal{G}(\omega_n)\right\|_\infty,\left\|\mathcal{G}_0(\omega_n)\right\|_\infty \leq \frac{1}{d(\omega_n,\partial \mathbb{D})}=\mathcal{O}(n^\gamma)
\end{equation*}
so by condition \eqref{eq:jacobi-perturb-cond2}, we find
\beq \no
\left\|(I+\left(\mathcal{C}_0-\mathcal{C} \right)\mathcal{G}(\omega_n)) \widetilde{P}_n^2\right\|_\infty \leq D
\eeq
for some constant $D>0$. Thus
\begin{multline}\label{eq:2ndresolventAB}
 \left\|\widetilde{P}_n^1 \mathcal{G}(\omega_n) \widetilde{P}_n^2 \left(\mathcal{C}_0-\mathcal{C} \right) \widetilde{P}_n^2\
 \mathcal{G}_0(\omega_n)\widetilde{P}_n^1\right\|_1\\=
D \left\| \widetilde{P}_n^1 \mathcal{G}_0(\omega_n)\right\|_2 \left\| \left(\mathcal{C}_0-\mathcal{C} \right) \widetilde{P}_n^2\right\|_\infty\left\|
\mathcal{G}_0(\omega_n) \widetilde{P}_n^1\right\|_2.
\end{multline}
Moreover, we have that
\beq
\left \|  \widetilde{P}_n^1 \mathcal{G}_0(\omega_n)  \right \|_2^2= \left \|   \mathcal{G}_0(\omega_n)  \widetilde{P}_n^1 \right \|_2^2 =
\sum_{j=[n-2mn^\beta]}^{[n+2mn^\beta]}\sum_{k=1}^\infty  \left| \left( \mathcal{G}_0(\omega_n)  \right)_{k,j} \right|^2.
\eeq

Now note that
\begin{equation} \nonumber
\begin{split}
2\text{Re}\left(F_0(\omega_n)\right)_{jj}&=\left(\left(\mathcal{C}_0+\omega_n \right)\left( \mathcal{C}_0-\omega_n \right)^{-1}+\left(\mathcal{C}_0^{-1}+\overline{\omega_n} \right)\left( \mathcal{C}_0^{-1}-\overline{\omega_n} \right)^{-1}\right)_{jj}\\
&=2\left(\textrm{Id}-\left| \omega_n\right|^2 \right) \left( \left(\mathcal{C}_0^{-1}-\overline{\omega_n} \right)^{-1} \left( \mathcal{C}_0-\omega_n \right)^{-1} \right)_{jj}
\end{split}
\end{equation}
so 
\begin{equation*}
\sum_{k=1}^\infty \left|\mathcal{G}_0(\omega_n)_{jk}\right|^2 =\left(1-\left|\omega_n\right|^2\right)^{-1}\text{Re}\left(F_0(\omega_n)\right)_{jj}.
\end{equation*}
Therefore
$$\left \|  \widetilde{P}_n^1 \mathcal{G}_0(\omega_n)  \right
\|_2^2=\frac{n^\gamma}{2\Re\eta+\frac{|\eta|^2}{n^\gamma}}\sum_{j=[n-2mn^\beta]}^{[n+2mn^\beta]}\Re \left( F_0(\omega_n)  \right)_{j,j}.$$
Putting all the above together we get \beq \label{eq:esssss}
\begin{split}
& \left \| \widetilde{P}_n^1 \mathcal{G}(\omega_n)\widetilde{P}_n^2\left(\mathcal{C}_0-\mathcal{C}
\right)\widetilde{P}_n^2 \mathcal{G}_0(\omega_n)\widetilde{P}_n^1 \right \|_1 \\
&\quad \leq \frac{Dn^\gamma}{2\Re\eta+\frac{|\eta|^2}{n^\gamma}} \left\|    \left(\mathcal{C}_0-\mathcal{C} \right) \widetilde{P}_n^2\right
\|_\infty  \left ( \sum_{j=[n-2mn^\beta]}^{[n+2mn^\beta]}\Re \left( F_0(\omega_n) \right)_{j,j} \right)
\end{split}
\eeq
which, when combined with \eqref{eq:measure-condition} and \eqref{eq:jacobi-perturb-cond2}, implies \eqref{eq:onenormresolvent1}. 
\end{proof}

We shall prove the following generalization of Theorem \ref{thm:UniversalityPoissonIntro} to $X^{\theta_0}_{\Psi_{n,\Omega}}$.

\begin{theorem} \label{thm:UniversalityPoissonGeneral}
Let $\mu$ and $\mu_0$ be two probability measures and denote by $\{\alpha_n\}_{n=1}^\infty$ and $\{\alpha_n^0\}_{n=1}^\infty$ the respective
associated recurrence coefficients. Let $\theta_0 \in (-\pi,\pi)$ be such that there exists a neighborhood $\theta_0 \in I$ on which the
following two conditions are satisfied: \\
(i) $\mu_0$ restricted to $I$ is absolutely continuous with respect to Lebesgue measure and its Radon-Nikodym derivative is bounded there.  \\
(ii) The orthonormal polynomials for $\mu_0$ are uniformly bounded on $I$.\\
Assume further that
\beq \label{eq:decay-rate-condition}
\alpha_n-\alpha_n^0=\mathcal O(n^{-\beta})
\eeq
as $n \to \infty$ for some $1>\beta>0$.

Then for any $0<\gamma<\beta$, $\zeta_1,\ldots,\zeta_l\in\mathbb{R}$, and $\eta_1,\ldots,\eta_l\in\mathbb{C}$ with $\text{Re}\eta_j>0$, 
$$\left|C_m^{(n)}\left(X^{\theta_0}_{\Psi_{n,\Omega}}\right)-C^{(n)}_{m,0}\left(X^{\theta_0}_{\Psi_{n,\Omega}}\right)\right|\xrightarrow[n\to\infty]{}0$$
for any $m\in\bbN$, where $\Psi_{n, \Omega}$ is defined in \eqref{eq:LinearCombinations}.
\end{theorem}

\begin{proof} 
Recall that $\Omega=\left(\omega_1,\ldots,\omega_l\right)$ where $\omega_j=\left(1-\frac{\eta}{n^\gamma}\right)e^{i\theta_0}$.

Fix $m$ and $\beta>\beta'>\gamma$. By Proposition \ref{th:comparison-general}
\beq \no
\begin{split}
& \left|C_m^{(n)}\left(X^{\theta_0}_{\Psi_{n,\Omega}}\right)-C^{(n)}_{m,0}\left(X^{\theta_0}_{\Psi_{n,\Omega}}\right)\right|| \\
&\quad \leq \frac{C(m,\beta')}{n^\gamma} \left \| P_{\left \{n-2mn^{\beta'},n+2mn^{\beta'} \right\}} \left( \sum_{j=1}^l \zeta_j \Re F(\omega_{n,j})-\sum
\zeta_j\Re F_0(\omega_{n,j}) \right) P_{\left \{n-2mn^{\beta'},n+2mn^{\beta'}\right \}} \right \|_1\\
&\qquad+o(1)\\
&\quad \leq \frac{C(m,\beta')}{n^\gamma} \sum_{j=1}^l|\zeta_j| \left \| P_{\left \{n-2mn^{\beta'},n+2mn^{\beta'} \right\}} \left( \Re F(\omega_{n,j})-\Re
F_0(\omega_{n,j}) \right) P_{\left \{n-2mn^{\beta'},n+2mn^{\beta'}\right \}} \right \|_1\\
&\qquad+o(1)\\
\end{split}
\eeq 
We shall show that \eqref{eq:measure-condition} and \eqref{eq:jacobi-perturb-cond2} hold for $\beta'$, thereby proving the theorem by Proposition \ref{prop:CMV-perturbation}.

To prove \eqref{eq:jacobi-perturb-cond2} note that by the conditions of the theorem
\beq \no
n^{\beta'} \left \|  \left(\mathcal{C}-\mathcal{C}_0 \right) P_{\left \{ n-6mn^{\beta'},n+6mn^{\beta'} \right \} } \right \| \leq C(m) n^{\beta'}
\left(n-6mn^{\beta'} \right)^{-\beta} \rightarrow 0
\eeq
as $n \rightarrow \infty$.

To prove \eqref{eq:measure-condition}, let $\left (\varphi_j\right)_{j=0}^\infty$ be the orthonormal polynomials with respect to $\mu_0$, recall the CMV basis introduced in Section 2.1, $\left( \chi_j \right)_{j=0}^\infty$, and note
the following relation (see, e.g., \cite[Section 4]{simonopuc} for a proof)
\begin{align*}
    \chi_{2k}=z^{-k}\varphi^*_{2k}\\
    \chi_{2k-1}=z^{-k+1}\varphi_{2k-1}
\end{align*}
thus
\beq \label{eq:boundednessgfunction1}
\begin{split}
\Re \left( F_0(\lambda) \right)_{j,j}&= \Re \int
|\chi_j(\theta)|^2\left(\frac{e^{i\theta}+\left(1-\frac{\eta}{n^\gamma}\right)e^{i\theta_0}}{e^{i\theta}-\left(1-\frac{\eta}{n^\gamma}\right)e^{i\theta_0}}\right)\textrm{d}
\mu_0(\theta)\\
&=  \int
|\varphi_j(\theta)|^2 \Re \left(\frac{e^{i\theta}+\left(1-\frac{\eta}{n^\gamma}\right)e^{i\theta_0}}{e^{i\theta}-\left(1-\frac{\eta}{n^\gamma}\right)e^{i\theta_0}}\right)\textrm{d}
\mu_0(\theta) \\
&= \int_{(-\pi,\pi) \setminus I}+ \int_{I}.
\end{split}
\eeq
For the first integral in the sum we write
\begin{equation} \nonumber
\begin{split}
&\left| \int_{(-\pi,\pi) \setminus I}|\varphi_j(\theta)|^2 \Re \left(\frac{e^{i\theta}+\left(1-\frac{\eta}{n^\gamma}\right)e^{i\theta_0}}{e^{i\theta}-\left(1-\frac{\eta}{n^\gamma}\right)e^{i\theta_0}}\right)\textrm{d}
\mu_0(\theta) \right| \\
&\quad \leq \left( \frac{\Re \eta}{n^\gamma}+\frac{|\eta|^2}{2n^{2\gamma}} \right)  \int_{(-\pi,\pi) \setminus I} \frac{|\varphi_j(\theta)|^2\textrm{d}\mu_0(\theta)}{(1-\cos(\theta_0-\theta))+\Re\frac{\eta}{n^{\gamma}}\left( e^{i(\theta_0-\theta)}-1 \right)+\frac{|\eta|^2}{2n^{2\gamma}}} \\
&\quad \leq C(I) \left( \frac{\Re \eta}{n^\gamma}+\frac{|\eta|^2}{2n^{2\gamma}} \right)  \int_{(-\pi,\pi) \setminus I} |\varphi_j(\theta)|^2\textrm{d}\mu_0(\theta)\\
&\quad \leq C(I) \left(\frac{\Re \eta}{n^\gamma}+\frac{|\eta|^2}{2n^{2\gamma}} \right)  \int_{(-\pi,\pi)} |\varphi_j(\theta)|^2\textrm{d}\mu_0(\theta)=C(I) \left( \frac{\Re \eta}{n^\gamma}+\frac{|\eta|^2}{2n^{2\gamma}} \right)
\end{split}
\end{equation}
where  $C(I)$ is a constant that depends on $I$, and the last equality follows by the fact that $\phi_j$ is normalized.

For the second integral in \eqref{eq:boundednessgfunction1} we note that by condition $(ii)$ in the theorem
\begin{equation} \nonumber
\begin{split}
&\int_I
|\varphi_j(\theta)|^2 \Re \left(\frac{e^{i\theta}+\left(1-\frac{\eta}{n^\gamma}\right)e^{i\theta_0}}{e^{i\theta}-\left(1-\frac{\eta}{n^\gamma}\right)e^{i\theta_0}}\right)\textrm{d}
\mu_0(\theta)  \\
& \quad \leq C \int_I \Re \left(\frac{e^{i\theta}+\left(1-\frac{\eta}{n^\gamma}\right)e^{i\theta_0}}{e^{i\theta}-\left(1-\frac{\eta}{n^\gamma}\right)e^{i\theta_0}}\right) \chi_I(\theta) \textrm{d}
\mu_0(\theta)
\end{split}
\end{equation}
which is bounded as the real part of the Cayley transform of $\chi_I(\theta)\textrm{d}\mu_0(\theta)$ which by condition $(i)$ of the
theorem, is an absolutely continuous measure with bounded derivative. This proves \eqref{eq:measure-condition} and completes the proof.
\end{proof}

\begin{proof}[Proof of Theorem \ref{thm:CLTPoissonIntro}]
Fix $\alpha \in (-1,1)$. By Proposition \ref{prop:PoissonForConstant}, \eqref{eq:CLTPoisson} and \eqref{eq:VariancePoisson} hold for the OPE associated with $\mu_\alpha$. Now, if $\mu$ has Veblunsky coefficients $(\alpha_n )_{n=0}^\infty$ satisfying $\lim_{n\rightarrow \infty}\alpha_n=\alpha$ then $\supp (\mu)^\circ=\supp(\mu_\alpha)^\circ$. To see this, note that if $\mathcal{C}$ (resp.\ $\mathcal{C}_\alpha$) is the CMV matrix associated with $\mu$ (resp.\ $\mu_\alpha$) then $\mathcal{C}$ is a compact perturbation of $\mathcal{C}_\alpha$ so that $\left( \supp(\mu) \setminus \supp (\mu_\alpha) \right)\cup\left( \supp(\mu_\alpha) \setminus \supp (\mu) \right) $ is a discrete set. 

Now, clearly, for any $\theta_0 \in \supp(\mu_\alpha)^\circ$, there exists a neighborhood $\theta_0 \in I$ such that $(i)$ of Theorem \ref{thm:UniversalityPoissonGeneral} is satisfied. By \cite[equation (2.8)]{Golinskii} $I$ can be chosen so that $(ii)$ of that theorem holds as well. Thus, the result follows from Theorem \ref{thm:UniversalityPoissonGeneral}.
\end{proof}

\begin{remark} \label{Rem:VaryingEta}
For the next section it is useful to note that all the previous results are obtained with virtually no change when replacing $\eta$ by a sequence $(\eta_n)$ satisfying $\eta_n=\eta+O(n^{-\gamma})$ with $\Re \eta>0$. The proofs go through with almost no modification.  
\end{remark}


\section{Extending Theorems \ref{thm:CLTgeneral} and \ref{thm:main-result} for\\ $f\in C_c^1(\partial\bbD)$.}

In this section we want to extend the results of the previous sections to general $C^1$ functions supported on an arc. As remarked in the Introduction, it is not straightforward to scale a function on the circle and we shall do so using the map $\mathcal{M}$ defined in \eqref{eq:Mobius}, $ \mathcal{M}(z)=\frac{i-z}{i+z}$, which will also allow us to exploit the density arguments of \cite[Section 5]{realope}. We shall also extensively use the map $\mathcal{M}^{-1}:\overline{\mathbb{D}}\rightarrow \mathbb{H}_+ \cup\{\infty\}$ given by
\beq \no 
\mathcal{M}^{-1}(\omega)=i\frac{1-\omega}{1+\omega},
\eeq 
(for $\omega=e^{i\theta}\in\partial\mathbb{D}$, $\mathcal{M}^{-1}(e^{i\theta})=\tan(\theta/2)$ but we shall not use this).

The first step is a necessary slight modification to the function $\Psi_{n, \gamma,\eta}$. Let $\widetilde{\eta}\in \bbC$ be such that $\text{Im}\widetilde{\eta}>0$ and let
\beq \no 
\omega_n=\mathcal{M}\left( \frac{\widetilde{\eta}}{n^\gamma} \right)=\left(1+\frac{2i\widetilde{\eta}}{n^\gamma-i\widetilde{\eta}}\right).
\eeq 
Consider the function 
$$g(x)=\text{Im}\frac{1}{x-\widetilde{\eta}}$$ 
and let $h:\partial\bbD\to\bbR$ be defined by $h=g\circ \mathcal{M}^{-1}$ (with $h(-1)=0$) and note that, for $0 \leq \theta_0<2\pi$,
\beq \begin{split}
&\widetilde{h}_n\left(e^{i(\theta-\theta_0)}\right)\\
&\quad \equiv h\left(\mathcal{M}\left(n^\gamma\left(\mathcal{M}^{-1}\left(e^{i\left(\theta-\theta_0\right)}\right)\right)\right)\right)\\
&\quad =\frac{1}{n^\gamma}\text{Im}\frac{1}{\mathcal{M}^{-1}(e^{i(\theta-\theta_0)})-\frac{\widetilde{\eta}}{n^{\gamma}}}=\frac{1}{n^\gamma}\text{Im}\frac{1}{\mathcal{M}^{-1}(e^{i(\theta-\theta_0)})-\mathcal{M}^{-1}(\omega_n)}\\
&\quad =\frac{\left(1+\cos\left(\theta-\theta_0\right)\right)}{2n^\gamma}\Re\left(\frac{e^{i\theta}+\omega_n e^{i\theta_0}}{e^{i{\theta}}-\omega_n e^{i\theta_0}}\right).
    \end{split}
\eeq
Denote
$$\Psi_{\omega_n,n}\left(e^{i\theta}\right)=\frac{1}{n^\gamma}\text{Re}\left(\frac{e^{i\theta}+\omega_n}{e^{i\theta}-\omega_n}\right).$$

By essentially the same argument proving \cite[Lemma 2.2]{realope} we see that for any measure on $\partial \mathbb{D}$ and any two bounded functions $f_1, f_2:\partial \mathbb{D}\rightarrow \mathbb{R}$
\begin{multline} \label{eq:diff-of-momentgen1}
\left|\bbE \left[{\rm e}^{ t \left(X_{f_1,n}-\bbE X_{f_1,n}\right)} \right]-\bbE \left[{\rm e}^{ t
\left(X_{f_2,n}-\bbE X_{f_2,n}\right)} \right]\right|\\
\leq |t|\left( \Var X_{(f_1-f_2),n} \right)^{1/2} {\rm e}^{C |t|^2\left(\Var X_{f_1,n}+\Var X_{(f_1-f_2),n}\right)}
\end{multline}
for any $|t|>0$ sufficiently small. 
In particular, the linear statistics associated with $\widetilde{h}_n$ and $\Psi_{\omega_n,n}$ satisfy, for any $\theta_0$,
\begin{multline} \label{eq:diff-of-momentgen}
\left|\bbE \left[{\rm e}^{ t \left(X^{\theta_0}_{\Psi_{\omega_n,n}}-\bbE X^{\theta_0}_{\Psi_{\omega_n,n}}\right)} \right]-\bbE \left[{\rm e}^{ t
\left(X_{\widetilde{h}_n,n,\gamma}^{\theta_0}-\bbE X_{\widetilde{h}_n,n,\gamma}^{\theta_0}\right)} \right]\right|\\
\leq |t|\left( \Var X^{\theta_0}_{\Psi_{\omega_n,n}-\widetilde{h}_n} \right)^{1/2} {\rm e}^{C |t|^2\left(\Var X^{\theta_0}_{\Psi_{\omega_n,n}}+\Var
X^{\theta_0}_{\Psi_{\omega_n,n}-\widetilde{h}_n}\right)},
\end{multline}
for small enough $|t|$. Thus, to show that the asymptotics of the fluctuations of
$X_{\widetilde{h}_n,n,\gamma}^{\theta_0}$ are the same as those of $X^{\theta_0}_{\Psi_{\omega_n,n}}$ we need to control $\Var X^{\theta_0}_{\Psi_{\omega_n,n}-\widetilde{h}_n}$.

\begin{lemma} \label{lemma:poisso-poisson} 
For any probability measure, $\mu$, on the unit circle, and any $0 \leq \theta_0 <2\pi$, we have
\begin{equation} \label{eq:Var-Difference}
\Var X^{\theta_0}_{\Psi_{n,\omega}-\widetilde{h}_n}\xrightarrow[n\to
\infty]{}0.
\end{equation}
\end{lemma}

\begin{proof}
Recall that
\beq
\label{eq:varianceformula}
\begin{split}
&\Var X^{\theta_0}_{\Psi_{n,\omega}-\widetilde{h}_n}\\
&\quad=\frac{1}{2}\iint_{\partial\mathbb{D}^2}\left|\left(\Psi_{n,\omega}-\widetilde{h}_n\right)\left(\theta-\theta_0\right)-\left(\Psi_{n,\omega}-\widetilde{h}_n\right)\left(\phi-\theta_0\right)\right|^2\left|K_n(\theta,\phi)\right|^2d\mu(\theta
)d\mu(\phi)
\end{split} 
\eeq
where $K_n(\theta,\phi)=\sum_{j=0}^{n-1}\varphi_j\left(e^{i\theta}\right)\overline{\varphi_j\left(e^{i\phi}\right)}$ is the Christoffel-Darboux kernel associated with the measure $\mu$ (recall Section 2.2). Since both $\mu$ and  $\theta_0$ are arbitrary we may assume (by rotating $\mu$) that $\theta_0=0$. Furthermore, the CD formula \cite[Theorem 2.2.7]{simonopuc} says that
\begin{equation} \label{eq:CDFormula}
K_n(\theta,\phi)=\frac{\varphi_n^*(e^{i\theta}) \overline{\varphi_n^*(e^{i\phi})}-\varphi_n(e^{i\theta}) \overline{\varphi_n(e^{i\phi})}}{1-e^{i(\theta-\phi)}}
\end{equation}
so that
\begin{equation} \nonumber
\begin{split}
&\Var X_{\Psi_{n,\omega}-\widetilde{h}_n}\\
&\quad=\iint_{\partial\mathbb{D}^2}\frac{\left|\left(\Psi_{n,\omega}-\widetilde{h}_n\right)\left(\theta\right)-\left(\Psi_{n,\omega}-\widetilde{h}_n\right)\left(\phi\right)\right|^2}{2\left|e^{i\phi}-e^{i\theta}\right|^2}
\left| \varphi_n^*(e^{i\theta}) \overline{\varphi_n^*(e^{i\phi})}-\varphi_n(e^{i\theta}) \overline{\varphi_n(e^{i\phi})} \right|^2d\mu(\theta)d\mu(\phi).
\end{split} 
\end{equation}

Thus, by the fact that $\varphi_n$ and $\varphi_n^*$ are normalized, we see that it suffices to show that 
$$
\lim_{n\rightarrow \infty}\left|\frac{\left(\Psi_{n,\omega}-\widetilde{h}\right)\left(\theta\right)-\left(\Psi_{n,\omega}-\widetilde{h}\right)\left(\phi\right)}{e^{i\phi}-e^{i\theta}}\right|=0
$$
uniformly in $\theta, \phi$. This will of course be implied by
\beq \label{eq:Derivative-Decay}
\lim_{n\rightarrow \infty}\sup_{\theta \in \partial \mathbb{D}}\left|\Psi'_{n,\omega}-\widetilde{h}'\right|(\theta)=0.
\eeq

Write
\beq \no 
\omega_n=r_ne^{i\phi_n}, \qquad  \theta_n=\theta-\phi_n
\eeq
and note that $r_n=1-\Delta_n$ with $c/n^\gamma\leq \Delta_n \leq C/n^\gamma$ and $\phi_n=O\left(\frac{1}{n^\gamma}\right)$. Now, by a straightforward computation 
\beq \nonumber
\begin{split}
\left(\Psi'_{n,\omega}-\widetilde{h}'\right)(\theta)&= A_n(\theta)+B_n(\theta)
\end{split}
\eeq
where
\beq \nonumber
A_n(\theta)=\frac{r_n}{n^\gamma}\frac{(1-\cos\theta)(1-r_n^2)\sin\theta_n}{(2r_n(1-\cos\theta_n)+\Delta_n^2)^2}
\eeq
and
\beq \nonumber
B_n(\theta)=\frac{\sin \theta}{2n^\gamma}\frac{1-r_n^2}{2r_n(1-\cos\theta_n)+\Delta_n^2}.
\eeq

To bound $\|A_n\|_\infty$, let $\frac{2}{3}\gamma<\varepsilon<\frac{3}{4}\gamma$ and note that for $|\theta_n|<\frac{1}{n^\varepsilon}$
\beq \no
|A_n(\theta)|\leq \frac{C}{n^{2\gamma}}\frac{(1-\cos\theta)|\sin\theta_n|}{\Delta_n^4}\leq \frac{C'}{n^{2\gamma+3\varepsilon-4\gamma}}=\frac{C'}{n^{3\varepsilon-2\gamma}}
\eeq
where $C, C'$ are some constants and the second inequality follows since if $|\theta_n|<\frac{1}{n^\varepsilon}$ then $|\theta|<\frac{2}{n^\varepsilon}$ for $n$ large enough (recall $\phi_n=O(n^{-\gamma})$).
For $|\theta_n|\geq \frac{1}{n^\varepsilon}$ write 
\beq \no 
\begin{split}
|A_n(\theta)| & \leq \frac{C}{n^{2\gamma}}\frac{(1-\cos\theta)}{(1-\cos\theta_n)^2}= \frac{C}{n^{2\gamma}}\frac{(1-\cos\theta_n)+O(n^{-\gamma})}{(1-\cos\theta_n)^2} \\
& \leq \frac{C}{n^{2\gamma}}\frac{1}{1-\cos\theta_n}+\frac{C'}{n^{3\gamma}} \frac{1}{(1-\cos\theta_n)^2}\leq \frac{C}{n^{2\gamma-2\varepsilon}}+\frac{C'}{n^{3\gamma-4\varepsilon}}.
\end{split}
\eeq
It follows that $\|A_n\|_\infty \rightarrow 0$ as $n \rightarrow \infty$.

For $B_n$ let $0<\varepsilon<\gamma$ and note that for $|\theta_n|<\frac{1}{n^\varepsilon}$
\beq \no 
|B_n(\theta)|\leq \frac{C}{n^{2\gamma}}\frac{|\sin(\theta_n)|}{\Delta_n^2}\leq \frac{C'}{n^{2\gamma+\varepsilon-2\gamma}}=\frac{C'}{n^\varepsilon}
\eeq
and for $|\theta_n|\geq \frac{1}{n^\varepsilon}$
\beq \no
|B_n(\theta)| \leq \frac{C}{n^{2\gamma}}\frac{1}{(1-\cos\theta_n)}\leq \frac{C'}{n^{2\gamma-2\varepsilon}}
\eeq
so that also $\|B_n\|_\infty \rightarrow 0$ as $n \rightarrow \infty$. This finishes the proof of the lemma.
\end{proof}

It follows by \eqref{eq:diff-of-momentgen} that the conclusions of Theorems \ref{thm:CLTPoissonIntro} and \ref{thm:UniversalityPoissonIntro} hold for $X_{\widetilde{h}_n,n,\gamma}^{\theta_0}$ (and linear combinations) with the limiting variance now equal to  
$$\frac{2}{\left(2i\eta+2\overline{i\eta}\right)^2}=-\frac{1}{2}\text{Re}\left(\frac{1}{\left(\eta-\overline{\eta}\right)^2}\right)=\frac{1}{ 4\pi^2} \iint_{\mathbb{R}^2} \left( \frac{g(x)-g(y)}{x-y} \right)^2 {\rm d} x
{\rm d} y$$
where the last equality follows from residue calculus (see for ex. \cite[equation 4.24]{realope}).

In \cite[Section 5]{realope} a density argument for functions of the type $\sum_j c_j \text{Im}\frac{1}{x-\widetilde{\eta_j}}$, was used to extend mesoscopic universality and CLT results to continuously differentiable compactly supported functions on $\mathbb{R}$. We want to import this argument to our setting. We start by recalling the space $\mathcal{L}_\omega$ (used in \cite{realope, 17 in bd}): we say that $f\in \mathcal{L}_\omega$ if it is a real function which vanishes at infinity and
$\sup_{x,y\in \bbR} \sqrt{1+x^2}\sqrt{1+y^2} \left|\frac{f(x)-f(y)}{x-y}\right| <\infty.$ The space $\mathcal{L}_\omega$ is a normed space with
the weighted Lipschitz norm 
$$\|f\|_{\mathcal L_w}:=\sup_{x,y\in \bbR} \sqrt{1+x^2}\sqrt{1+y^2} \left|\frac{f(x)-f(y)}{x-y}\right| .$$ 
Its importance for our study comes from the fact that this norm controls the variance. That this is true for orthogonal polynomial ensembles on $\mathbb{R}$ was shown in \cite[Proposition 5.1]{realope}. The following is the unit circle analog of that proposition.

\begin{prop}  \label{prop:varlomega1}
Let $\mathfrak{g}(x)=(x-i)^{-1}$ and let $\mathfrak{h} = \mathfrak{g}\circ \mathcal{M}^{-1}$. Then for any $ f\in C_c^1((-\pi,\pi))$ we have for any $0\leq\theta_0< 2\pi$
\beq
\Var X_{\widetilde{f},n,\gamma}^{\theta_0}\leq \left\|f\circ \mathcal{M}\right\|^2_{\mathcal{L}_\omega} \Var
X_{\widetilde{\mathfrak{h}},n,\gamma}^{\theta_0},\eeq 
\end{prop}
\begin{proof}
\begin{equation}
    \label{eq:boundvarianceh}
    \begin{split}
&\Var X_{\widetilde{f},n,\gamma}^{\theta_0}\\
&=\frac{1}{2}\iint_{\partial\mathbb{D}^2}\left|f\circ\mathcal{M}\left(n^\gamma\mathcal{M}^{-1}\left(e^{i\left(\theta-\theta_0\right)}\right)\right)-f\circ\mathcal{M}\left(n^\gamma\mathcal{M}^{-1}\left(e^{i\left(\phi-\theta_0\right)}\right)\right)\right|^2\left|K_n(\theta,\phi)\right|^2d\mu(\theta
)d\mu(\phi)\\
&\leq \frac{1}{2} \left(\sup_{\theta,\phi\in [-\pi,\pi]}
\left|\frac{f\circ\mathcal{M}(n^\gamma\mathcal{M}^{-1}\left(e^{i\left(\theta-\theta_0\right)}\right))-f\circ\mathcal{M}(n^\gamma\mathcal{M}^{-1}\left(e^{i\left(\phi-\theta_0\right)}\right))}{\mathfrak{g}(n^\gamma\mathcal{M}^{-1}\left(e^{i\left(\theta-\theta_0\right)}\right))-\mathfrak{g}(n^\gamma\mathcal{M}^{-1}\left(e^{i\left(\phi-\theta_0\right)}\right))}\right|\right)^2\\
&\quad\times\iint_{\partial\mathbb{D}^2}\left|\mathfrak{g}(n^\gamma\mathcal{M}^{-1}\left(e^{i\left(\theta-\theta_0\right)}\right))-\mathfrak{g}(n^\gamma\mathcal{M}^{-1}\left(e^{i\left(\phi-\theta_0\right)}\right))\right|^2\left|K_n(\theta,\phi)\right|^2d\mu(\theta
)d\mu(\phi)\\
&=\quad \left(\sup_{\theta,\phi\in [-\pi,\pi]}
\left|\frac{f\circ\mathcal{M}(n^\gamma\mathcal{M}^{-1}\left(e^{i\left(\theta-\theta_0\right)}\right))-f\circ\mathcal{M}(n^\gamma\mathcal{M}^{-1}\left(e^{i\left(\phi-\theta_0\right)}\right))}{\mathfrak{g}(n^\gamma\mathcal{M}^{-1}\left(e^{i\left(\theta-\theta_0\right)}\right))-\mathfrak{g}(n^\gamma\mathcal{M}^{-1}\left(e^{i\left(\phi-\theta_0\right)}\right))}\right|\right)^2\Var X_{\widetilde{\mathfrak{h}},n,\gamma}^{\theta_0}.
    \end{split}
\end{equation}
We know that $\mathcal{M}^{-1}$ is onto $\mathbb{R}$, thus we can write
\begin{equation*}\begin{split}
    &\sup_{\theta,\phi\in [-\pi,\pi]}
    \left|\frac{f\circ\mathcal{M}(n^\gamma\mathcal{M}^{-1}\left(e^{i\left(\theta-\theta_0\right)}\right))-f\circ\mathcal{M}(n^\gamma\mathcal{M}^{-1}\left(e^{i\left(\phi-\theta_0\right)}\right))}{\mathfrak{g}(n^\gamma\mathcal{M}^{-1}\left(e^{i\left(\theta-\theta_0\right)}\right))-\mathfrak{g}(n^\gamma\mathcal{M}^{-1}\left(e^{i\left(\phi-\theta_0\right)}\right))}\right|\\&\quad=\sup_{x,y\in
    \mathbb{R}} \left|\frac{f\circ\mathcal{M}(x)-f\circ\mathcal{M}(y)}{\mathfrak{g}(x)-\mathfrak{g}(y)}\right|=\sup_{x,y\in \mathbb{R}}
    \left|\frac{f\circ\mathcal{M}(x)-f\circ\mathcal{M}(y)}{(x-i)^{-1}-(y-i)^{-1}}\right|\\&\quad=\sup_{x,y\in \mathbb{R}}
    \left|\frac{f\circ\mathcal{M}(x)-f\circ\mathcal{M}(y)}{x-y}\right||x-i||y-i|=\sup_{x,y\in \mathbb{R}}
    \left|\frac{f\circ\mathcal{M}(x)-f\circ\mathcal{M}(y)}{x-y}\right|\sqrt{1+x^2}\sqrt{1+y^2}\\&\quad=\left\|f\circ
    \mathcal{M}\right\|_{\mathcal{L}_\omega}
    \end{split}
\end{equation*}
hence the statement follows.
\end{proof}

The preceding proposition shows why the following lemma is useful for us.

\begin{lemma} \cite[Lemma 5.3]{realope} \label{lem:density1}
Let $ f\in C_c^1(\bbR)$ $($the space of continuously differentiable function with compact support in $\mathbb{R})$. For any $\eps>0$, there exists $N\in
\bbN, c_1,\ldots,c_N\in\bbR$ and $\widetilde{\eta_1},\ldots,\widetilde{\eta_N}\in \bbH_+= \{\Im \eta>0\}$ such that
$$\left \|f(x)-\sum_{j=1}^N c_j \Im \frac{1}{x-\widetilde{\eta_j}} \right \|_{\mathcal L_w} <\eps.$$
\end{lemma}

Since $\mathcal{M}$ induces a bijection between $C_c^1(\mathbb{R})$ and $C_c^1(-\pi,\pi)$ (viewed as continuously differentiable functions with compact support in $\partial \mathbb{D}\setminus \{-1\}$), this lemma together with \eqref{eq:diff-of-momentgen} and Proposition \ref{prop:varlomega1} almost allow us to extend our results from Sections 3 and 4 to $C_c^1(-\pi,\pi)$. The missing ingredient is boundedness of $\Var X_{\widetilde{\mathfrak{h}},n,\gamma}^{\theta_0}$.  This is the content of the next

\begin{prop} \label{prop:BoundedVariance}
 Let $\mu_0$ be a probability measure on the unit circle. Let $0<\gamma<1$ and $\theta_0 \in \text{supp}(\mu_0)$ be such that there exists a neighborhood $\theta_0 \in I$ on which the following two conditions are
 satisfied: \\
(i) $\mu_0$ restricted to $I$ is absolutely continuous with respect to Lebesgue measure and its Radon-Nikodym derivative is bounded there.  \\
(ii) The orthonormal polynomials for $\mu_0$ are uniformly bounded on $I$. \\
Finally, as above, let $\mathfrak{g}(x)=(x-i)^{-1}$ and let $\mathfrak{h} = \mathfrak{g}\circ \mathcal{M}^{-1}$. Then $\Var X_{\widetilde{\mathfrak{h}},n,\gamma}^{\theta_0} $ is bounded.
\end{prop}
\begin{proof}
We assume without loss of generality that $\theta_0=0$. We note that
\begin{equation}
    \begin{split}
     &\left|\mathfrak{g}(n^\gamma\mathcal{M}^{-1}(e^{i\theta}))-\mathfrak{g}(n^\gamma\mathcal{M}^{-1}(e^{i\phi}))\right|^2\\
&= \left({\mathcal{M}}^{-1}(e^{i\theta})-{\mathcal{M}}^{-1}(e^{i\phi})\right)^2\Im  \frac{1}{{M}^{-1}(e^{i\theta})-{\rm i} /n^\gamma} \Im
\frac{1}{{M}^{-1}(e^{i\phi})-{\rm i} /n^\gamma}\\
&=\left(\frac{\sin\left(\theta\right)}{1+\cos\left(\theta\right)}-\frac{\sin\left(\phi\right)}{1+\cos\left(\phi\right)}\right)^2\left(1+\cos\left(\theta\right)\right)\left(1+\cos\left(\phi\right)\right)\\
&\quad \times
\text{Re}\left(\frac{e^{i\theta}+\left(1-\frac{2}{n^\gamma+1}\right)}{e^{i\theta}-\left(1-\frac{2}{n^\gamma+1}\right)}\right)\text{Re}\left(\frac{e^{i\phi}+\left(1-\frac{2}{n^\gamma+1}\right)}{e^{i\phi}-\left(1-\frac{2}{n^\gamma+1}\right)}\right)
\end{split}
\end{equation}
Using \eqref{eq:varianceformula} and \eqref{eq:CDFormula} we have
\beq
\begin{split}
&\Var X_{\widetilde{\mathfrak{h}},n,\gamma}^{\theta_0} \\
&=\iint_{\partial\mathbb{D}^2}\left(\frac{\sin\left(\theta\right)}{1+\cos\left(\theta\right)}-\frac{\sin\left(\phi\right)}{1+\cos\left(\phi\right)}\right)^2\frac{\left(1+\cos\left(\theta\right)\right)\left(1+\cos\left(\phi\right)\right)}{\left|e^{i\theta}-e^{i\phi}\right|^2}\\
&\quad \times
\text{Re}\left(\frac{e^{i\theta}+\left(1-\frac{2}{n^\gamma+1}\right)}{e^{i\theta}-\left(1-\frac{2}{n^\gamma+1}\right)}\right)\text{Re}\left(\frac{e^{i\phi}+\left(1-\frac{2}{n^\gamma+1}\right)}{e^{i\phi}-\left(1-\frac{2}{n^\gamma+1}\right)}\right)\\
&\quad \times
\left|\overline{\varphi_{n+1}^*(\theta)}\varphi_{n+1}^*(\phi)-\overline{\varphi_{n+1}(\theta)}\varphi_{n+1}(\phi)\right|^2d\mu_0(\theta
)d\mu_0(\phi) \\
&=\iint_{\partial\mathbb{D}^2}\text{Re}\left(\frac{e^{i\theta}+\left(1-\frac{2}{n^\gamma+1}\right)}{e^{i\theta}-\left(1-\frac{2}{n^\gamma+1}\right)}\right)\text{Re}\left(\frac{e^{i\phi}+\left(1-\frac{2}{n^\gamma+1}\right)}{e^{i\phi}-\left(1-\frac{2}{n^\gamma+1}\right)}\right)\\
&\quad \times
\left|\overline{\varphi_{n+1}^*(\theta)}\varphi_{n+1}^*(\phi)-\overline{\varphi_{n+1}(\theta)}\varphi_{n+1}(\phi)\right|^2d\mu_0(\theta
)d\mu_0(\phi) 
\end{split}
\eeq
since
$$\left(\frac{\sin\theta }{1+\cos \theta }-\frac{\sin\phi}{1+\cos\phi}\right)^2\frac{\left(1+\cos\theta \right)\left(1+\cos \phi \right)}{\left|e^{i\theta}-e^{i\phi}\right|^2}=1.$$

By expanding the square and writing $\omega'_n=\omega'_n(\gamma)=\left(1-\frac{2}{n^\gamma+1}\right)$ we get 
\beq \label{eq:VarianceBound1}
\begin{split}
&\iint_{\partial\mathbb{D}^2}\text{Re}\left(\frac{e^{i\theta}+\omega'_n}{e^{i\theta}-\omega'_n}\right)\text{Re}\left(\frac{e^{i\phi}+\omega'_n}{e^{i\phi}-\omega'_n}\right)\left|\overline{\varphi_{n+1}^*(\theta)}\varphi_{n+1}^*(\phi)-\overline{\varphi_{n+1}(\theta)}\varphi_{n+1}(\phi)\right|^2d\mu_0(\theta
)d\mu_0(\phi)\\
&=2\int_{\partial\mathbb{D}}\text{Re}\left(\frac{e^{i\theta}+\omega'_n}{e^{i\theta}-\omega'_n}\right)\left|\varphi_{n+1}(\theta)\right|^2d\mu_0(\theta)
\int_{\partial\mathbb{D}}\text{Re}\left(\frac{e^{i\phi}+\omega'_n}{e^{i\phi}-\omega'_n}\right)\left|\varphi_{n+1}(\phi)\right|^2d\mu_0(\phi)\\
&\quad- \int_{\partial\mathbb{D}}\text{Re}\left(\frac{e^{i\theta}+\omega'_n}{e^{i\theta}-\omega'_n}\right)\overline{\varphi_{n+1}(\theta)}\varphi_{n+1}^*(\theta)d\mu_0(\theta)
 \int_{\partial\mathbb{D}}\text{Re}\left(\frac{e^{i\phi}+\omega'_n}{e^{i\phi}-\omega'_n}\right)\varphi_{n+1}(\phi)\overline{\varphi_{n+1}^*(\phi)}d\mu_0(\phi)\\
&\quad-\int_{\partial\mathbb{D}}\text{Re}\left(\frac{e^{i\theta}+\omega'_n}{e^{i\theta}-\omega'_n}\right)\overline{\varphi_{n+1}^*(\theta)}\varphi_{n+1}(\theta)d\mu_0(\theta)
\int_{\partial\mathbb{D}}\text{Re}\left(\frac{e^{i\phi}+\omega'_n}{e^{i\phi}-\omega'_n}\right)\varphi_{n+1}^*(\phi)\overline{\varphi_{n+1}(\phi)}d\mu_0(\phi)\\
&=2\left\langle\varphi_{n+1},\text{Re}\left(\frac{e^{i\phi}+\omega'_n}{e^{i\phi}-\omega'_n}\right)
\varphi_{n+1}\right\rangle^2-2\left|\left\langle\varphi_{n+1},\text{Re}\left(\frac{e^{i\phi}+\omega'_n}{e^{i\phi}-\omega'_n}\right)
\varphi^*_{n+1}\right\rangle\right|^2.
\end{split}
\eeq
We use the following identities which hold on $\partial\mathbb{D}$ (see \cite[(2.18)(2.19),(2.20)]{cmv5years})
\beq  \label{eq:cmvbasisopuc}
\begin{split}
   &\qquad\varphi_{2k}(\theta) =e^{ik\theta}\overline{\chi_{2k}(\theta)}\qquad  \varphi_{2k-1}(\theta)=e^{i(k-1)\theta}\chi_{2k-1}(\theta)\\
   &\qquad\varphi_{2k}^*(\theta) =e^{ik\theta}{\chi_{2k}(\theta)}\qquad  \varphi_{2k-1}^*(\theta)=e^{ik\theta}\overline{\chi_{2k-1}(\theta)}
\end{split}
\eeq
to get
\beq\begin{split}
&\left\langle\varphi_{n+1},\text{Re}\left(\frac{e^{i\phi}+\omega'_n}{e^{i\phi}-\omega'_n}\right)
\varphi_{n+1}\right\rangle^2=\left\langle
\chi_{n+1},\text{Re}\left(\frac{e^{i\phi}+\omega'_n}{e^{i\phi}-\omega'_n}\right)
\chi_{n+1}\right\rangle^2\\&=\left(\left[\text{Re}\:F_0(\omega'_n)\right]_{n+1,n+1}\right)^2
\end{split}\eeq
where $F_0(\omega'_n)=\left( \mathcal{C}_0+\omega'_n \right)\left(\mathcal{C}_0-\omega'_n \right)^{-1}$ and $\mathcal{C}_0$ is CMV matrix for $\mu_0$. By the same argument as the one in the proof of Theorem \ref{thm:UniversalityPoissonGeneral} (showing that \eqref{eq:boundednessgfunction1} is bounded) we thus see that the first summand in \eqref{eq:VarianceBound1} is bounded.

In order to treat the second summand we employ the following two identities
$$\chi_{2k}=-\alpha_{2k-1}\overline{\chi_{2k}}+\rho_{2k-1}\overline{\chi_{2k-1}}\quad
\chi_{2k-1}=\overline{\alpha_{2k}}\overline{\chi_{2k}}+{\rho_{2k}}\overline{\chi_{2k-1}}$$ 
which hold on the unit circle. They imply for $n=2k$:
\beq
\begin{split}
&\left\langle\varphi_{n},\text{Re}\left(\frac{e^{i\phi}+\omega'_n}{e^{i\phi}-\omega'_n}\right)
\varphi^*_{n}\right\rangle=\left\langle\overline{\chi_{2k}},\text{Re}\left(\frac{e^{i\phi}+\omega'_n}{e^{i\phi}-\omega'_n}\right)
\chi_{2k}\right\rangle\\
&=-\overline{\alpha_{2k-1}}\left\langle{\chi_{2k}},\text{Re}\left(\frac{e^{i\phi}+\omega'_n}{e^{i\phi}-\omega'_n}\right)
\chi_{2k}\right\rangle+\rho_{2k-1}\left\langle{\chi_{2k-1}},\text{Re}\left(\frac{e^{i\phi}+\omega'_n}{e^{i\phi}-\omega'_n}\right)
\chi_{2k}\right\rangle\\
&=-\overline{\alpha_{2k-1}}\left(\left[\text{Re}\:F_0(\omega'_n)\right]_{n,n}\right)+\rho_{2k-1}\left(\left[\text{Re}\:F_0(\omega'_n)\right]_{n-1,n}\right),
\end{split}
\eeq
and for $n=2k-1$:
\beq\begin{split}
   &\left\langle\varphi_{n},\text{Re}\left(\frac{e^{i\phi}+\omega'_n}{e^{i\phi}-\omega'_n}\right)
   \varphi^*_{n}\right\rangle=
   \left\langle{\chi_{2k-1}},\text{Re}\left(\frac{e^{i\phi}+\omega'_n}{e^{i\phi}-\omega'_n}\right)
   e^{i\phi}\overline{\chi_{2k-1}}\right\rangle\\
   &=\alpha_{2k}\left\langle{\chi_{2k-1}},\text{Re}\left(\frac{e^{i\phi}+\omega'_n}{e^{i\phi}-\omega'_n}\right)
   e^{i\phi}{\chi_{2k}}\right\rangle+\rho_{2k}\left\langle{\chi_{2k-1}},\text{Re}\left(\frac{e^{i\phi}+\omega'_n}{e^{i\phi}-\omega'_n}\right)
   e^{i\phi}{\chi_{2k-1}}\right\rangle\\
   &={\alpha_{2k}}\left(\left[\mathcal{C}_0\text{Re}\:F_0(\omega'_n)\right]_{n,n+1}\right)+\rho_{2k}\left(\left[\mathcal{C}_0\text{Re}\:F_0(\omega'_n)\right]_{n,n}\right).
\end{split}
\eeq

We recall that $\mathcal{C}_0$ has a 5-diagonal structure, thus $\left[\mathcal{C}_0\text{Re}\:F_0(\omega'_n)\right]_{n,n+1}$ and
$\left[\mathcal{C}_0\text{Re}\:F_0(\omega'_n)\right]_{n,n}$ are linear combinations of elements from
$$\left\{\left[\text{Re}\:F_0(\omega'_n)\right]_{n-1,n+1},\ldots,\left[\text{Re}\:F_0(\omega'_n)\right]_{n+3,n+1}\right\}$$ and
$$\left\{\left[\text{Re}\:F_0(\omega'_n)\right]_{n-2,n},\ldots,\left[\text{Re}\:F_0(\omega'_n)\right]_{n+2,n}\right\}$$ respectively, with
coefficients bounded by 1. The conditions of the theorem imply, by repeating the argument as in the proof of Theorem \ref{thm:UniversalityPoissonGeneral} that indeed the second summand is bounded as well, which proves the proposition.
\end{proof}

\begin{prop} \label{cor:contvar}
Let $\mu,\mu_0$ be measures satisfying the assumptions of \ref{thm:UniversalityPoissonIntro} for some $\theta_0 \in [0,2\pi)$ and $0<\beta<1$. 
Let, as above, $\mathfrak{g}(x)=\frac{1}{x-i}$ and $\mathfrak{h}=\mathfrak{g}\circ\mathcal{M}^{-1}$. Then for any $0<\gamma<\beta$, $\Var X_{\widetilde{\mathfrak{h}},n,\gamma}^{\theta_0}$ (with respect to $\mu$) is bounded.
Moreover, there exists $c(\mu)>0$ such that for any $f \in C_c^1((-\pi,\pi))$, 
\beq \label{eq:BoundedVarianceCor}
\Var X_{\widetilde{f},n,\gamma}^{\theta_0}\leq c(\mu) \left\|f\circ
\mathcal{M}\right\|^2_{\mathcal{L}_\omega}.
\eeq
\end{prop}

\begin{proof}
Let $\mathcal{C}_0$ and $\mathcal{C}$ be the CMV matrices corresponding to $\mu_0$ and $\mu$ respectively. For $\omega'_n=\left(1-\frac{2}{n^\gamma+1}\right)e^{i\theta_0}$, let as above $F_0(\omega'_n)=\left( \mathcal{C}_0+\omega'_n \right)\left(\mathcal{C}_0-\omega'_n \right)^{-1}$ and $F(\omega'_n)=\left( \mathcal{C}+\omega'_n \right)\left(\mathcal{C}-\omega'_n \right)^{-1}$.

As shown in the proof of Proposition \ref{prop:BoundedVariance}, $\Var X_{\widetilde{\mathfrak{h}},n,\gamma}^{\theta_0}$ is a linear combination of elements from
$$\left\{\left[{\Psi}_{n,\alpha,2}(\mathcal{C})\right]_{n-2,n},\ldots,\left[{\Psi}_{n,\alpha,2}(\mathcal{C})\right]_{n+2,n},\left[{\Psi}_{n,\alpha,2}(\mathcal{C})\right]_{n-1,n+1},\ldots,\left[{\Psi}_{n,\alpha,2}(\mathcal{C})\right]_{n+3,n+1}\right\}$$ with coefficients from the fixed size block around the $(n,n)$ entry of the corresponding CMV matrix (where $\Psi_n(e^{i\theta})=\text{Re}\left(\frac{e^{i\theta}+\omega'_n}{e^{i\theta}-\omega'_n}\right)$). Thus, since we know that $\alpha_n-\alpha_n^{0}=O(n^{-\beta})$, in order to prove the boundedness of the variance it suffices to show that 
$$\left[\text{Re}\:F_0(\omega'))\right]_{n+j,n+k}-\left[\text{Re}\:F(\omega')\right]_{n+j,n+k}\to 0$$ as $n\to \infty$ for any $j,k \in \mathbb{Z}$. But the proof of this is a straightforward modification of \cite[Proposition 3.5]{realope}, with $\mathcal{C}$ instead of $J$ and \eqref{eq:resolvent-formula} playing the role of the resolvent formula.

Now \eqref{eq:BoundedVarianceCor} follows from Proposition \ref{prop:varlomega1}.
\end{proof}

We are ready to prove Theorems \ref{thm:main-result} and \ref{thm:CLTgeneral}.
\begin{proof}[Proof of Theorems  \ref{thm:main-result} and \ref{thm:CLTgeneral}]
The argument is the same as the one given in \cite[Subsection 5.2]{realope}. Having set up the ingredients for the case of $\partial \mathbb{D}$ we see no need to repeat it. We just give a short sketch. 

The idea behind the proof of Theorem \ref{thm:main-result} is to replace $f\in C_c^1((\pi,\pi))$ with $g(x)=\sum_{j=1}^k c_j \Im \frac{1}{x-\widetilde{\eta}_j}$ for an appropriate choice of $c_j$'s and $\widetilde{\eta}_j$'s, leading to an $\varepsilon$ error in the $\mathcal{L}_\omega$ norm, by Lemma \ref{lem:density1}. Since this norm controls the variance, by \eqref{eq:diff-of-momentgen1} together with Proposition \ref{cor:contvar} it is now possible to control  
$$
 \left|\bbE \left[{\rm e}^{ t \left(X_{\widetilde{f},n,\gamma}^{\theta_0} - \bbE X_{\widetilde{f},n,\gamma}^{\theta_0} \right)} \right]-\bbE_0 \left[{\rm
e}^{ t \left( X_{\widetilde{f},n,\gamma}^{\theta_0} - \bbE_0 X_{\widetilde{f},n,\gamma}^{\theta_0}\right)} \right]\right|
$$
by
$$
\left|\bbE \left[{\rm e}^{ t\left(X_{\widetilde{g\circ\mathcal{M}^{-1}},n,\gamma}^{\theta_0} - \bbE X_{\widetilde{g\circ\mathcal{M}^{-1}},n,\gamma}^{\theta_0}\right)}
\right]-\bbE_0 \left[{\rm e}^{ t \left(X_{\widetilde{g\circ\mathcal{M}^{-1}},n,\gamma}^{\theta_0} - \bbE_0
X_{\widetilde{g\circ\mathcal{M}^{-1}},n,\gamma}^{\theta_0}\right)} \right]\right|+2\varepsilon.
$$
But we know this goes to zero by the discussion after the proof of Lemma \ref{lemma:poisso-poisson}.

The proof of Theorem  \ref{thm:CLTgeneral} uses the same idea, together with \cite[Lemma 2.1]{realope} (which says that the cumulants are controlled by the variance), and the fact that
\beq \no 
\iint_{\mathbb{R}^2}\left( \frac{f(x)-f(y)}{x-y} \right)^2 dxdy \leq \| f\|_{\mathcal{L}_\omega}.
\eeq 
\end{proof}

\end{document}